%% file: arxiv.tex
\documentclass[12pt]{article}

\usepackage{comment}
\usepackage[utf8]{inputenc}
\usepackage[margin=1in]{geometry}
\usepackage{amsmath}
\usepackage{multirow}
\usepackage{amssymb}
\usepackage{indentfirst}
\usepackage{amsfonts}
\newcommand\numberthis{\addtocounter{equation}{1}\tag{\theequation}}
\usepackage{amsthm}
\usepackage{scalerel}
\usepackage{natbib}
\usepackage{subcaption}
\usepackage[font=small, labelfont=bf]{caption}
\usepackage[colorlinks=true, allcolors=blue]{hyperref}
\usepackage[dvipsnames,svgnames]{xcolor}
\usepackage[english]{babel}
\usepackage[mathscr]{euscript}

\usepackage{booktabs} 

\usepackage{fancyvrb} 

\usepackage{siunitx}

\usepackage{xpatch}
\xapptocmd\normalsize{%
 \abovedisplayskip=4pt plus 2pt minus 2pt
 \belowdisplayskip=4pt plus 2pt minus 2pt
}{}{}

\usepackage{colortbl}

\definecolor{NatureLightGreen}{RGB}{212, 224, 192}
\definecolor{NatureLightRed}{RGB}{240, 201, 196}

\definecolor{NatureMidRed}{RGB}{204, 97, 92}

\usepackage{thmtools}
\usepackage{mathtools}

\usepackage[capitalize, nameinlink, sort]{cleveref}
\crefname{thm}{Theorem}{Theorems}
\crefname{cor}{Corollary}{Corollary}
\crefname{lem}{Lemma}{Lemmas}
\crefname{asu}{Assumption}{Assumptions}
\crefname{rmk}{Remark}{Remarks}
\crefname{defn}{Definition}{Definitions}
\crefname{thmlisti}{Theorem}{Theorems}
\crefname{asulisti}{Assumption}{Assumptions}

\usepackage{enumitem}

\usepackage{letltxmacro}

\newlist{thmlist}{enumerate}{1}
\setlist[thmlist]{label=(\roman*), ref=\thethm\,(\roman*)}

\newlist{asulist}{enumerate}{1}
\setlist[asulist]{label=(\roman*), ref=\theasu\,(\roman*)}

\usepackage{thmtools}
\declaretheorem[
    name=Theorem,
    Refname={Theorem,Theorems}]{thm}
\declaretheorem[
    name=Lemma,
    Refname={Lemma,Lemmas},
    sibling=thm]{lem}

\declaretheorem[
    name=Assumption,
    numberwithin=subsection,
    Refname={Assumption,Assumptions}]{asu}
\declaretheorem[
    name=Remark,
    Refname={Remark,Remarks}]{rmk}
\declaretheorem[
    name=Definition,
    Refname={Definition,Definitions}]{defn}
\declaretheorem[
    name=Example,
    Refname={Example,Examples}]{exam}

\newcommand{\prob}{\text{Pr}}

\newcommand{\logit}{\text{logit}}
\newcommand{\expit}{\text{expit}}
\newcommand{\ba}{\bar{a}}
\newcommand{\bA}{\bar{A}}

\newcommand{\bX}{\bar{X}}

\newcommand{\USR}{U^\mathrm{SR}}
\newcommand{\UGR}{U^\mathrm{GR}}

\newcommand{\PP}{\mathbb{P}}
\newcommand{\EE}{\mathbb{E}}
\newcommand{\RR}{\mathbb{R}}

\newcommand{\pto}{\stackrel{p}{\to}}
\newcommand{\dto}{\stackrel{d}{\to}}

\newcommand{\cee}{\mathrm{CEE}}

\newcommand{\bp}{{\boldsymbol{p}}}
\newcommand{\bpi}{{\boldsymbol{\pi}}}
\newcommand\Tau{\scalerel*{\tau}{T}}

\newcommand{\sumt}{\sum_{t=1}^T}

\usepackage[ruled]{algorithm2e}

\newcommand{\lessrate}{\stackrel{<}{\sim}}

\def\spacingset#1{\renewcommand{\baselinestretch}%
{#1}\small\normalsize} \spacingset{1}

\begin{document}

\title{\bf Robust Estimation of Moderated Causal Excursion Odds Ratio in Micro-Randomized Trials \vspace*{0.5cm}}
   \date{\vspace{-2ex}} 
  \author{Jiaxin Yu \\ 
   Department of Statistics, University of California, Irvine, \\ Bren Hall 2019
Irvine, CA 92697-1250, \\ jiaxiny4@uci.edu    \vspace*{0.3cm}\\ 
   Tianchen Qian \\
   Department of Statistics, University of California, Irvine, \\ Bren Hall 2019
Irvine, CA 92697-1250, \\ t.qian@uci.edu}

\maketitle
\begin{abstract}
Micro-randomized trials (MRTs) have become increasingly popular for developing and evaluating mobile health interventions that promote healthy behaviors and manage chronic conditions. The recently proposed causal excursion effects have become the standard measure for interventions' marginal and moderated effect in MRTs. Existing methods for MRTs with binary outcomes focus on causal excursion effects on the relative risk scale. However, a causal excursion effect on the odds ratio scale is attractive for its interpretability and valid predicted probabilities, making it a valuable supplement to causal excursion relative risk. In this paper, we propose two novel estimators for the moderated causal excursion odds ratio for MRTs with longitudinal binary outcomes. When the prespecified moderator fully captures the way interventions are sequentially randomized, we propose a doubly robust estimator that remains consistent if either of two nuisance models is correctly specified. For more general settings in which treatment randomization depends on variables beyond the chosen moderator, we propose a general estimator that incorporates an association nuisance model. We further establish the general estimator's robustness to the misspecification of the association nuisance model under no causal effect, and extend the general estimator to accommodate any link functions. We establish the consistency and asymptotic normality of both estimators and demonstrate their performance through simulation studies. We apply the methods to Drink Less, a 30-day MRT for developing mobile health interventions to help reduce alcohol consumption, where the proximal outcome is whether the user opens the app in the hour following the notification.
\end{abstract}

\noindent%
\vfill

\newpage
\spacingset{1.9} 

\tableofcontents
\clearpage

\input{manuscript}

\clearpage
\appendix
\part*{Supplementary Material}
\addcontentsline{toc}{part}{Appendix}
\input{appendix}

\bibliographystyle{agsm}

\bibliography{mhealth-ref}

\end{document}

%% file: manuscript.tex
\section{Introduction}
\label{sec:introduction}

Mobile health (mHealth) interventions, such as push notifications delivered through smartphones, are used in a variety of domains to encourage healthy behavior change \citep{klasnja2019efficacy,necamp2020assessing, bell2023notifications}. Micro-randomized trials (MRTs) are commonly employed to develop and evaluate these interventions. In an MRT, individuals are repeatedly randomized among intervention options at hundreds or thousands of decision points \citep{klasnja2015, liao2016sample}. Causal excursion effects (CEE) are standard estimands in MRT primary and secondary analyses \citep{boruvka2018, qian2021estimating}. These effects contrast the longitudinal proximal outcome under two excursion policies that deviate from the sequential randomization policy used in the MRT, and are used to assess mHealth intervention effects and time-varying effect moderation.

Many MRTs involve binary proximal outcomes measuring user engagement, adherence, or response to prompts \citep{nahum2021translating,bell2023notifications,hurley2025fluctuations}. For such outcomes, existing methods typically model the CEE on the causal relative risk scale \citep{qian2021estimating,shi2023meta,liu2024incorporating}. Although relative risk is advantageous for its collapsibility, it does not guarantee that the estimated probabilities stay within $[0,1]$, potentially leading to numerical instability \citep{dukes2018note}. In contrast, causal odds ratio avoids this issue and offers additional advantages, such as symmetry in outcome definition \citep{sheps1958shall}. Moreover, it is well known that the presence or direction of effect moderation can depend on the chosen scale \citep{brumback2008effect}. These points make the causal odds ratio an attractive alternative and a valuable complement to the relative risk for MRTs with binary proximal outcomes. However, causal models with a logit link are technically more challenging, as no conventional doubly robust estimators exist that allow for misspecification of either the outcome or propensity score model \citep{tchetgen2010doubly,vansteelandt2014structural}. To address this critical gap, we develop methodology that enables robust estimation for causal excursion odds ratios.


In each MRT analysis, researchers first specify the effect moderators of interest, $S_t$, which determines whether the estimated CEE is marginal or conditional. For example, setting $S_t = \emptyset$ yields a fully marginal CEE, while $S_t = \text{location}_t$ or $S_t = \text{day}_t$ captures effect modification by location or study day. We develop estimators for causal odds ratio under two settings, distinguished by the complexity of the MRT randomization policy relative to the analysis-specific $S_t$. The first setting, ``Simple Randomization'', refers to cases where the randomization probability depends on at most $S_t$. For example, MRTs with a constant randomization probability always satisfy Simple Randomization for any $S_t$, including $S_t = \emptyset$. The second setting, ``General Randomization'', allows the randomization probability to possibly depend on history beyond $S_t$. This setting encompasses most MRTs, except certain recent designs where an individual's randomization probability depends on others' outcomes \citep{trella2025deployed}.

For the Simple Randomization setting, we propose a doubly robust estimator for the moderated CEE on the odds ratio scale, motivated by the logistic partially linear model literature \citep{tchetgen2013closed, tan2019doubly}. This estimator requires only one of two nuisance models to be correctly specified: an outcome regression model or a variant of the propensity score that additionally conditions on the outcome being zero. For the General Randomization setting, we propose an estimator that requires correct specification of an association model, motivated by \citet{vansteelandt2003causal}. We establish a robustness property that the proposed estimator remains valid even if the association model is misspecified, under the null hypothesis of no causal effects. To our knowledge, these are the first estimators targeting moderated CEE on the odds ratio scale. To improve efficiency without compromising robustness, both estimators allow the incorporation of auxiliary covariates via semiparametric efficiency theory and the use of flexible machine learning algorithms to fit nuisance parameters. We establish their asymptotic properties and evaluate finite-sample performance through simulations. Finally, we apply the methods to assess the effect of daily reminders on near-term app engagement in the Drink Less MRT \citep{bell2023notifications}.

In Section \ref{sec:notation}, we define the moderated causal excursion effect on the odds ratio scale. We present the proposed methods in Sections \ref{sec:method-simple-randomization} and \ref{sec: method-general}. Simulation studies are presented in Section \ref{sec: simulation}. The method is illustrated using the Drink Less data in Section \ref{sec: application}. We conclude with a discussion in Section \ref{sec: discussion}.

\section{Moderated Causal Excursion Odds Ratio}
\label{sec:notation}

\subsection{MRT Data Structure}
\label{subsec:mrt-data-structure}

Consider an MRT with $n$ individuals, each enrolled for $T$ decision points at which treatments will be randomly assigned. Variables without subscript $i$ represents observations from a generic individual. Let $A_t$ represent treatment assignment at decision point $t$ with $A_t = 1$ indicating treatment and $A_t = 0$ indicating no treatment. Let $X_t$ represent observations between decision points $t-1$ and $t$. The overbar denotes a sequence of variables from $t=1$, e.g., $\bA_t=(A_1,A_2,\ldots,A_t)$. Information gathered from an individual up to decision point $t$ is $H_t = (X_1,A_1,\ldots,X_{t-1},A_{t-1}, X_t) = (\bX_{t}, \bA_{t-1})$. At each $t$, $A_t$ is randomized with probability $p_t(H_t) := P(A_t = 1 \mid H_t)$, and we sometimes omit $H_t$ to write $p_t$. We assume that observations across individuals are independent and identically distributed.

Let $Y_{t, \Delta} \in \{0,1\}$ denote the binary proximal outcome following treatment assignment at decision point $t$. This outcome depends on data collected over a time window of $\Delta$ decision points after $t$, where $\Delta \geq 1$ is a fixed integer. For example, if decision points occur daily at 8pm, a proximal outcome defined over the next hour or over the next 23 hours (e.g., opening the app within that time window) corresponds to $\Delta = 1$; a proximal outcome defined over the next 47 hours corresponds to $\Delta = 2$. This flexibility in choosing $\Delta$ allows researchers to examine immediate or delayed effects. The choice of $\Delta$ also influences the weights used in estimation. The data structure is illustrated in Figure \ref{fig: data structure}. 
\begin{figure}[ht]
    \centering
    \includegraphics[width=0.9\linewidth]{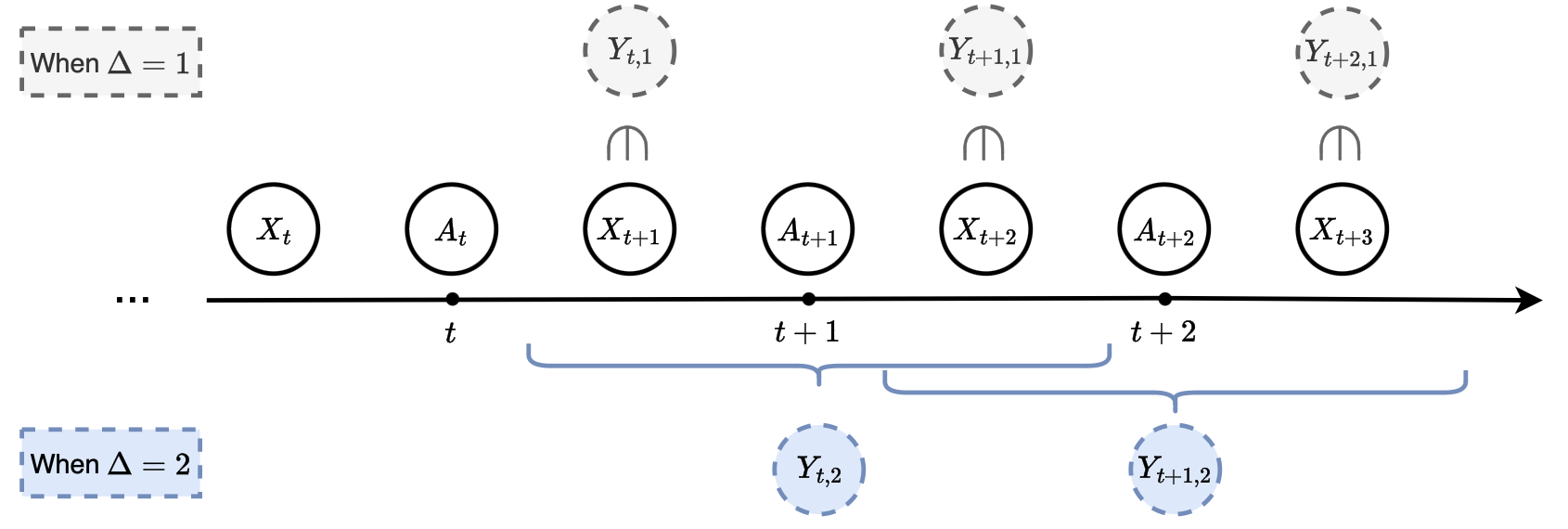}
    \caption{Data structure with examples of $\Delta = 1$ and $\Delta = 2$.} 
    \label{fig: data structure}
\end{figure}

In intervention design, researchers may consider it unsafe or unethical to assign treatments at certain decision points (e.g., when a participant is driving). At these times, participants are deemed ineligible for randomization, which is also referred to as ``unavailable'' in the literature \citep{boruvka2018}. Formally, let $I_t \in X_t$ denote the eligibility indicator at $t$, with $I_t = 1$ denoting being eligible for randomization at $t$ and $I_t = 0$ otherwise. If $I_t = 0$, then $A_t = 0$ deterministically. For simplicity, we sometimes write $I_t$ explicitly alongside $H_t$ in conditional expectations, even though $I_t$ is included in $H_t$: for example, $\EE(\cdot | H_t, I_t = 1)$ should be understood as $\EE(\cdot | H_t \setminus \{I_t\}, I_t = 1)$.

We use double subscripts to denote a sub-sequence of variables; for example, $\bA_{t+1:t+\Delta-1} := (A_{t+1},A_{t+2},\ldots,A_{t+\Delta-1})$. Let $\logit(x) := \log\{x/(1-x)\}$ and $\expit(x) := \{1 + \exp(-x)\}^{-1}$. $\PP_n$ to denote the empirical average over all individuals. For a positive integer $k$, define $[k] := \{1, 2, \ldots, k\}$. We use superscript $\star$ to denote quantities corresponding to the true data generating distribution, $P_0$. We use $\| \cdot \|$ to denote the $L_2$ norm, i.e., $\| f(O) \| = \{\int f(o)^2 P(\mathrm{d}o)\}^{1/2}$ for any function $f$ of the observed data $O$. For a vector $\alpha$ and a vector-valued function $f(\alpha)$, $\partial_\alpha f(\alpha) := \partial f(\alpha) / \partial \alpha^T$ denotes the matrix where the $(i,j)$-th entry is the partial derivative of the $i$-th entry of $f$ with respect to the $j$-th entry of $\alpha$.


\subsection{Causal Excursion Effect on Odds Ratio Scale}
\label{subsec:def-cee}

To define causal effects, we adopt the potential outcome framework \citep{rubin1974estimating, robins1986new}. For an individual, let $X_t(\ba_{t-1})$ be the contextual information that would have been observed, and $A_t(\ba_{t-1})$ be the treatment that would have been observed at decision point $t$ if the individual had been assigned a treatment sequence of $\ba_{t-1}$. The potential outcome of $H_t$ under $\ba_{t-1}$ is $H_t(\ba_{t-1}) = \{X_1, A_1, X_2(a_1), A_2(a_1), X_3(a_2), \ldots, A_{t-1}(\bar{a}_{t-2}), X_t(\ba_{t-1})\}$. The potential outcome for the proximal outcome is $Y_{t,\Delta}(\bar{a}_{t + \Delta - 1})$.

We focus on estimating the causal excursion effect (CEE), defined as a contrast between potential outcomes under two excursions policies, i.e., treatment policies that deviate from the original MRT policy \citep{boruvka2018,guo2021discussion}. Here, treatment policies refer to dynamic treatment regimes that assign treatments sequentially at each decision point based on an individual's history \citep{murphy2003optimal}. Let $\bp$ denote the MRT policy, where for each $t \in [T]$, the randomization probability of $A_t$ is $p_t(H_t)$, and let $\bpi$ denote a reference policy, where for each $t \in [T]$ the randomization probability $\pi_t (H_t)$, subject to the eligibility constraint $\pi_t(H_t) = 0$ if $I_t = 0$. Let $S_t(\bA_{t-1})$ denote a summary of $H_t(\bA_{t-1})$, representing the moderators of interest.  We define the CEE of $A_t$ on $Y_{t, \Delta}$, on the log odds ratio scale, as
\begin{align*}
 &\quad \cee_{\bp, \bpi; \Delta}\{t; S_t(\bA_{t-1})\} \\
   &= \logit \big[\EE_{\bA_{t-1} \sim \bp, \bA_{t+1:t+\Delta-1} \sim \bpi}\{Y_{t, \Delta}(\bA_{t-1}, 1, \bA_{t+1:t+\Delta-1}) \mid S_t(\bA_{t-1}), I_t(\bA_{t-1}) = 1 \} \big] \\
       & ~~~ - \logit \big[\EE_{\bA_{t-1} \sim \bp, \bA_{t+1:t+\Delta-1} \sim \bpi}\{Y_{t, \Delta}(\bA_{t-1}, 0, \bA_{t+1:t+\Delta-1}) \mid S_t(\bA_{t-1}), I_t(\bA_{t-1}) = 1 \} \big].
       \numberthis
    \label{eq: causal excursion effect}
\end{align*}
The effect is conditional on $I_t(\bA_{t-1}) = 1$, which means that we focus only on decision points when a participant is eligible for randomization. This restriction is scientifically meaningful because researchers are interested in effects only when interventions can actually be delivered.

The right-hand side of \eqref{eq: causal excursion effect} contrasts potential outcomes under two excursion policies: $(\bA_{t-1}, 1, \bA_{t+1:t+\Delta-1})$ and $(\bA_{t-1}, 0, \bA_{t+1:t+\Delta-1})$, where $\bA_{t-1}$ (treatments before $A_t$) follow the MRT policy $\bp$, and $\bA_{t+1:t+\Delta-1} = (A_{t+1},...,A_{t+\Delta-1})$ (treatments after $A_t$ but before measuring $Y_{t,\Delta}$) follow the policy $\bpi$. The choice of $\bpi$ influences the interpretation of the CEE. For example, setting $\pi_t = p_t$ for all $t$ resembles the lagged effects considered in \citet{boruvka2018}, now defined on log odds ratio scale. Setting $\pi_t = 0$ for all $t$ (implying $\bA_{t+1:t+\Delta-1} = 0$) resembles the effects considered in \citet{dempsey2020stratified} and \citet{qian2021estimating}. Additional discussion about the general $\bpi$ can be found in \citet{shi2023assessing}. When $\Delta = 1$ (i.e., when the proximal outcome $Y_{t,1}$ is immediately measured before $t+1$), $\bpi$ becomes irrelevant in \eqref{eq: causal excursion effect}, and the CEE simplifies to a contrast between $Y_{t, 1}(\bA_{t-1}, 1)$ and $Y_{t, 1}(\bA_{t-1}, 0)$, which is the most common scenario in applications \citep[e.g.,][]{klasnja2019efficacy,bell2023notifications}.

\begin{rmk}[Effect measures for binary outcome]
    \normalfont
    The causal effect measure for binary outcomes has been widely discussed in the literature \citep{sheps1958shall, brumback2008effect, richardson2017}. Among the commonly used measures, the relative risk and risk difference are notable for their collapsibility, whereas the odds ratio is notable for ensuring probabilities remain within permissible bounds. Similar considerations in defining and interpreting causal effect measures arise in MRTs with binary outcomes \citep{qian2021estimating}. Here we focus on the causal excursion odds ratio, not to argue its superiority, but to complement existing MRT work on the causal excursion risk difference and relative risk by offering an alternative that has so far been infeasible to estimate. This is useful because these measures can lead to categorically different conclusions about effect modification; see Section A of the Supplementary Material for concrete examples in realistic settings. We therefore recommend reporting all three measures to provide a more complete picture.
\end{rmk}

\subsection{Causal Assumptions and Identification}

To identify the CEE, we make the following causal assumptions.
\begin{asu}
    \label{assu: consistency, Positivity, Sequential ignorability.}
    \begin{asulist}
    \item \label{assu: consistency} (Consistency).
    The observed data is the same as the potential outcome under the observed treatment assignment, and one's potential outcomes are not affected by others' treatment assignments. Specifically, $Y_{t, \Delta} = Y_{t, \Delta}(\bA_{t + \Delta - 1})$ and $X_t = X_t(\bA_{t-1})$ for $t \in [T]$. 
    \item \label{assu: positivity} (Positivity).
    There exists $c > 0$ such that $c < \prob(A_t = 1|H_t, I_t = 1) < 1 - c$ almost surely for all $t\in[T]$. 
    \item \label{assu: Sequential ignorability} (Sequential ignorability).
    The potential outcomes $\{X_{t+1}(\ba_{t}),  A_{t+1}(\ba_{t}), \ldots,  X_{T+1}(\ba_{T})\}$ are independent of $A_t$ conditional on $H_t$ for $t \in [T]$. 
    \end{asulist}
\end{asu}  

Positivity and sequential ignorability are guaranteed by the MRT design since the sequential randomization probabilities of treatments are known. Consistency is violated if interference is present, i.e., if the treatment assigned to one participant affects the potential outcome of another participant. In the Section B of Supplementary Material, we show that under Assumption \ref{assu: consistency, Positivity, Sequential ignorability.}, the causal excursion effect in \eqref{eq: causal excursion effect} can be
expressed in terms of observed data:
\begin{align*}
    \cee_{\bp, \bpi; \Delta}(t; S_t) &= \logit \big[\EE\{\EE( W_{t, \Delta} Y_{t, \Delta} \mid A_t = 1, H_t) \mid S_t, I_t = 1 \} \big] \\
    & ~~~ - \logit \big[\EE\{\EE( W_{t, \Delta} Y_{t, \Delta} \mid A_t = 0, H_t) \mid S_t, I_t = 1 \} \big],
    \numberthis
    \label{eq: observed causal excursion effect}
\end{align*}
where $W_{t, \Delta} := \prod_{u = t+1}^{t + \Delta - 1} \big\{\frac{\pi_u(H_u)}{p_u(H_u)}\big\}^{A_u} \big\{\frac{1 - \pi_u(H_u)}{1 - p_u(H_u)}\big\}^{1 - A_u}$  can be interpreted as a change of probability from $p_u$ to $\pi_u$ for future assignment $\bA_{t+1:t+\Delta-1}$, and we set $W_{t, \Delta} := 1$ if $\Delta = 1$.

\subsection{Estimand and Two Analysis Settings}
\label{subsec: cee model}

Let $f_t(S_t)$ be a $p$-dimensional vector of features. We consider estimating the best linear projection of $\cee_{\bp, \bpi; \Delta}(t; S_t)$ on $f_t(S_t)$, averaged over all decision points. Specifically, the true parameter $\beta^\star \in \RR^p$ is defined as
\begin{align}
  \beta^\star = \text{arg} \min_{\beta \in \RR^p} \sum_{t = 1}^T \omega(t) \Big[ \EE \{\cee_{\bp, \bpi; \Delta}(t; S_t) - f_t(S_t)^T\beta\}^2 \Big], \label{eq:cee-model}
\end{align}
where $\omega(t) $ is a pre-specified weight function with $\sumt \omega(t) =1$. This resembles the specification $\cee_{\bp,\bpi;\Delta}=f_t(S_t)^T\beta^\star$, but the estimand remains interpretable even when the true CEE is not linear in $f_t(S_t)$.

The choices of $f_t(S_t)$ and $\omega(t)$ depend on the specific scientific questions. For example, by setting $S_t = \emptyset$ and $f_t(S_t) = 1$, the scalar $\beta^\star$ models the population- and time-averaged causal log odds ratio across all subjects. Alternatively, one can set $f_t(S_t) = (1,S_t)^T$ to assess the population- and time-averaged causal log odds ratio within the strata defined by $S_t$. $f_t(S_t)$ can also include basis functions of $t$ or $S_t$ to capture potential nonlinear effects. Setting the weight $\omega(t) = 1/T$ for all $t$ allows all decision points to equally influence the estimation of $\beta$. Setting $\omega(t_0) = 1$ and $\omega(t) = 0$ for $t \neq t_0$ estimates the effects at a specific $t_0$.

We define the two analysis settings---Simple Randomization and General Randomization---under which we develop separate estimators. The distinction between Simple and General refers to the complexity of the MRT randomization policy relative to the analysis-specific summary variable $S_t$.

\begin{defn}[Simple Randomization (SR)]
    \label{defn: simple randomization}
    The analysis setting is Simple Randomization if the researcher's chosen moderator $S_t$ satisfies $A_t \perp H_t \mid S_t$, i.e.,
    \begin{align*}
        P(A_t = 1|H_t) = P(A_t = 1|S_t). \numberthis \label{eq: simple randomization}
    \end{align*}
    In other words, the $S_t$ fully determines the randomization probability of $A_t$. In this setting, with a slight abuse of notation we write $p_t(H_t) = p_t(S_t)$.
\end{defn}

\begin{defn}[General Randomization (GR)]
    \label{defn: general randomization}
    The analysis setting is General Randomization if the randomization probability possibly depend on the individual's history information beyond $S_t$, i.e. $A_t$ not necessarily independent of $H_t$ given $S_t$. Simple Randomization is a special case of General Randomization.
\end{defn}

Simple vs. General Randomization refers to the relationship between the complexity of the MRT randomization policy and the chosen moderator $S_t$, rather than to the complexity of the randomization policy alone. For example, if the randomization probability for $A_t$ depends only on a covariate $X_{1t}$, then choosing $S_t = X_{1t}$ or $S_t = (X_{1t}, X_{2t})$ falls under Simple Randomization, whereas choosing $S_t = \emptyset$ falls under General Randomization. Most MRTs employ a constant randomization probability \citep[e.g.,][]{klasnja2019efficacy, bell2023notifications}, in which case any choice of $S_t$, including $S_t = \emptyset$, falls under Simple Randomization.

The General Randomization setting implicitly assumes that each individual's randomization probability depends only on their own history and not on others'. MRTs that employ real-time pooling of information across individuals---thereby violating the independence assumption---do not satisfy General Randomization \citep{trella2025deployed}; such settings are not considered in this work.


\section{A Doubly Robust Estimator Under Simple Randomization}
\label{sec:method-simple-randomization}
Under Simple Randomization, the causal excursion effect defined in \eqref{eq: observed causal excursion effect} simplifies to
\begin{align*}
    \cee_{\bp, \bpi; \Delta}(t; S_t) &= \logit \{\EE( W_{t, \Delta} Y_{t, \Delta} \mid S_t, A_t = 1, I_t = 1) \} \\
    & ~~~ - \logit \{\EE( W_{t, \Delta} Y_{t, \Delta} \mid S_t, A_t = 0, I_t = 1) \},
    \numberthis
     \label{eq:cee-identifiability-simple-randomization}
\end{align*}
where $W_{t, \Delta}$ is now $\prod_{u = t+1}^{t + \Delta - 1} \big\{\frac{\pi_u(S_u)}{p_u(S_u)}\big\}^{A_u} \big\{\frac{1 - \pi_u(S_u)}{1 - p_u(S_u)}\big\}^{1 - A_u}$. The proof for \eqref{eq:cee-identifiability-simple-randomization} is in Section B of the Supplementary Material. Because of \eqref{eq:cee-identifiability-simple-randomization}, $\beta^\star$ defined in \eqref{eq:cee-model} can be estimated by a doubly robust procedure, as we describe below.

We define two infinite-dimensional nuisance parameters. For each $t \in [T]$, $r_t(\cdot)$ is a function with truth $r_t^\star(S_t) : = \logit~\EE( W_{t, \Delta} Y_{t, \Delta} \mid S_t, A_t = 0, I_t = 1)$, and $m_t(\cdot)$ is a function with truth $m^\star_t(S_t) := \Pr(A_t = 1\mid S_t, Y_{t,\Delta} = 0, I_t = 1)$. Let $r = \{r_t: t\in[T]\}$ and $m = \{m_t: t\in[T]\}$. Inspired by the logistic partially linear model literature \citep{tchetgen2013closed,tan2019doubly}, we propose the following preliminary estimating function for $\beta^\star$
\begin{align}
    U^\mathrm{SR-prelim}(\beta, r, m) = \sumt \omega(t) & I_t \left\{ W_{t, \Delta}Y_{t, \Delta}e^{-A_t f_t(S_t)^T \beta} - (1 - W_{t, \Delta}Y_{t, \Delta}) e^{r_t(S_t)} \right\} \nonumber \\
    & ~~~~ \times \{ A_t - m_t(S_t) \} f_t(S_t). \label{eq:ee-1-simple-randomization}
\end{align}
We further improve the estimation efficiency using techniques from semiparametric efficiency theory \citep{bkrw1993}, by subtracting from $U^\mathrm{SR-prelim}(\beta, r, m)$ its projection onto the score functions of the treatment selection probabilities. We thus obtain a more efficient estimating function
\begin{align}
    \USR(\beta, r, m, \mu) = \sumt \omega(t) I_t & \bigg[ \{ W_{t, \Delta}Y_{t, \Delta} - A_t \mu_{1t} - (1-A_t)\mu_{0t}\} \big\{ e^{-A_t f_t(S_t)^T \beta} + e^{r_t(S_t)} \big\} \{A_t - m_t(S_t)\} \nonumber \\
    & ~~ + \big\{ \mu_{1t}e^{-f_t(S_t)^T\beta} - (1 - \mu_{1t})e^{r_t(S_t)} \big\} \{ 1 - m_t(S_t) \} p_t(S_t) \nonumber \\
    & ~~ - \big\{ \mu_{0t} - (1 - \mu_{0t})e^{r_t(S_t)} \big\} m_t(S_t) \{1 - p_t(S_t)\} \bigg] f_t(S_t). \label{eq:ee-2-simple-randomization}
\end{align}
Here, $\mu = \{\mu_t: t \in [T]\}$ is another set of nuisance functions, with the truth of $\mu_t(\cdot)$ being $\mu^\star_t(H_t, A_t) := \EE(W_{t, \Delta}Y_{t, \Delta} | H_t, A_t, I_t = 1)$, and in \eqref{eq:ee-2-simple-randomization} we use shorthand notation $\mu_{at} := \EE(W_{t, \Delta}Y_{t, \Delta} | H_t, A_t = a, I_t = 1)$ for $a \in \{0,1\}$. A detailed derivation of \eqref{eq:ee-2-simple-randomization} is in Section C of Supplementary Material. 

A two-stage estimation for $\beta^\star$ proceeds by first obtaining nuisance function estimators $\hat{r}_t$, $\hat{m}_t$, and $\hat\mu_t$ for $t\in[T]$, and then obtain $\hat\beta$ by solving the estimating equation $\PP_n \USR(\beta, \hat{r}, \hat{m}, \hat\mu) = 0$, as detailed in Algorithm \ref{algo:estimator-dr}. We denote this estimator by $\hat{\beta}^\mathrm{SR}$. Theorem \ref{thm: CAN of beta SR} establishes its asymptotic normality, and the proof is in Section D of Supplementary Materials.

\begin{algorithm}[htbp]
    \caption{The doubly robust estimator $\hat\beta^\mathrm{SR}$ under Simple Randomization}
    \label{algo:estimator-dr}
    \spacingset{1.5}
    \vspace{0.3em}
    \textbf{Step 1:} For $t\in[T]$, fit a model for $\EE(W_{t, \Delta}Y_{t, \Delta} \mid S_t, A_t = 0, I_t = 1)$ and denote the fitted model by $\hat{r}(S_t)$; fit a model for $\Pr(A_t = 1\mid S_t, Y_{t,\Delta} = 0, I_t = 1)$ denote the fitted model by $\hat{m}(S_t)$; fit a model for $\EE(W_{t, \Delta}Y_{t, \Delta} \mid H_t, A_t, I_t = 1)$ denote the fitted model by $\hat\mu_t(H_t, A_t)$. In practice, these models may be fitted by pooling across $t \in [T]$.

    \textbf{Step 2:} Obtain $\hat\beta^\mathrm{SR}$ by solving $\PP_n \USR (\beta, \hat{r}, \hat{m}, \hat\mu) = 0$.

    \vspace{0.3em}
\end{algorithm}


\begin{thm}
\label{thm: CAN of beta SR}
Suppose Assumptions \ref{assu: consistency, Positivity, Sequential ignorability.} and regularity conditions in Section D of Supplementary Material hold. Let $r'_t$, $m'_t$, and $\mu'_t$ denote the $L_2$-limits of $\hat{r}_t$, $\hat{m}_t$, and $\hat\mu_t$, respectively, and let $r'$, $m'$, and $\mu'$ denote the corresponding collections over $t \in [T]$. For each $t \in [T]$, suppose that either $r'_t = r^\star_t$ or $m'_t = m^\star_t$, then $\hat{\beta}^\mathrm{SR}$ is consistent. Furthermore, if for each $t \in [T]$
\begin{align}
     \lVert\hat{r}_t-r_t^{\star}\rVert\lVert\hat{m}_t-m^{\star}_t\rVert=o_{p}(n^{-1/2}),
     \label{eq: rate double robustness}
\end{align}
then $\hat{\beta}^\mathrm{SR}$ is asymptotically normal: $\sqrt{n}(\hat{\beta}^\mathrm{SR}-\beta^{\star})\xrightarrow{d}N(0,V^\mathrm{SR})$
as $n\rightarrow\infty$, where 
\begin{align*}
 V^\mathrm{SR} = & \ \ \EE\{\partial_{\beta}\USR(\beta^{{\star}}, r', m', \mu'\}^{-1} ~ \EE\{\USR(\beta^{{\star}}, r', m', \mu')\USR(\beta^{{\star}}, r', m', \mu')^{T}\} \\ 
 & \times \EE\{\partial_{\beta}\USR(\beta^{{\star}},r', m', \mu')\}^{-1,T},
\end{align*}
 and $ V^\mathrm{SR} $ can be consistently estimated by 
\begin{align*}
& \ \ \Big[\PP_{n}\{\partial_{\beta}\USR(\hat{\beta}^\mathrm{SR},\hat{r}, \hat{m}, \hat\mu)\}\Big]^{-1} \Big[\PP_{n}\{\USR(\hat{\beta}^\mathrm{SR},\hat{r}, \hat{m}, \hat\mu)\USR(\hat{\beta}^\mathrm{SR},\hat{r}, \hat{m}, \hat\mu)^{T}\} \Big] \\
& \times \Big[\PP_{n}\{\partial_{\beta}\USR(\hat{\beta}^\mathrm{SR},\hat{r}, \hat{m}, \hat\mu)\}\Big]^{-1,T}.
\end{align*}
\end{thm}

\begin{rmk}
    \normalfont
    The estimator $\hat{\beta}^\mathrm{SR}$ is doubly robust as its consistency requires correct specification of only one of $r_t(S_t)$ and $m_t(S_t)$. These two nuisance functions differ from the outcome regression and propensity score pair used in doubly robust estimation for causal excursion risk difference and relative risk \citep{shi2023meta,liu2024incorporating}, due to technical challenges associated with the logit link and the noncollapsibility of the odds ratio \citep{tchetgen2010doubly}.
    The third nuisance function, $\mu_t(H_t, A_t)$, does not need to be correctly specified, as it is not required for identification and only serves to improve efficiency \citep{lok2024estimating}.

\end{rmk}

\begin{rmk}
    \normalfont
    A wide variety of methods can be used in Step 1 of Algorithm \ref{algo:estimator-dr} to fit the nuisance parameters including kernel regression, spline methods with complexity penalties, and ensemble methods \citep{fan1996local, ruppert2003semiparametric, wang2016functional,wood2017generalized}. The condition \eqref{eq: rate double robustness} is referred to as ``rate double robustness'' \citep{smucler2019unifying}, and a sufficient condition is for both $\hat{r}_t$ and $\hat{m}_t$ to converge to their corresponding truths at $o_p(-n^{1/4})$, Many of the data-adaptive algorithms mentioned earlier can acheive such convergence rate. There is no rate requirement for $\mu_t(H_t,A_t)$ and it can be arbitrarily misspecified. If one were to use parametric models for fitting the nuisance functions, as long as either $r_t$ or $m_t$ is correctly specified, $\hat\beta^\mathrm{SR}$ is consistent and asymptotically normal, and its standard error can be estimated using bootstrap even if \eqref{eq: rate double robustness} does not hold.
\end{rmk}

\section{An Alternative Estimator Under General Randomization}
\label{sec: method-general}

Under General Randomization (Definition \ref{defn: general randomization}), Equation \eqref{eq:cee-identifiability-simple-randomization} does not hold necessarily and thus the estimator in Section \ref{sec:method-simple-randomization} may fail. Below we propose an alternative estimator for the General Randomization setting.

For each $t \in [T]$, we define a nuisance function $\psi_t(\cdot)$ with truth 
\begin{align}
    \psi_t^\star(S_t) := \logit[\EE\{ \EE(W_{t, \Delta} Y_{t, \Delta} \mid H_t, A_t = 1) \mid S_t, I_t = 1\}]. \label{eq:association-model}
\end{align}
The quantity $\psi_t^\star(S_t)$ is the first term in the CEE identification formula \eqref{eq: observed causal excursion effect}. Roughly speaking, it captures the association between the weighted outcome and $S_t$ among decision points with $A_t = 1$. We estimate $\psi_t^\star(S_t)$ and use it to facilitate the estimation of $\beta^\star$. A related auxiliary association model to collapse the causal odds ratio over covariates was considered by \citet{vansteelandt2003causal} for point treatment problems. Suppose for a pre-specified $q$-dimensional function $g_t(S_t)$, there exists $\alpha^\star \in \RR^q$ such that $\psi_t^\star(S_t) = g_t(S_t)^T \alpha^\star$. Then $\hat\alpha$, an estimator for $\alpha^\star$, can be obtained by solving $\PP_n Q(\alpha) = 0$ with
\begin{align}
    Q(\alpha) := \sumt \frac{A_t}{p_t(H_t)}  \Big[ W_{t, \Delta} Y_{t, \Delta} - \expit\{ g_t(S_t)^T\alpha \} \Big] g_t(S_t). \label{eq:ee-for-association}
\end{align}


The combination of the causal excursion odds ratio model \eqref{eq:cee-model} and the association model \eqref{eq:association-model} provides a prediction for the exposure-free outcome for each subject, $\expit\{g(S_t)^T\alpha - f_t(S_t)^T\beta\} A_t + W_{t, \Delta}  Y_{t, \Delta}(1-A_t)$. Thus, sequential ignorability (Assumption \ref{assu: Sequential ignorability}) implies that the following preliminary estimating function is unbiased:
\begin{align}
     U^\mathrm{GR-prelim}(\beta, \alpha) =  \sumt \omega(t) \frac{A_t - p_t(H_t)}{p_t(H_t)\{1 - p_t(H_t)\}} I_t & \Big[ \expit\{g(S_t)^T\alpha - f_t(S_t)^T\beta\} A_t \nonumber \\
     & ~~ + W_{t, \Delta}  Y_{t, \Delta}(1-A_t) \Big] f_t(S_t). \label{eq: ee of method 2}
\end{align}
We then use the projection technique as in Section \ref{sec:method-simple-randomization} to obtain a more efficient estimating function:
\begin{align}
    \UGR(\beta, \alpha, \mu_t) = \sumt \omega(t) I_t & \bigg[ \expit\{g_t(S_t)^T \alpha - f_t(S_t)^T\beta\} - \mu_{0t} \nonumber \\
    & ~~ - \frac{1 - A_t}{1 - p_t(H_t)} \{W_{t, \Delta}Y_{t, \Delta} - \mu_{0t} \} \bigg] f_t(S_t).
    \label{eq: ee for method2 with improved efficiency}
\end{align}

Algorithm \ref{algo:estimator-improved} presents a two-stage estimation procedure that first obtains $\hat\alpha$ from $\PP_n Q(\alpha) = 0$ and $\hat\mu_{at}$ as a model fit for $\EE(W_{t, \Delta}Y_{t, \Delta}|H_t, A_t = a, I_t = 1)$, then solves for $\hat\beta$ from $\PP_n \UGR(\beta, \hat\alpha, \hat\mu) = 0$. We call this estimator $\hat{\beta}^\mathrm{GR}$, and we establish its asymptotic property in Theorem \ref{thm: CAN of beta GR}. 
\begin{algorithm}[htbp]
    \caption{The alternative estimator $\hat\beta^\mathrm{GR}$ under General Randomization}
    \label{algo:estimator-improved}
    \spacingset{1.5}
    \vspace{0.3em}
    \textbf{Step 1:} For each $t\in[T]$, fit a model for $\mu_t(H_t, A_t) = \EE(W_{t, \Delta}Y_{t, \Delta} \mid H_t, A_t, I_t = 1)$ and denote the fitted model by $\hat\mu_t(H_t, A_t)$; in practice, this model fitting can pool over $t$. Obtain $\hat{\alpha}$ by solving \eqref{eq:ee-for-association} using weighted logistic regression.

    \textbf{Step 2:} Obtain $\hat{\beta}^\mathrm{GR}$ by solving $\PP_n \UGR(\beta, \alpha, \hat\mu_t) = 0$.

    \vspace{0.3em}
\end{algorithm}

\begin{thm}
\label{thm: CAN of beta GR}
    Suppose Assumption \ref{assu: consistency, Positivity, Sequential ignorability.} and regularity conditions in Section E of Supplementary Material hold. Suppose that given $g_t(S_t)$, there exists $\alpha^\star$ such that $\psi_t^\star(S_t) = g_t(S_t)^T\alpha^\star$ for $\psi_t^\star$ defined in \eqref{eq:association-model}. Then $\hat{\beta}^\mathrm{GR}$ obtained in Algorithm \ref{algo:estimator-improved} is consistent and asymptotically normal: $\sqrt{n}(\hat{\beta}^\mathrm{GR}-\beta^{\star})\xrightarrow{d}N(0,V^\mathrm{GR})$ as $n\rightarrow\infty$. Furthermore, $V^\mathrm{GR}$ can be consistently estimated by the upper diagonal $p$ by $p$ block matrix of 
    \begin{align}
        \bigg[\PP_{n} \bigg\{ \frac{\partial \Phi(\hat{\beta}, \hat{\alpha}, \hat\mu) }{\partial (\beta^T, \alpha^T)} \bigg\}\bigg]^{-1} ~ \bigg[\PP_{n} \bigg \{\Phi(\hat{\beta}, \hat{\alpha}, \hat\mu) \Phi(\hat{\beta}, \hat{\alpha}, \hat\mu) ^{T}\bigg\} \bigg] ~ \bigg[\PP_{n}\bigg\{ \frac{\partial \Phi(\hat{\beta}, \hat{\alpha}, \hat\mu) }{\partial (\beta^T, \alpha^T)} \bigg\}\bigg]^{-1,T}. \label{eq:vcov-betaGR}
    \end{align}
    Here, $\Phi(\beta, \alpha, \mu):=(\UGR(\beta, \alpha, \mu)^T, Q(\alpha)^T)^T$ is the stacked estimating function for $\beta$ and $\alpha$.
\end{thm}

\begin{rmk}
    \normalfont
    The consistency and asymptotic normality (CAN) of $\hat\beta^\mathrm{GR}$ (Theorem \ref{thm: CAN of beta GR}) requires that the association model $\psi_t(S_t)$ is correctly specified parametrically via $g_t(S_t)^T \alpha$. 
    While we need this condition to theoretically establish CAN, we find in simulation studies that nonparametric estimators for $\psi_t(S_t)$, such as tree-based methods or splines, can still yield good finite-sample performance of $\hat\beta^\mathrm{GR}$. Similar to Theorem \ref{thm: CAN of beta SR}, Theorem \ref{thm: CAN of beta GR} does not require the correct model specification for $\mu_t(H_t, A_t)$.
\end{rmk}

Although the CAN of $\hat\beta^\mathrm{GR}$ requires a correct parametric model for $\psi_t(S_t)$, $\hat\beta^\mathrm{GR}$ is locally robust to misspecified $\psi_t(S_t)$ when $\beta^\star = 0$. This robustness is useful because it controls the type I error rate for testing the null hypothesis of no average causal excursion odds ratio, even when the nuisance models are not correctly specified. We formally establish the result in the following theorem.

\begin{thm}
    \label{thm: robustness of GR under null}
    Suppose Assumption \ref{assu: consistency, Positivity, Sequential ignorability.} and regularity conditions in Section F of Supplementary Material hold. When $\beta^\star = 0$, $\hat{\beta}^\mathrm{GR}$ is consistent and asymptotically normal: $\sqrt{n}(\hat{\beta}^\mathrm{GR}-\beta^{\star})\xrightarrow{d}N(0,V^\mathrm{GR})$ as $n\rightarrow\infty$. Furthermore, $V^\text{GR}$ can be consistently estimated by the upper diagonal $p$ by $p$ block matrix of \eqref{eq:vcov-betaGR}. Here, $\hat{\alpha}$ is the solution to $\PP_n \{Q(\alpha)\} = 0$, and we do not assume that $\psi_t^\star$ is correctly modeled by $g_t(S_t)^T\alpha$ for any $\alpha$. 
\end{thm}


\subsection{Generalized Causal Excursion Effect Models}
\label{subsec:generalization_of_gr_to_arbitray_link_functions}

We extend the result of $\hat\beta^\mathrm{GR}$ beyond logit links in \eqref{eq: observed causal excursion effect} and \eqref{eq:association-model} to other link functions. Let $h(\cdot)$ and $l(\cdot)$ denote strictly monotone, continuously differentiable, and invertible link functions. The generalized causal excursion effect is then defined as
\begin{align*}
    \cee_{\bp, \bpi; \Delta}(t; S_t) &= h \big[\EE\{\EE( W_{t, \Delta} Y_{t, \Delta} \mid A_t = 1, H_t) \mid S_t, I_t = 1 \} \big] \\
    & ~~~ - h \big[\EE\{\EE( W_{t, \Delta} Y_{t, \Delta} \mid A_t = 0, H_t) \mid S_t, I_t = 1 \} \big].
    \numberthis
    \label{eq: generalized CEE}
\end{align*}
We introduce a similar nuisance function $\psi_t(S_t)$, which can be parametrized by a finite-dimensional parameter $\alpha$, i.e.,
\begin{align*}
    \psi_t^\star(S_t) := l[\EE\{ \EE(W_{t, \Delta} Y_{t, \Delta} \mid H_t, A_t = 1) \mid S_t, I_t = 1\}] = g_t(S_t)^T\alpha^\star
\end{align*}
for a pre-specified $g_t(S_t)$. The link $l(\cdot)$ allows flexible modeling of the nuisance function while preserving interpretability of the causal excursion effect in \eqref{eq: generalized CEE}. The estimation of $\alpha^\star$ follows from solving $\PP_n Q^\text{Generalized}(\alpha) = 0$, where
\begin{align*}
    Q^\text{Generalized}(\alpha) = \sumt \frac{A_t}{p_t(H_t)}  I_t \Big[ W_{t, \Delta} Y_{t, \Delta} - l^{-1}\{ g_t(S_t)^T\alpha \} \Big] g_t(S_t).
\end{align*}

After solving for $\alpha$, we propose the following preliminary estimating function for $\beta^\star$:
\begin{align*}
    \sumt & \omega(t) \frac{A_t - p_t(H_t)}{p_t(H_t)\{1 - p_t(H_t)\}} I_t \\ 
     \times & \Bigg\{ h^{-1} \bigg( h\big[l^{-1}\{g(S_t)^T\alpha \} \big]  - f_t(S_t)^T\beta \bigg)  A_t + W_{t, \Delta}  Y_{t, \Delta}(1-A_t) \Bigg\} f_t(S_t). 
\end{align*}
Using the similar projection technique as in Section \ref{sec:method-simple-randomization} to obtain a more efficient estimating function:
\begin{align*}
    U^\text{GR-Generalized}(\beta, \alpha, \mu) = \sumt \omega(t) I_t \Bigg\{ & h^{-1} \bigg( h\big[l^{-1}\{g(S_t)^T\alpha \} \big]  - f_t(S_t)^T\beta\bigg) - \mu_{0t} \\ 
    &- \frac{1-A_t}{1-p_t(H_t)}(W_{t, \Delta}  Y_{t, \Delta} - \mu_{0t}) \Bigg\} f_t(S_t)
\end{align*}

This extension allows researchers to adopt any link function they prefer. For example, the choice $h(a) = \text{probit}(a)$, the inverse of the cumulative distribution function of the standard normal distribution, leads to another widely used link function for binary outcome. Importantly, the asymptotic results established in Theorem \ref{thm: CAN of beta GR} and Theorem \ref{thm: robustness of GR under null} remain valid under this extension and the proof can be found in Section G of Supplementary Material.

\section{Simulation Study}
\label{sec: simulation}

The simulation study demonstrates the performance of the proposed estimators $\hat\beta^\text{SR}$ and $\hat\beta^\text{GR}$ under two distinct scenarios which differ by whether Simple Randomization (Definition \ref{defn: simple randomization}) holds. In Section \ref{subsec: S1}, the randomization probability depends on a time-varying covariate $X_t$, and we set $S_t = X_t$ to assess moderated effect by $X_t$. In this case, the Simple Randomization condition is satisfied. In Section \ref{subsec: S2}, the randomization probability is again a function of $X_t$, and there we set $S_t = \emptyset$ to assess the marginal effect. In this case, not including $X_t$ in $S_t$ violates the Simple Randomization condition. Throughout, we focus on the immediate causal excursion effect with $\Delta = 1$.

\subsection{Estimators' performance under Simple Randomization}
\label{subsec: S1}

The total number of decision points per individual is set to $T = 20$. For each individual, the time-varying covariate $X_t$ is exogenous (independent of past history) and was generated as $X_t \sim \text{Uniform}(0,2)$. The binary variables $(A_t,Y_{t, 1})$ given $H_t$ are generated jointly using Bernoulli distribution with $\Pr(Y_{t,1} = y, A_{t} = a \mid H_t) = s_{ya}/s$, where $s_{00} := 1$, $s_{01}:= e^{\gamma_0 + h_1(X_t)}$, $s_{10} := e^{\gamma_1 + h_2(X_t)}$, $s_{11} := e^{\beta_0^\star + \beta_1^\star X_t + \gamma_0^\star + \gamma_1^\star + h_1(X_t) + h_2(X_t)}$, and $s := s_{00} + s_{10} + s_{01} + s_{11}$. This implies (with detailed derivation in Section H of the Supplementary Material),
\begin{align}
    \Pr(Y_{t,1 } = 1|X_{t},A_{t}) &= \expit{\{(\beta_0^\star + \beta_1^\star X_t) A_t + \gamma_1^\star + h_2(X_t)\}}, \label{eq: simulation S1 Yt generating probability}\\
   \Pr(A_t = 1|X_{t},Y_{t,1 } = 0) &= \expit{ \{ \gamma_0^\star + h_1(X_t) \}}. \nonumber
\end{align}
The true parameter values are set as $\beta_0^\star = 1, \beta^\star_1 = -0.9, \gamma_0^\star = 0.25, \gamma_1^\star = -0.25$. $h_1(X_t),h_2(X_t)$ are nonlinear functions of $X_t$ and $t$: $h_1(X_t) = -0.5  + 1.1q_{2,2}(X_t/2) - 1.2q_{2,2}(t/T)$, $h_2(X_t) = -0.6 -0.4q_{2,2}(X_t/2) + 2q_{2,2}(t/T)$ with $q_{2,2}$ being the probability density function of Beta$(2, 2)$ distribution. Additional simulation results with different choices of $h_1(\cdot)$ and $h_2(\cdot)$, including linear and periodic functions, are included in Section H of the Supplementary Material.

We set $S_t = X_t$ to assess effects moderated by $X_t$. This makes the Simple Randoization condition satisfied. The nuisance function $r_t(S_t), m_t(S_t)$, and $\psi_t(S_t)$ at truth are in a form that can be correctly specified using a nonlinear model with logit link. Specifically,
\begin{align}
    r_t^\star(S_t) &= \logit \Pr(Y_{t,1} = 1|X_t, A_t = 0) = \gamma_1^\star + h_1(X_t) \nonumber\\
    m_t^\star(S_t) &= \Pr(A_t = 1 | X_t, Y_{t,1} = 0) = \expit{ \{ \gamma_0^\star + h_1(X_t) \}} \nonumber\\
    \psi_t^\star(S_t) &= \logit \Pr(Y_{t,1} = 1|X_t, A_t = 1) = \beta_0^\star + \beta_1^\star X_t + \gamma_1^\star + h_2(X_t). \label{eq: simulation S1 true nuisance functions}
\end{align}
It follows from \eqref{eq: simulation S1 Yt generating probability} immediately that the causal excursion effect at truth is $\cee_{\bp, \Delta = 1} (t;S_{t}) = \beta_0^\star + \beta_1^\star X_t$. Four sample sizes ($n=20, 50, 100, 200$), each with 1000 replications, were simulated.

We consider four implementations of $\hat\beta^\text{SR}$ and $\hat\beta^\text{GR}$ (i.e., Algorithms \ref{algo:estimator-dr} and \ref{algo:estimator-improved}) that differ by their ways of fitting the working models for $\hat{r}_t(S_t)$, $\hat{m}_t(S_t)$, $\hat{\psi}_t(S_t)$, and $\hat\mu_t(H_t, A_t)$ as listed in Table \ref{tab:simulation S1 implementation}. Correctly specified models are shown in green, and misspecified models in red. Misspecified nuisance models omit some covariates from the corresponding truth in \eqref{eq: simulation S1 true nuisance functions}. All nuisance models are fitted with generalized additive models (GAMs) with spline bases for the included covariates, using the \texttt{gam} function in the \texttt{mgcv} package in \textsf{R} \citep{wood2017generalized}. Theorem \ref{thm: CAN of beta SR} implies that $\hat{\beta}^{\mathrm{SR}}$ is consistant under Implementations A, B, and C. Theorem \ref{thm: CAN of beta GR} implies that $\hat{\beta}^{\mathrm{GR}}$ is consistant under Implementations A and C.

Additionally, we simulate two competitor methods: a logistic generalized estimating equation (GEE) and a logistic GAM. To make fair comparisons, we correctly specify the mean model for GEE and GAM in Implementation A, and misspecify the mean model for GEE and GAM by leaving out covariates $(t, A_t t)$ under Implementations B, C, and D, matching the models for $\mu_t$ in Table \ref{tab:simulation S1 implementation}. 

\begin{table}[tb]
    \spacingset{1.5}
    \scriptsize
    \centering
    \begin{tabular}{c|cccc}
    \hline
    \textbf{Implementation} & Model for $r_t$ & Model for $m_t$ & Model for $\psi_t$ & Model for $\mu_t$  \\
    \hline
    A & \cellcolor{NatureLightGreen} $s(t) + s(X_t)$
      & \cellcolor{NatureLightGreen} $s(t) + s(X_t)$
      & \cellcolor{NatureLightGreen} $s(t) + s(X_t)$
      & \cellcolor{NatureLightGreen} $s(t) + s(X_t) + A_ts(t) + A_ts(X_t)$ \\

    B & \cellcolor{NatureMidRed} $s(X_t)$
      & \cellcolor{NatureLightGreen} $s(t) + s(X_t)$
      & \cellcolor{NatureMidRed} $s(X_t)$
      & \cellcolor{NatureMidRed} $s(X_t) + A_ts(X_t)$ \\

    C & \cellcolor{NatureLightGreen} $s(t) + s(X_t)$
      & \cellcolor{NatureMidRed} $s(X_t)$
      & \cellcolor{NatureLightGreen} $s(t) + s(X_t)$
      & \cellcolor{NatureMidRed} $ s(X_t) + A_ts(X_t)$ \\

    D & \cellcolor{NatureMidRed} $s(X_t)$
      & \cellcolor{NatureMidRed} $s(X_t)$
      & \cellcolor{NatureMidRed} $s(X_t)$
      & \cellcolor{NatureMidRed} $ s(X_t) + A_ts(X_t)$ \\
      \hline
    \end{tabular}
    \caption{\footnotesize  Four implementations in Simulation S1 that differ in the working models for the nuisance parameters. Expressions like $s(t)$ denote generalized additive models with penalized spline terms for $t$ with $\logit$ link function. Cells are colored green or red to indicate correctly specified or misspecified models.}
    \label{tab:simulation S1 implementation}
\end{table}

Figures~\ref{fig: simulation result of beta0} and \ref{fig: simulation result of beta1} show the bias, mean squared error (MSE), and coverage probability (CP) for $\beta_0^\star$ and $\beta_1^\star$ under the four implementations. As sample size increases, the bias and MSE of $\hat{\beta}_0^{\mathrm{SR}}$ and $\hat{\beta}_1^{\mathrm{SR}}$ decrease under Implementations A, B, and C, and the coverage probability of 95\% confidence interval is close to the nominal level (green solid lines). $\hat{\beta}_0^{\mathrm{GR}}$ and $\hat{\beta}_1^{\mathrm{GR}}$ have decreasing bias and MSE, and close-to-nominal coverage probability under Implementations A and C (blue dashed lines). On the other hand, the logistic GAM is biased under Implementations B, C, and D, i.e., whenever the mean model is misspecified (orange dot-dashed lines). The logistic GEE is biased under all implementations (magenta dotted lines).

\begin{figure}[ht]
    \centering
    \includegraphics[width=0.9\linewidth]{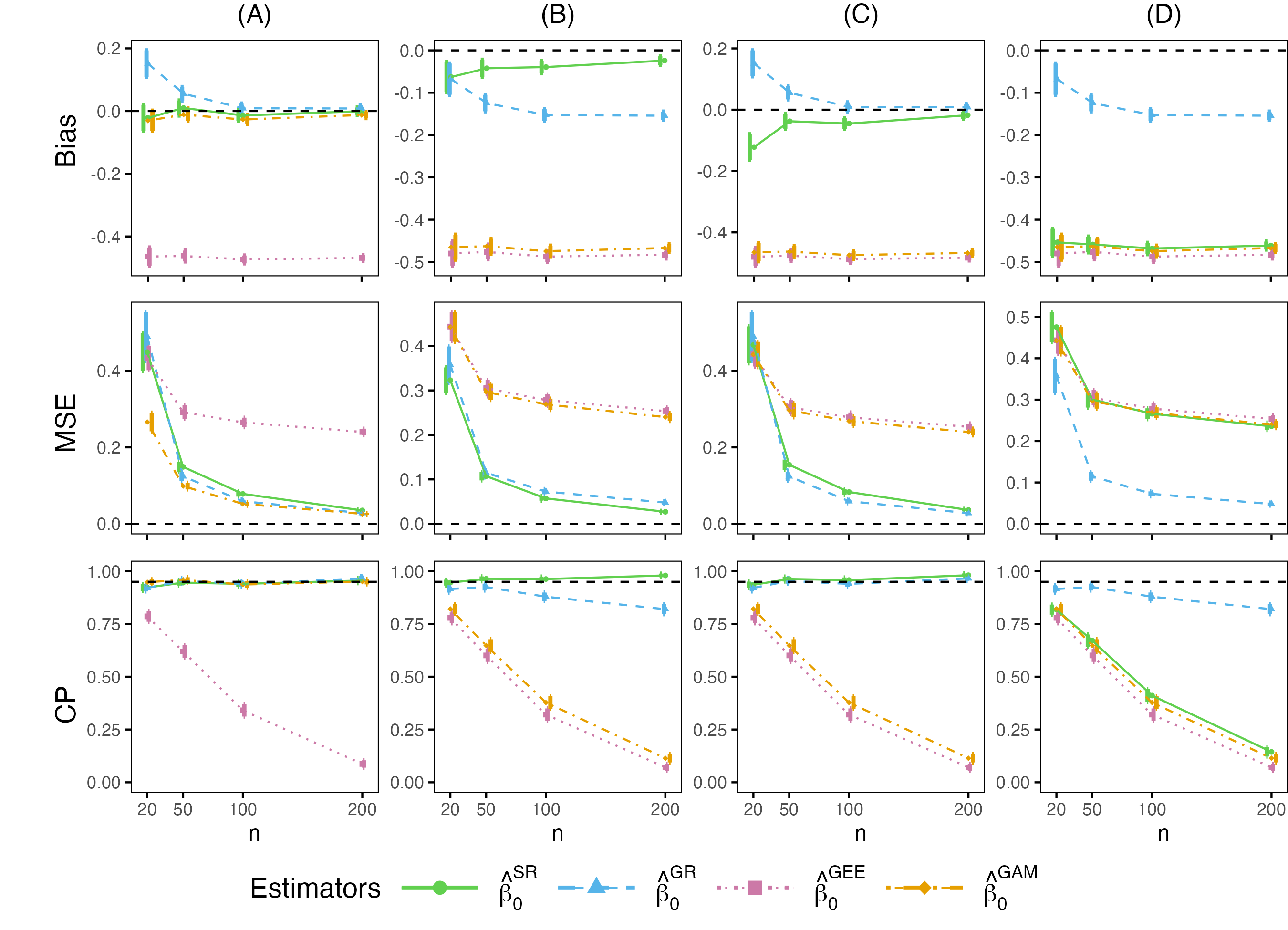}
    \caption{Bias, mean squared error, and coverage probability of $\hat{\beta}_0^\mathrm{SR}$, $\hat{\beta}_0^\mathrm{GR}$, $\hat{\beta}_0^\text{GEE}$, $\hat{\beta}_0^\text{GAM}$ in simulations under Simple Randomization (Section \ref{subsec: S1}). A, B, C, and D represent the implementations specified in Table \ref{tab:simulation S1 implementation}. The vertical bars represent their correponding 95\% confidence intervals.}
    \label{fig: simulation result of beta0}
\end{figure}

\begin{figure}[ht]
    \centering
    \includegraphics[width=0.9\linewidth]{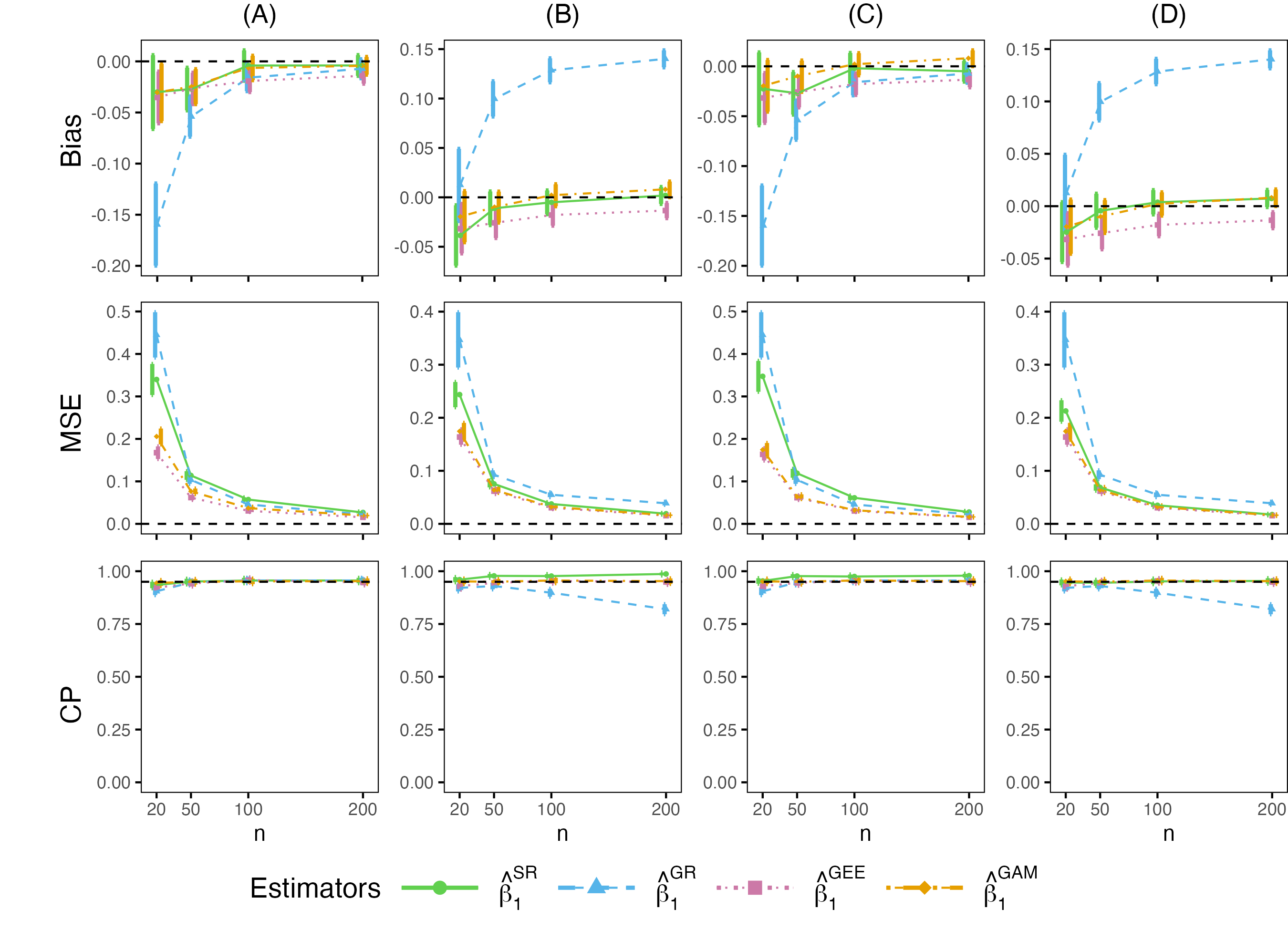}
    \caption{bias, mean squared error, and coverage probability of $\hat{\beta}_1^\mathrm{SR}$, $\hat{\beta}_1^\mathrm{GR}$, $\hat{\beta}_1^\text{GEE}$, $\hat{\beta}_1^\text{GAM}$ in simulations under simple randomization setting (Section \ref{subsec: S1}). A, B, C, and D represent the implementations specified in Table \ref{tab:simulation S1 implementation}. The vertical bars represent their correponding 95\% confidence intervals.}
    \label{fig: simulation result of beta1}
\end{figure}

\subsection{Estimators' performance under General Randomization}
\label{subsec: S2}

The total number of decision points per individual is $T = 20$. The time-varying covariate is generated as $X_t \sim \text{Uniform}(0,2)$. The treatment variable $A_t$ follows a Bernoulli distribution with success probability
\[
\Pr(A_t = 1 \mid X_t) = \expit\{2 - 2 (X_t - 1)\}.
\]
The outcome $Y_t$ is generated from a Bernoulli distribution with success probability $\mu_t^\star(H_t, A_t)$, where $\mu_t^\star(H_t, 1)$ and $\mu_t^\star(H_t, 0)$ can take one of the three forms: linear, where $\mu_t^\star(H_t, 1) = 0.8 - 0.3X_t + 0.1t/T$, $\mu_t^\star(H_t, 0) = 0.1 + 0.3X_t + 0.1t/T$; simple nonlinear, where $\mu_t^\star(H_t, 1) = 0.4 + 0.3 q_{2,2}(X_t/2) - 0.1q_{2,2}(t/T)$, $\mu_t^\star(H_t, 0) = 0.7 -0.4 q_{2,2}(X_t/2) + 0.1q_{2,2}(t/T)$; periodic, where $\mu_t^\star(H_t, 1) = 0.6 + 0.1 \sin(3X_t) - 0.1\sin(t)$, $\mu_t^\star(H_t, 1) = 0.45 + 0.1 \sin(3X_t) + 0.05 \sin(t)$; We set $S_t = \emptyset$, and we numerically compute the true marginal causal excursion effects under each outcome generating models: 0.40, 0.57, or 0.81, respectively, when $\mu_t^\star(H_t, A_t)$ takes the linear, simple nonlinear, or periodic form. 

We specify the models for $\hat{r}_t(S_t)$, $\hat{m}_t(S_t)$, and $\psi_t(S_t)$ to only include $t$, and the model for $\hat\mu_t(H_t, A_t)$ to include $(t, A_t t)$. They are fitted using the \texttt{gam} function from the \texttt{mgcv} package in \textsf{R}. For comparison, we also implement logistic GEE and logistic GAM, with the mean model specified to include the covariates $(A_t, t)$.

Figure~\ref{fig: simulation result S2} presents the simulation results, where each column corresponds to one of three different outcome generating models. Among all four estimators, only $\hat{\beta}_0^{\mathrm{GR}}$ exhibits a decrease in both bias and MSE when the sample size increases, and its 95\% confidence interval achieves a coverage probability close to the nominal level (blue dashed lines). This verifies Theorem \ref{thm: CAN of beta GR}. As expected, none of the other estimators are consistent in these scenarios.

\begin{figure}[htbp]
    \centering
    \includegraphics[width=0.8\linewidth]{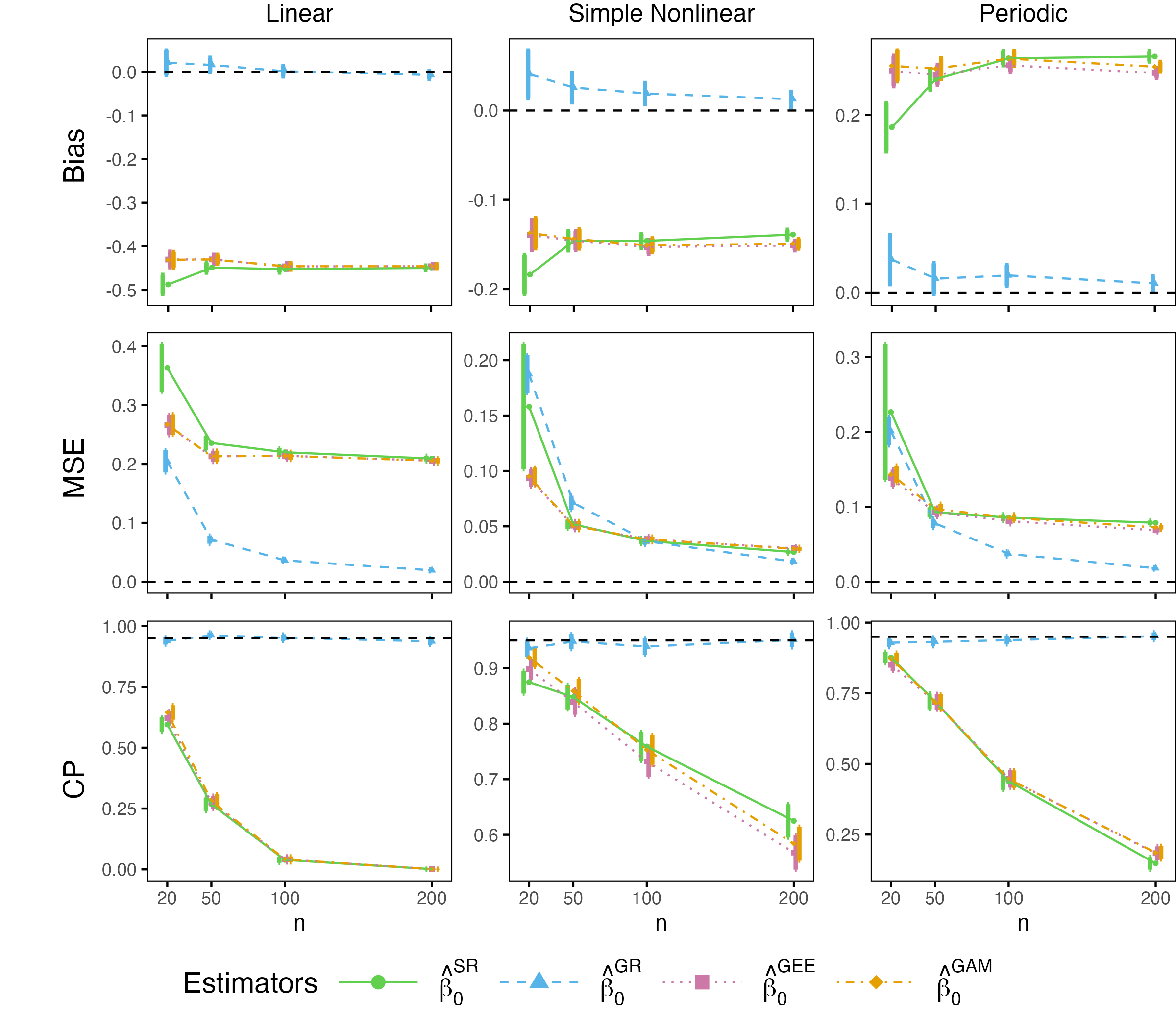}
    \caption{Bias, mean squared error, and coverage probability of $\hat{\beta}_0^\mathrm{SR}$, $\hat{\beta}_0^\mathrm{GR}$, $\hat{\beta}_0^\text{GEE}$, $\hat{\beta}_0^\text{GAM}$ in simulation described in Section \ref{subsec: S2}. The vertical bars represent their correponding 95\% confidence intervals.}
    \label{fig: simulation result S2}
\end{figure}

\section{Application: Drink Less Data}
\label{sec: application}
Drink Less is a smartphone app that aimed to help users reduce harmful alcohol consumption \citep{garnett2019development}. We analyze data from the Drink Less MRT, which evaluated the effect of push notifications on app engagement \citep{bell2020notifications}. 349 participants were randomized daily at 8 p.m. for 30 days. At each decision point, a participant had a 0.6 probability of receiving a push notification prompting the user to open the app and record their daily drinks and a 0.4 probability of receiving no notifications. The proximal outcome is near-term engagement, defined as an indicator of whether the participant opens the app in the hour following the notification (8 p.m. to 9 p.m.).


We conducted one marginal analysis ($S_t = \emptyset$) and three moderation analyses, with the moderator being the decision time index, an indicator for app use before 8 p.m. on that day, and an indicator for receiving a notification on the previous day as a proxy for treatment burden, respectively. In each analysis, we applied both the SR estimator and the GR estimator, with all nuisance models fitted using generalized additive models. Specifically, for the marginal analysis, we included a penalized spline term of the decision point index in the nuisance models. For latter two moderation analyses, we additionally included the respective binary moderator.

\begin{table}[ht]
\centering
\small
\begin{tabular}{l cc c cc}
\hline
Moderator & $\hat{\beta}_0^\text{SR}$ & $\hat{\beta}_1^\text{SR}$ & &  $\hat{\beta}_0^\text{GR}$ & $\hat{\beta}_1^\text{GR}$ \\ 
\hline
None & $1.36 \ (1.18, 1.54)$ & --- &  & $1.36 \ (1.18, 1.54)$ & --- \\
Decision point & $1.52 \ (1.20, 1.83)$ & $-0.01 \ (-0.03, 0.01)$ & & $1.52 \ (1.20, 1.84)$ & $-0.01 \ (-0.03, 0.01)$ \\
Use before 8pm & $1.40 \ (1.15, 1.64)$ & $0.02 \ (-0.34, 0.38)$ &  & $1.39 \ (1.15, 1.64)$ & $0.04 \ (-0.33, 0.41)$ \\
Prior-day notification & $1.36 \ (1.07, 1.65)$ & $-0.01 \ (-0.38, 0.36)$ & &  $1.37 \ (1.08, 1.66)$ & $-0.02 \ (-0.38, 0.35)$ \\
\hline
\end{tabular}

\caption{Estimated marginal and moderated CEE using the SR and GR estimators for the Drink Less MRT. The coefficients and the 95\% confidence intervals (in parentheses) are on the log odds ratio scale.}
\label{tab: DA results}
\end{table}

Table \ref{tab: DA results} shows the estimated CEE parameters with their 95\% confidence intervals. When the moderator is None ($S_t = \emptyset$), $\hat{\beta}_0^\text{SR}$ and $\hat{\beta}_0^\text{GR}$ represent the estimated marginal effect using the proposed SR and GR estimators. When a specific moderator is included, $\hat{\beta}_0^\text{SR}$ and $\hat{\beta}_1^\text{GR}$ represent the intercept in the CEE model, whereas $\hat{\beta}_1^\text{SR}$ and $\hat{\beta}_1^\text{GR}$ represent the slope. Push notifications significantly increase the rate of near-term engagement, with a log odds ratio of 1.36 (95\% confidence interval [1.18, 1.54]), which translates to an odds ratio of 3.89 [3.25, 4.66].
We didn't detect significant effect moderation by any of the three moderators. The numeric similarity between the two estimators is likely due to the constant randomization probability in the Drink Less MRT. Section I of Supplementary Material contains estimates of OR, RR, and RD measures, and the results are qualitatively in the same direction.

\section{Discussion}
\label{sec: discussion}

In this paper, we developed two new methods for estimating the causal excursion effect for binary outcomes on the odds ratio scale. The first is a doubly robust estimator that applies when the treatment is randomized based solely on a prespecified set of covariates. This setting, referred to as Simple Randomization, means that the randomization probability depends only on variables included in the moderator set $S_t$, which is selected by the analyst at the analysis stage. The second estimator is designed for more general scenarios in which the randomization mechanism may depend on variables outside of $S_t$. In such cases, we introduce an alternative estimator that incorporates an association model to adjust for the additional complexity in treatment assignment.

From a practical standpoint, the choice between these two estimators depends on how the randomization mechanism relates to the analyst’s choice of moderators. If the randomization probability depends only on variables included in $S_t$, the doubly robust estimator should be used, as it offers protection against misspecification of either the outcome model or the treatment model. This includes common cases in MRTs where randomization is constant over time or depends on a subset of observed covariates. However, if the analyst is interested in estimating a marginal effect (e.g., setting $S_t = \emptyset$) but randomization depends on history information, or if interested in estimating effect moderation by some $S_t$ but the randomization probability further depends on other variables in $H_t$, then Simple Randomization condition is violated. In these settings, the general estimator with the association model is recommended.

There are several directions for future research. One direction for future work is to incorporate nonparametric model in the causal excursion effect model. While our approach employs nonparametric methods for estimating nuisance functions, the causal excursion effect model remains parametric for interpretability of low-dimensional models. Future work could explore nonparametric specifications of the causal excursion effect, allowing for greater flexibility in capturing complex relationships between covariates and treatment effects. Another important direction for future research is to develop a corresponding sample size calculator for detecting a pre-specified differential effect between treatment options. Existing work for binary outcome has focused on relative risk \citep{cohn2023sample}, and extending these to causal excursion odds ratio would be highly valuable for planning micro-randomized trials. Finally, future work could explore improving efficiency of the casual excursion odds rario estimator, for example, by incorporating information from auxiliary variables \citep{shi2025incorporating}.

\section*{Acknowledgements}

We thank Dr. Claire Garnett, Dr. Olga Perski, Dr. Henry W.W. Potts, and Dr. Elizabeth Williamson, and Dr. Lauren Bell for their important contributions to the Drink Less MRT.

\section*{Supplementary Material}
Supplementary material includes theoretical proofs, including identifiability result and proofs of Theorem \ref{thm: CAN of beta SR}, Theorem \ref{thm: CAN of beta GR}, and Theorem \ref{thm: robustness of GR under null}. We also present details on the generative model of simulation, addition simulation results, and illustrations of different measures for binary outcome.  The code for replicating the simulations and data analysis can be accessed at \url{https://github.com/jiaxin4/logisticCEE}.




%% file: appendix.tex
\section{Illustrative Examples of Discrepant Effect Moderation Directions}
\label{supp-sec: discrepancy} 
For binary outcomes, various effect moderation measures exist, including odds ratios (ORs), risk ratios (RRs), and risk differences (RDs). It is known that they can disagree on the \emph{direction} of effect moderation \citep{brumback2008effect}. In the following, we use simple numerical examples to illustrate this: identity, log, and logit. Consider a simple setting with $T=1$, in which case the covariate-treatment-outcome trio is denoted by $(X,A,Y)$ and the causal excursion effects become standard causal effects. The binary covariate $X$ is generated from $\text{Bernoulli}(0.6)$, and the treatment variable $A$ is generated using $\text{Bernoulli}(0.4)$. The binary response satisfies
\begin{align}
	\text{Pr}(Y = 1|A = 1, X) &= a_1 + b_1 X \nonumber\\
	\text{Pr}(Y = 1|A = 0, X) &= a_0 + b_0 X, \label{eq: supp-A generative model}
\end{align}
where $a_1, b_1, a_0, b_0$ are constants $\in (0,1)$ and $a_i + b_i < 1$ for $i = 0, 1$. We consider three causal excursion effects with different link functions:
\begin{align}
	\text{CEE}_{\text{RD}}(X) &= \PP(Y|A = 1, X) - \PP(Y|A = 0, X) = \beta_0^\text{RD} + \beta_1^\text{RD} X, \nonumber \\
	\text{CEE}_{\text{RR}}(X) &= \log \frac{\PP(Y|A = 1, X)}{\PP(Y|A = 0, X)}  = \beta_0^\text{RR} + \beta_1^\text{RR} X \nonumber \\
	\text{CEE}_{\text{OR}}(X) &= \logit{\PP(Y|A = 1, X)} - \logit{\PP(Y|A = 0, X)} = \beta_0^\text{OR} + \beta_1^\text{OR} X,
\label{eq: definition of three link functions}
\end{align}
where $\beta_0^\text{RD}, \beta_0^\text{RR}, \beta_0^\text{OR}$ are intercepts, and $\beta_1^\text{RD}, \beta_1^\text{RR}, \beta_1^\text{OR}$ are the slopes capturing the direction and magnitude of effect moderation by $X$. In particular, $\text{CEE}_{\text{RD}}(X)$ represents causal risk difference in \citet{boruvka2018}, $\text{CEE}_{\text{RR}}(X)$ represents causal risk ratio in \citet{qian2021estimating}, and $\text{CEE}_{\text{OR}}(X)$ represents the causal odds ratio in this paper. We say the effect moderation measures defined in \eqref{eq: definition of three link functions} disagree in direction if one of the effect moderation coefficients $\{\beta_1^\text{RD}, \beta_1^\text{RR}, \beta_1^\text{OR}\}$ has the opposite sign to the other two.

Under the generative model \eqref{eq: supp-A generative model}, the intercepts and the slopes in \eqref{eq: definition of three link functions} can be written in terms of $a_1, b_1, a_0, b_0$:
\begin{align*}
	\text{CEE}_{\text{RD}}(X) &= (a_1 - a_0) + (b_1 - b_0) X, \\
	\text{CEE}_{\text{RR}}(X) &= \log \frac{a_1 + b_1 X}{a_0 + b_0 X} = \log\frac{a_1}{a_0} + \log\frac{a_0(a_1 + b_1)}{a_1(a_0 + b_0)} X, \\
	\text{CEE}_{\text{OR}}(X) &= \logit{(a_1 + b_1 X)} - \logit{(a_0 + b_0 X)} \\ 
	&= \logit{(a_1)} - \logit{(a_0)} + \{\logit{(a_1 + b_1)} - \logit{(a_0 + b_0)} - \logit{(a_1)} + \logit{(a_0)}\}X.
\end{align*}
In the following examples, we set different values of $a_1, b_1, a_0, b_0$ to provide examples where the effect moderation measures disagree in direction. The conditional success probabilities of $Y$ for each example are displayed in Table \ref{tab:conditional probability of Y} to show that none of the examples are extreme or pathological in terms of the success probabilities.

\begin{exam} ($\beta_1^\text{RD}$ disagrees in direction with $\beta_1^\text{RR}$ and $\beta_1^\text{OR}$)
	When $a_1 = 0.48, b_1 = 0.36, a_0 = 0.10, $and $b_0 = 0.32$, the effect moderations for each causal effects defined in \eqref{eq: definition of three link functions} can be calculated as
	\begin{align*}
		\beta_1^\text{RD} = 0.04, \beta_1^\text{RR} \approx -0.875, \beta_1^\text{OR} \approx -0.136.
	\end{align*}
\label{exam: identity is different}
The effect at each strata of $X = 1$ and $X = 0$ can be found in Table \ref{tab:cee_examples}. 
\end{exam}

\begin{exam} ($\beta_1^\text{RR}$ disagrees in direction with $\beta_1^\text{RD}$ and $\beta_1^\text{OR}$)
	When $a_1 = 0.71, b_1 = 0.14, a_0 = 0.09, $and $b_0 = 0.05$, the effect moderations for each causal effects defined in \eqref{eq: definition of three link functions} can be calculated as
	\begin{align*}
		\beta_1^\text{RD} = 0.09, \beta_1^\text{RR} \approx -0.26, \beta_1^\text{OR} \approx 0.34.
	\end{align*}
\label{exam: log is different}
\end{exam}

\begin{exam} ($\beta_1^\text{OR}$ disagrees in direction with $\beta_1^\text{RD}$ and $\beta_1^\text{RR}$)
	When $a_1 = 0.89, b_1 = 0.08, a_0 = 0.59, $and $b_0 = 0.15$, the effect moderations for each causal effects defined in \eqref{eq: definition of three link functions} can be calculated as
	\begin{align*}
		\beta_1^\text{RD} = -0.07, \beta_1^\text{RR} \approx -0.14, \beta_1^\text{OR} \approx 0.70.
	\end{align*}
\label{exam: logit is different}
\end{exam}

\begin{table}[ht]
\centering
\begin{tabular}{ccccc}
\toprule
Example & Strata & RD within Strata & log RR within Strata & log OR within Strata \\
\midrule
\multirow{2}{*}{A.1}
  & X = 0 & 0.38 & 1.57 & 2.12 \\
  & X = 1 & 0.42 & 0.69 & 1.98 \\
\addlinespace
\multirow{2}{*}{A.2}
  & X = 0 & 0.62 & 2.07 & 3.21 \\
  & X = 1 & 0.71 & 1.80 & 3.55 \\
\addlinespace
\multirow{2}{*}{A.3}
  & X = 0 & 0.30 & 0.41 & 1.73 \\
  & X = 1 & 0.23 & 0.27 & 2.43 \\
\addlinespace
\multirow{2}{*}{A.4}
  & X = 0 & -0.25 & -0.98 & -1.33 \\
  & X = 1 & -0.50 & -0.98 & -2.23 \\
 \addlinespace
\multirow{2}{*}{A.5}
  & X = 0 & 0.20 & 0.41 & 0.81 \\
  & X = 1 & 0.10 & 0.12 & 0.81 \\
\bottomrule
\end{tabular}
\caption{Causal excursion risk difference (RD), relative risk (RR), and odds ratio (OR) within strata \(X=0\) and \(X=1\) for each example in Section \ref{supp-sec: discrepancy}}
\label{tab:cee_examples}
\end{table}

\begin{exam} ($\beta_1^\text{RR} = 0, \beta_1^\text{OR} \neq 0$)
	When $a_1 = 0.15, b_1 = 0.15, a_0 = 0.4, $and $b_0 = 0.4$, the effect moderations for relative risk and odds ratio can be calculated as
	\begin{align*}
		\beta_1^\text{RR} = 0, \beta_1^\text{OR} \approx -0.90.
	\end{align*}
\label{exam: log is 0 but logit is not}
\end{exam}

\begin{exam} ($\beta_1^\text{RR} \neq 0, \beta_1^\text{OR} = 0$)
	When $a_1 = 0.6, b_1 = 0.3, a_0 = 0.4, $and $b_0 = 0.4$, the effect moderations for relative risk and odds ratio can be calculated as
	\begin{align*}
		\beta_1^\text{RR} \approx -0.29, \beta_1^\text{OR} \approx 0.
	\end{align*}
\label{exam: logit is 0 but log is not}
\end{exam}

\begin{table}[ht]
\centering
\scriptsize
\begin{tabular}{ccccc}
\toprule
Example & $\Pr(Y=1|A=1, X=1)$ & $\Pr(Y=1|A=1, X=0)$ & $\Pr(Y=1|A=0, X=1)$ & $\Pr(Y=1|A=0, X=0)$  \\
\midrule
A.1 & 0.84	&0.48	&0.42	&0.10 \\ 
A.2 & 0.85	&0.71	&0.14	&0.09 \\
A.3 & 0.97	& 0.89	&0.74	&0.59 \\
A.4 & 0.3	&0.15	&0.8	&0.4 \\
A.5 & 0.9	&0.6	&0.8	&0.4 \\
\bottomrule
\end{tabular}
\caption{The conditional probability of $Y$ for each example in Section \ref{supp-sec: discrepancy}}
\label{tab:conditional probability of Y}
\end{table}

These examples make clear that relying on a single link can mask or even reverse the pattern of effect moderation. In particular, Example \ref{exam: identity is different}, Example \ref{exam: log is different}, and Example \ref{exam: logit is different} show that these measures can yield opposite signs of effect moderation. Example \ref{exam: log is 0 but logit is not} and Example \ref{exam: logit is 0 but log is not} demonstrate that one measure can indicate zero moderation while another measure shows a nonzero effect. Thus, we recommend that researchers estimate and report all measures to achieve a more complete understanding of the effect, and this work for estimating causal excursion effects with a logit link supplements existing literature on identity and log links.

\section{Details of Simulations}
\label{supp-sec:details_of_simulations}

\subsection{Data generating mechanisms}
For each individual, time-varying covariate \( X_t \) is exogenous (independent of past history) and was generated using a uniform distribution: \( X_t \sim \text{Unif}(0,2) \). We jointly generated binary variables $A_t, Y_{t,1}$ given $X_t$ according to the probabilities proportional to the entries in the following table:
\[
\begin{array}{|c|c|c|}
	\hline
	& A_t  = 0 & A_t = 1 \\ 
	\hline
	Y_{t, 1} = 0 & 1 & e^{\alpha_0 + h_1(X_t)} \\ 
	Y_{t, 1} = 1 & e^{\alpha_1 + h_2(X_t)} & e^{\beta_0^\star + \beta_1^\star X_t + \alpha_0 + \alpha_1 + h_1(X_t) + h_2(X_t)} \\
		\hline
\end{array}
\]
Here, $\beta_0, \beta_1, \alpha_0, \alpha_1$ are true parameter values, and $h_1(X_t),h_2(X_t)$ are functions in $X_t$ which follow one of three patterns: linear, simple nonlinear, or periodic in $X_t$ and $t$; the detailed functional forms are presented in Table \ref{tab: data generating function for h1 and h2}.
\begin{table}[h]
	\centering
	\begin{tabular}{c c c c} 
		\toprule
		&  $g(\lambda_0, \lambda_1, \lambda_2)$  & $h_1(X_t)$ & $h_2(X_t)$ \\
		\midrule
		Linear &  $\lambda_0 + \lambda_1 X_t +  \lambda_2 t$  & $g(-1, 1, -0.1)$ & $g(0.5, 0.2, -0.1)$\\
		Simple nonlinear & $\lambda_0 + \lambda_1\{q_{2,2}(X_t/2) + \lambda_2\{q_{2,2}(t/T)\}$ & $g(-0.5, 1.1, -1.2)$& $g(-0.6, -0.4, 0.9)$ \\
		Periodic & $\lambda_0 + \lambda_1\sin(X_t) + \lambda_2\sin(t) \}$  & $g(-0.5, 0.8, -0.8)$& $g(-0.2, -0.4, 1)$ \\
		\bottomrule
	\end{tabular}
	\caption{The generating functions are defined through a $g(\cdot, \cdot, \cdot)$ function, which is defined in the last column. $q_{2,2}(\cdot)$ denotes the density function of Beta(2,2) distribution}
		\label{tab: data generating function for h1 and h2}
\end{table}

To find the causal excursion effect model and nuisance function model at truth under such data generating mechanism, one can calculate 
\begin{align*}
	 P(Y_{t,1 } = 1|X_{t},A_{t} = 1) &= \frac{P(Y_{t,1} = 1, A_{t} = 1||X_{t})}{P(A_t = 1|X_t)} \\
	&= \frac{e^{\beta_0^\star + \beta_1^\star X_t+ \alpha_0 + \alpha_1 + h_1(X_t) + h_2(X_t)}}{e^{\alpha_0 + h_1(X_t)} + e^{\beta_0^\star + \beta_1^\star X_t + \alpha_0 + \alpha_1 + h_1(X_t) + h_2(X_t)}} \\
	& = \expit{\{\beta_0^\star + \beta_1^\star X_t + \alpha_1 + h_2(X_t)\}},
\end{align*}
and similarly one can obtain $P(Y_{t,1 } = 1|X_{t},A_{t} = 0) = \expit{\{\alpha_1 + h_2(X_t)\}}$. Thus, the nuisance function $r_t(S_t)$ under the truth is
\begin{align*}
	r_t^\star(S_t) = \logit{\{P(Y_{t,1 } = 1|X_{t},A_{t})\}} = (\beta_0^\star + \beta_1^\star X_t) A_t + \alpha_1 + h_2(X_t).
\end{align*}
To find $m^\star_t(S_t)$, we can calculate
\begin{align*}
	P(A_t = 1|X_{t},Y_{t,1 } = 0) &= \frac{P(A_t = 1,Y_{t,1 } = 0|X_{t})}{P(Y_{t,1 } = 0|X_t)} \\
	&= \frac{e^{\alpha_0 + h_1(X_t)}}{1 + e^{\alpha_0 + h_1(X_t)}} \\
	&= \expit{ \{ \alpha_0 + h_1(X_t) \}},
\end{align*}
thus the nuisance function $m_t(S_t)$ under the truth is
\begin{align*}
	m^\star_t(S_t) = P(A_t = 1|X_{t},Y_{t,1 } = 0) = \expit{ \{ \alpha_0 + h_1(X_t) \}}.
\end{align*}
The causal excusion effect model given $S_t = X_t$ in this case is
\begin{align*}
 	\text{CEE} (t;X_{t}) &= \text{logit} \{P(Y_{t,1} = 1 |X_t,A_t =1) \}-  \text{logit} \{P(Y_{t,1} = 1 |X_t,A_t =0)\} \\
 	&=\{ \beta_0^\star + \beta_1^\star X_t +  \alpha_1 + h_2(X_t) \}- \{ \alpha_1 + h_2(X_t)\} \\
 	& = \beta_0^\star + \beta_1^\star X_t
\end{align*}

\subsection{Additional simulation results}
\label{supp-sec:additional_simulation_results}
 In the simulation section of the paper, we presented the performance of our estimators with simple nonlinear generating functions under scenario 1 (when the randomization probability only depends at most on $S_t$). In this section, we focus on linear and periodic generating functions. The data generating machanism is same as the main paper except that the generating functions for $h_1(X_t)$ and $h_2(X_t)$. We set the total number of decision points per individual to \( T = 20 \). For each individual, the time-varying covariate \( X_t \) is exogenous (independent of past history) and was generated using a uniform distribution: \( X_t \sim \text{Unif}(0,2) \). Then the binary variables $(A_t,Y_{t, 1})$ given $X_t$ are generated jointly using Bernoulli distribution such that the following mechanisms are obtained,
\begin{align*}
    \Pr(Y_{t,1 } = 1|X_{t},A_{t}) &= \expit{\{(\beta_0^\star + \beta_1^\star X_t) A_t + \alpha_1^\star + h_2(X_t)\}}, \\
   \Pr(A_t = 1|X_{t},Y_{t,1 } = 0) &= \expit{ \{ \alpha_0^\star + h_1(X_t) \}}.
\end{align*}
where $\beta_0^\star = 1, \beta^\star_1 = -0.9, \alpha_0^\star = 0.25, \alpha_1^\star = -0.25$; $h_1(X_t),h_2(X_t)$ are functions in $X_t$ which can one of the two forms in $X_t$ and $t$: linear, where $h_1(X_t) = -1 + X_t - 0.1t$ and $h_2(X_t) = 0.5 + 0.2 X_t - 0.1t$; periodic, where $h_1(X_t) = -0.5 + 0.8\sin(X_t) - 0.8\sin(t)$ and $h_2(X_t) = -0.2 - 0.4\sin(X_t) + 1\sin(t)$. In addition, the causal excursion effect is $\cee_{\bp, \Delta = 1} (t;S_{t}) = \beta_0^\star + \beta_1^\star X_t$. For each simulation setup, we used four sample sizes: 20, 50, 100, 200, each with 1000 replications.

We considered four implementations that differed in the ways of the model fit for nuisance parameters. Each implementation specifies four generalized additive models with splines of specified variables, fitted using the \texttt{gam} function from the \texttt{mgcv} package in R \citep{wood2017generalized}.  Details can be found in Table \ref{supp-tab:simulation S1 implementation}.
\begin{table}[tb]
    \scriptsize
    \centering
    \begin{tabular}{c|cccc}
    \hline
    \textbf{Implementation} & Model for $r_t$ & Model for $m_t$ & Model for $\psi_t$ & Model for $\mu_t$  \\
    \hline
    1 & \cellcolor{NatureLightGreen} $s(t) + s(X_t)$
      & \cellcolor{NatureLightGreen} $s(t) + s(X_t)$
      & \cellcolor{NatureLightGreen} $s(t) + s(X_t) + A_ts(t) + A_ts(X_t)$
      & \cellcolor{NatureLightGreen} $s(t) + s(X_t)$ \\

    2 & \cellcolor{NatureMidRed} $s(X_t)$
      & \cellcolor{NatureLightGreen} $s(t) + s(X_t)$
      & \cellcolor{NatureMidRed} $s(X_t) + A_ts(X_t)$
      & \cellcolor{NatureMidRed} $s(X_t)$ \\

    3 & \cellcolor{NatureLightGreen} $s(t) + s(X_t)$
      & \cellcolor{NatureMidRed} $s(X_t)$
      & \cellcolor{NatureLightGreen} $s(t) + s(X_t) + A_ts(t) + A_ts(X_t)$
      & \cellcolor{NatureMidRed} $s(X_t)$ \\
    4 & \cellcolor{NatureMidRed} $s(X_t)$
      & \cellcolor{NatureMidRed} $s(X_t)$
      & \cellcolor{NatureMidRed} $s(X_t) + A_ts(X_t)$
      & \cellcolor{NatureMidRed} $s(X_t)$ \\
      \hline
    \end{tabular}
    \caption{\footnotesize  Four implementations in Section \ref{supp-sec:additional_simulation_results} that differ in how nuisance parameters are estimated. Expressions like $s(t)$ denote generalized additive models with penalized spline terms for $t$ with $\logit$ link function. Cells are colored green or red to indicate correctly specified or misspecified models.}
    \label{supp-tab:simulation S1 implementation}
\end{table}

Additionally, we compare our estimators to two benchmark methods: a generalized estimating equation (GEE) with a logit link and a logistic GAM with a logit link. To make fair comparisons, we correctly specify the mean model for GEE and GAM in Implementation (A), and misspecify the mean model for GEE and GAM by leaving out covariates $(t, A_t t)$ in in Implementation (B), (C), (D).

\begin{figure}[ht]
    \centering
    \includegraphics[width=0.9\linewidth]{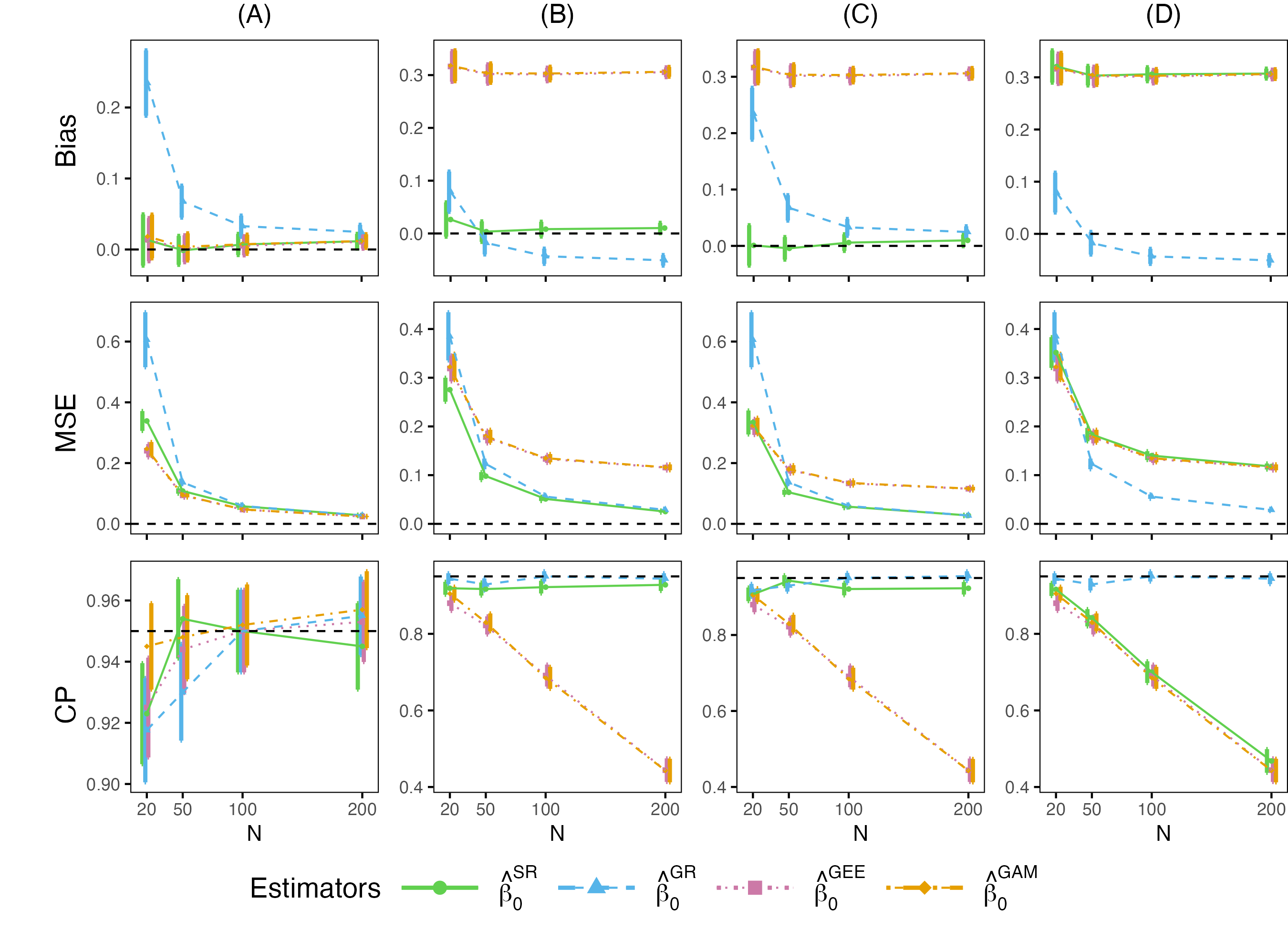}
    \caption{Bias, mean squared error, and coverage probability of $\hat{\beta}_0^\text{SR}$, $\hat{\beta}_0^\text{GR}$, $\hat{\beta}_0^\text{GEE}$, $\hat{\beta}_0^\text{GAM}$ in simulations with linear generating function described in Section \ref{supp-sec:additional_simulation_results}.}
    \label{sup-fig: simulation result of beta0 with linear generating function}
\end{figure}

\begin{figure}[ht]
    \centering
    \includegraphics[width=0.9\linewidth]{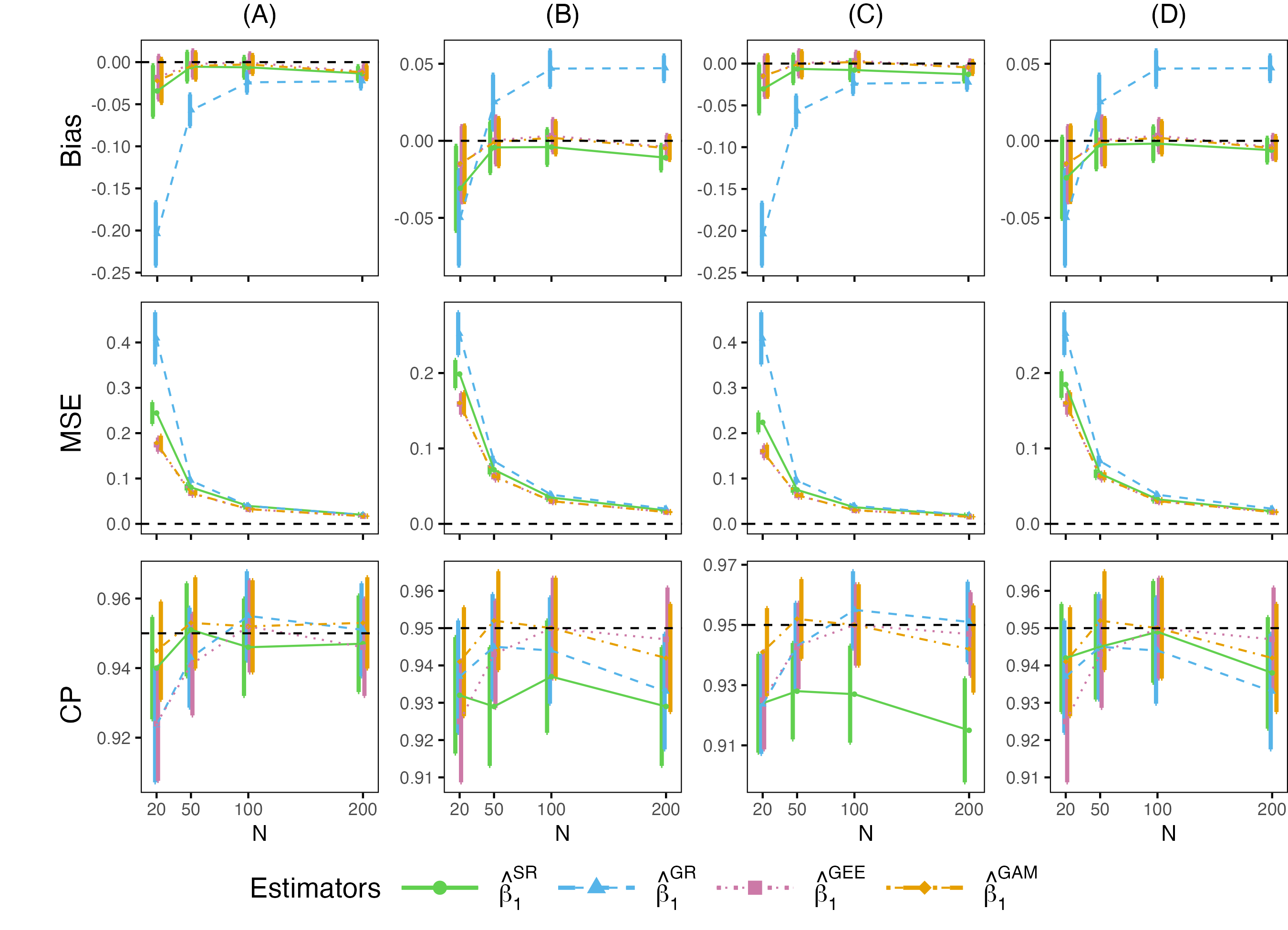}
    \caption{bias, mean squared error, and coverage probability of $\hat{\beta}_1^\text{SR}$, $\hat{\beta}_1^\text{GR}$, $\hat{\beta}_1^\text{GEE}$, $\hat{\beta}_1^\text{GAM}$ in simulations with linear generating function described in Section \ref{supp-sec:additional_simulation_results}.}
    \label{sup-fig: simulation result of beta1 with linear generating function}
\end{figure}{}

\begin{figure}[ht]
    \centering
    \includegraphics[width=0.9\linewidth]{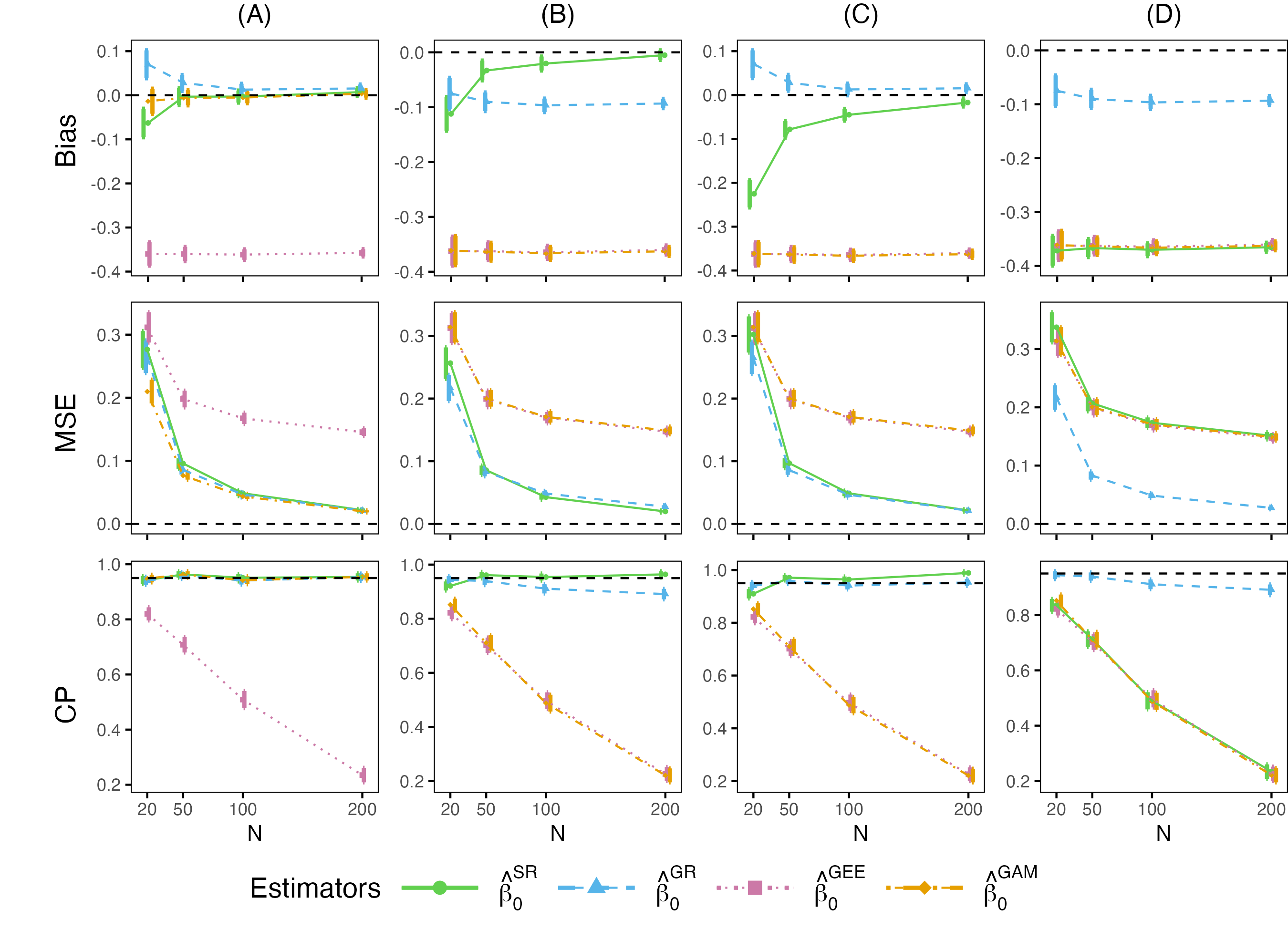}
    \caption{Bias, mean squared error, and coverage probability of $\hat{\beta}_0^\text{SR}$, $\hat{\beta}_0^\text{GR}$, $\hat{\beta}_0^\text{GEE}$, $\hat{\beta}_0^\text{GAM}$ in simulations with periodic generating function described in Section \ref{supp-sec:additional_simulation_results}.}
    \label{sup-fig: simulation result of beta0 with periodic generating function}
\end{figure}

\begin{figure}[ht]
    \centering
    \includegraphics[width=0.9\linewidth]{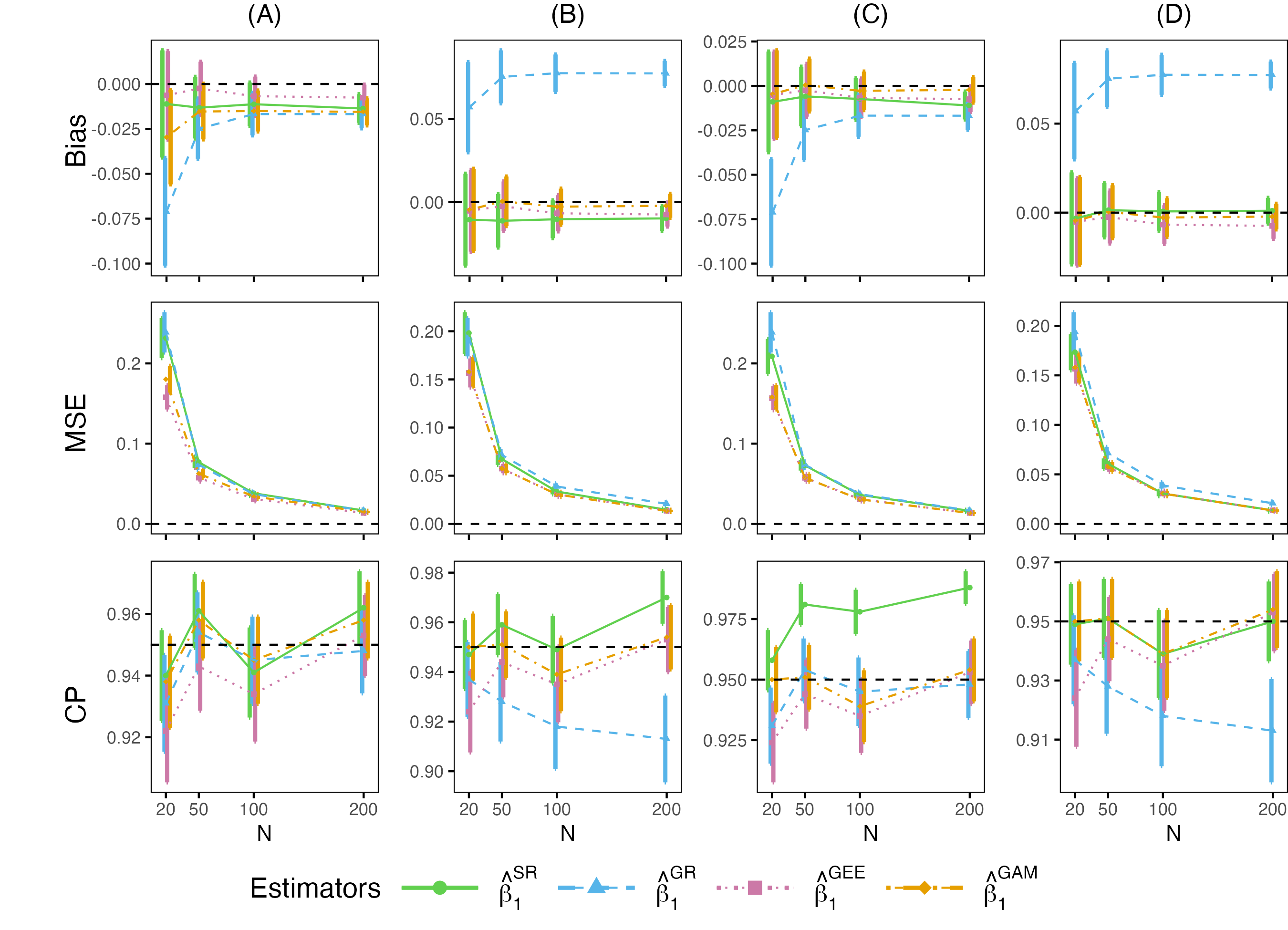}
    \caption{bias, mean squared error, and coverage probability of $\hat{\beta}_1^\text{SR}$, $\hat{\beta}_1^\text{GR}$, $\hat{\beta}_1^\text{GEE}$, $\hat{\beta}_1^\text{GAM}$ in simulations with periodic generating function described in Section \ref{supp-sec:additional_simulation_results}.}
    \label{sup-fig: simulation result of beta1 with periodic generating function}
\end{figure}

Figure \ref{sup-fig: simulation result of beta0 with linear generating function} and \ref{sup-fig: simulation result of beta1 with linear generating function} show the bias, mean squared error (MSE), and coverage probability (CP) of our estimators and two comparators with linear generating functions. As expected, bias and MSE of $\hat{\beta}^\text{SR}$ decreases to 0 and CP gets close to nominal level as sample size increases under implementation (A), (B), and (C). The estimator $\hat{\beta}^\text{GR}$'s MSE decreases and CP gets close to nominal level under implementation (A) and (C). The estimators obtained from GAM and GEE are biased whenever the mean model is misspecified (implementation (B), (C), (D)). Figure \ref{sup-fig: simulation result of beta0 with periodic generating function} and \ref{sup-fig: simulation result of beta1 with periodic generating function} shows the bias, MSE, and CP of the estimators under periodic generating functions. The estimator $\hat{\beta}^\text{SR}$'s MSE decreases and CP gets close to nominal level under implementation (A), (B), and (C). The estimator $\hat{\beta}^\text{GR}$ has good performance under implementation (A) and (C).

\section{Proof of Identifiability Results}
\label{supp-sec:proof_of_identifiability_results}

\subsection{Proof of General Identifiability Result Eq. (2.2)}
The proof in this section extends the identifiability proof for continuous outcome in \citet{yu2024doubly}. Formally, we would like to establish that
\begin{align*}
    \cee_{\bp, \bpi; \Delta}(t; S_t) &= \logit[\EE\{\EE( W_{t, \Delta} Y_{t, \Delta} \mid A_t = 1, H_t) \mid S_t, I_t = 1 \} ] \\
    & ~~~ - \logit[\EE\{\EE( W_{t, \Delta} Y_{t, \Delta} \mid A_t = 0, H_t) \mid S_t, I_t = 1 \} ].
    \numberthis
    \label{supp-eq: observed causal excursion effect}
\end{align*}
It suffices to show
\begin{align*}
&\EE\{Y_{t, \Delta}(\bA_{t-1}, a, \bA_{t+1:t+\Delta-1})|S_t(\bA_{t-1}), I_t(\bA_{t-1}) = 1\} \\ 
 = &\EE\{\EE( W_{t, \Delta} Y_{t, \Delta} \mid A_t = a, H_t) \mid S_t, I_t = 1 \} \ \ \text{for } a = 0, 1.
\end{align*}

\subsubsection{Lemma}
\begin{lem}
\label{lem: identifibility lemma} Under Assumption 1 in the main paper, for any $1\leq k\leq\Delta$,
we have
\begin{align}
 & \EE(Y_{t,\Delta}|H_{t},A_{t}=a,I_{t}=1)\nonumber \\
= & \EE\{\prod_{u=t+1}^{t+k-1}(\frac{\pi_{u}}{p_{u}})^{A_{u}}(\frac{1-\pi_{u}}{1-p_{u}})^{1-A_{u}}Y_{t,\Delta}|H_{t},A_{t}=a,I_{t}=1\}\label{eq: identi lemma 1 - result equality}
\end{align}
\end{lem}
\begin{proof}
For $k=1$, (\ref{eq: identi lemma 1 - result equality}) holds because
we defined $\prod_{u=t+1}^{t}(\frac{\pi_{u}}{p_{u}})^{A_{u}}(\frac{1-\pi_{u}}{1-p_{u}})^{1-A_{u}}=1.$
In the following, we assume $\Delta\geq2$, and we prove the lemma
by mathemtical induction.

Suppose (\ref{eq: identi lemma 1 - result equality}) holds for $k=k_{0}$
for some $1\leq k_{0}\leq\Delta-1$, that is 
\begin{align}
 & \EE(Y_{t,\Delta}|H_{t},A_{t}=a,I_{t}=1)\nonumber \\
 & =\EE\{\prod_{u=t+1}^{t+k_{0}-1}(\frac{\pi_{u}}{p_{u}})^{A_{u}}(\frac{1-\pi_{u}}{1-p_{u}})^{1-A_{u}}Y_{t,\Delta}|H_{t},A_{t}=a,I_{t}=1\}\label{eq: lemma - hold for k0}
\end{align}

Let $\zeta=\prod_{u=t+1}^{t+k_{0}-1}(\frac{\pi_{u}}{p_{u}})^{A_{u}}(\frac{1-\pi_{u}}{1-p_{u}})^{1-A_{u}}$,
we have
\begin{align}
 & \EE\{\prod_{u=t+1}^{t+k_{0}}(\frac{\pi_{u}}{p_{u}})^{A_{u}}(\frac{1-\pi_{u}}{1-p_{u}})^{1-A_{u}}Y_{t,\Delta}|H_{t},A_{t}=a,I_{t}=1\}\nonumber \\
= & \EE\big\{\zeta(\frac{\pi_{t+k_{0}}}{p_{t+k_{0}}})^{A_{t+k_{0}}}(\frac{1-\pi_{t+k_{0}}}{1-p_{t+k_{0}}})^{1-A_{t+k_{0}}}Y_{t,\Delta}|H_{t},A_{t}=a,I_{t}=1\big\}\nonumber \\
= & \EE\Big[\EE\big\{\zeta(\frac{\pi_{t+k_{0}}}{p_{t+k_{0}}})Y_{t,\Delta}|H_{t+k_{0}},A_{t+k_{0}}=1,H_{t},A_{t}=a,I_{t}=1\big\} p_{t+k_{0}}\nonumber \\
 & +\EE\big\{\zeta(\frac{1-\pi_{t+k_{0}}}{1-p_{t+k_{0}}})Y_{t,\Delta}|H_{t+k_{0}},A_{t+k_{0}}=0,H_{t},A_{t}=a,I_{t}=1\big\}(1-p_{t+k_{0}})\Big|H_{t},A_{t}=a,I_{t}=1\Big]\label{eq:lemm1 - using iterated expectation}\\
= & \EE\Big[\EE\big\{ Y_{t,\Delta}|H_{t+k_{0}},A_{t+k_{0}}=1,H_{t},A_{t}=a,I_{t}=1\big\}\pi_{t+k_{0}} \nonumber \\
&+\EE\big\{Y_{t,\Delta}|H_{t+k_{0}},A_{t+k_{0}}=0,H_{t},A_{t}=a,I_{t}=1\big\}(1-\pi_{t+k_{0}})\Big|H_{t},A_{t}=a,I_{t}=1\Big]\label{eq: lemma - by using k =00003D k0 hold}
\end{align}
where (\ref{eq:lemm1 - using iterated expectation}) follows from
law of iterated expectation. (\ref{eq: lemma - by using k =00003D k0 hold})
follows from (\ref{eq: lemma - hold for k0}) and sequentially ignorability.
Also, we have
\begin{align}
 & \EE(Y_{t,\Delta}|H_{t},A_{t}=a,I_{t}=1)\nonumber \\
= & \EE\big\{\EE(Y_{t,\Delta}|H_{t+k_{0}},A_{t+k_{0}}=1,H_{t},A_{t}=a,I_{t}=1)\pi_{t+k_{0}}\nonumber \\
 & +\EE(Y_{t,\Delta}|H_{t+k_{0}},A_{t+k_{0}}=1,H_{t},A_{t}=a,I_{t}=1)(1-\pi_{t+k_{0}})\Big|H_{t,}A_{t}=a,I_{t}=1\big\}.\label{eq: lemma - expand Y}
\end{align}
Then, by (\ref{eq: lemma - expand Y}) and (\ref{eq: lemma - by using k =00003D k0 hold}),
we have
\begin{align*}
 & \EE(Y_{t,\Delta}|H_{t},A_{t}=a,I_{t}=1)\\
 & =\EE\{\prod_{u=t+1}^{t+k_{0}}(\frac{\pi_{u}}{p_{u}})^{A_{u}}(\frac{1-\pi_{u}}{1-p_{u}})^{1-A_{u}}Y_{t,\Delta}|H_{t},A_{t}=a,I_{t}=1\},
\end{align*}
i.e., we showed that (\ref{eq: identi lemma 1 - result equality})
holds for $k=k_{0}+1$. This completes the proof.
\end{proof}

Now, we show the following: under Assumption 1 (SUTVA, positivity, sequential ignorability),
\begin{align*}
 & \EE\{Y_{t,\Delta}(\bar{A}_{t-1},a,\bar{A}_{t+1:(t+\Delta-1)})|S_{t}(\bar{A}_{t-1})=s,I_{t}(\bar{A}_{t-1})=1\}\\
= & \EE\{\EE(W_{t,\Delta}Y_{t,\Delta}|A_{t}=a,H_{t},I_{t}=1)|S_{t}=s,I_{t}=1\}
\end{align*}

\begin{proof}
We have the following sequence of equality:
\begin{align}
 & \EE\{Y_{t,\Delta}(\bar{A}_{t-1},a,\bar{A}_{t+1:(t+\Delta-1)})|S_{t}(\bar{A}_{t-1})=s,I_{t}(\bar{A}_{t-1})=1\} \nonumber \\
= & \EE\Big[\EE\{Y_{t,\Delta}(\bar{A}_{t-1},a,\bar{A}_{t+1:(t+\Delta-1)})||H_{t}(\bar{A}_{t-1}),I_{t}(\bar{A}_{t-1})=1\}|S_{t}(\bar{A}_{t-1})=s,I_{t}(\bar{A}_{t-1})=1\Big] \label{eq:identi lemma2 - law of iterated expectation} \\
= & \EE\Big[\EE\{Y_{t,\Delta}(\bar{A}_{t-1},a,\bar{A}_{t+1:(t+\Delta-1)})|H_{t},I_{t}=1\}|S_{t}=s,I_{t}=1\Big]\label{eq:identi lemma2 - consistency assumption} \\
= & \EE\Big[\EE\{Y_{t,\Delta}(\bar{A}_{t-1},a,\bar{A}_{t+1:(t+\Delta-1)})|H_{t},I_{t}=1,A_{t}=a\}|S_{t}=s,I_{t}=1\Big]\label{eq:identi lemma2 - sequantial ignorability assumption}\\
= & \EE\Big[\EE\{Y_{t,\Delta}(\bar{A}_{t-1},A_{t},\bar{A}_{t+1:(t+\Delta-1)})|H_{t},I_{t}=1,A_{t}=a\}|S_{t}=s,I_{t}=1\Big]\\
= & \EE\{\EE(Y_{t,\Delta}|H_{t},I_{t}=1,A_{t}=a)|S_{t}=s,I_{t}=1\}\label{eq: identi lemm2 - potential Y =00003D observed Y}\\
= & \EE\{\EE(W_{t,\Delta}Y_{t,\Delta}|A_{t}=a,H_{t},I_{t}=1)|S_{t}=s,I_{t}=1\}\label{eq:identi lemma2 - from lemma 1}
\end{align}
where (\ref{eq:identi lemma2 - law of iterated expectation}) follows
from the law of iterated expectation, (\ref{eq:identi lemma2 - consistency assumption})
follows from consistency assumption, (\ref{eq:identi lemma2 - sequantial ignorability assumption})
follows from sequential ignorability assumption, \eqref{eq: identi lemm2 - potential Y =00003D observed Y}
follows from consistency assumption, and (\ref{eq:identi lemma2 - from lemma 1})
follows from Lemma \ref{lem: identifibility lemma}. This completes
the proof. 
\end{proof}

\subsection{Proof of Identifibaility Result Under Simple Randomization}
In this section, we would like to show that under simple randomization, 
\begin{align*}
    \cee_{\bp, \bpi; \Delta}(t; S_t) &= \logit \{\EE( W_{t, \Delta} Y_{t, \Delta} \mid S_t, A_t = 1, I_t = 1) \} \\
    & ~~~ - \logit \{\EE( W_{t, \Delta} Y_{t, \Delta} \mid S_t, A_t = 0, I_t = 1) \}.
\end{align*}
It suffices to show that
\begin{align*}
&\EE\{\EE( W_{t, \Delta} Y_{t, \Delta} \mid A_t = a, H_t) \mid S_t, I_t = 1 \} \\
= &\EE( W_{t, \Delta} Y_{t, \Delta} \mid S_t, A_t = a, I_t = 1) \ \ \text{for } a = 0, 1.
\end{align*}

We have the following equality:
\begin{align}
	&\EE\{\EE( W_{t, \Delta} Y_{t, \Delta} \mid A_t = a, H_t) \mid S_t, I_t = 1 \} \nonumber \\
= & \EE\{Y_{t, \Delta}(\bA_{t-1}, a, \bA_{t+1:t+\Delta-1})|S_t(\bA_{t-1}), I_t(\bA_{t-1}) = 1\} \label{eq: simple randomization proof use identifiability result}\\
= & \EE\{Y_{t, \Delta}(\bA_{t-1}, a, \bA_{t+1:t+\Delta-1})|S_t(\bA_{t-1}), A_t = a, I_t(\bA_{t-1}) = 1\} \label{eq: simple randomization proof use simple randomization} \\
= & \EE\{Y_{t, \Delta}|S_t, A_t = a, I_t = 1\} \label{eq: simple randomization proof use consistency} \\ 
= & \EE\{W_{t, \Delta}Y_{t, \Delta}|S_t, A_t = a, I_t = 1\} \label{eq: simple randomization proof use lemma} 
\end{align}
where \eqref{eq: simple randomization proof use identifiability result} holds because of the identifiability result; \eqref{eq: simple randomization proof use simple randomization} holds because of the simple randomization setting; \eqref{eq: simple randomization proof use consistency} holds because of the consistency assumption; and \eqref{eq: simple randomization proof use lemma} follows from the variation of Lemma \ref{lem: identifibility lemma}, which replaces $H_t$ with $S_t$.

\section{Derivations of the Efficient Estimation Equations}
\label{supp-sec:derivations_of_the_estimation_equations_usr_beta_r_m_mu_and_ugr_beta_alpha_mu_}

\subsection{\texorpdfstring{Derivation of the Estimation Equation $\USR(\beta, r, m, \mu)$}{Derivation of the Estimation Equation USR(beta, r, m, mu)}}
\label{supp-sec:_texorpdfstring_derivation_of_the_estimation_equation_usr_beta_r_m_mu_}
We subtract from $ U^\text{SR-prelim}(\beta, r, m)$ onto the score of the treatment assignment to obtain a more efficient estimating equation:
\begin{align}
   U^\text{SR-prelim}(\beta, r, m) - \sum_{u = 1}^T \Big[ \EE\{ U^\text{SR-prelim}(\beta, r, m)|H_u, A_u\}- \EE\{ U^\text{SR-prelim}(\beta, r, m)|H_u\} \Big] \Bigg).  \label{eq: derivation of USR the big equation}
\end{align}
Define $U_t^\text{SR-prelim}(\beta, r_t, m_t)$ to ba summand in $U^\text{SR-prelim}(\beta, r, m)$ such that $U^\text{SR-prelim}(\beta, r, m) = \sum_{t=1}^T U_t^\text{SR-prelim}(\beta, r_t, m_t)$, and 
\begin{align*}
U_t^\text{SR-prelim}(\beta, r_t, m_t) = \omega(t) & I_t \left\{ W_{t, \Delta}Y_{t, \Delta}e^{-A_t f_t(S_t)^T \beta} - (1 - W_{t, \Delta}Y_{t, \Delta}) e^{r_t(S_t)} \right\}  \\
    & ~~~~ \times \{ A_t - m_t(S_t) \} f_t(S_t). 
\end{align*}

Then \eqref{eq: derivation of USR the big equation} can further expanded to
\begin{align}
	&\sum_{t=1}^T \Bigg(U_t^\text{SR-prelim}(\beta, r_t, m_t) - \EE\{ U_t^\text{SR-prelim}(\beta, r_t, m_t)|H_t, A_t\} + \EE\{ U_t^\text{SR-prelim}(\beta, r_t, m_t)|H_t\} \Bigg) \label{eq: USR first part expanded}\\
	&-\sum_{t=1}^T \Bigg(\sum_{1\leq u \leq T, u\neq t} \Big[ \EE\{ U_t^\text{SR-prelim}(\beta, r_t, m_t)|H_u, A_u\} - \EE\{ U_t^\text{SR-prelim}(\beta, r_t, m_t)|H_u\} \Big] \Bigg). \label{eq: USR second part expanded}
\end{align}

We will first derive \eqref{eq: USR first part expanded}. By definition of $U_t^\text{SR-prelim}(\beta, r_t, m_t)$, we have,
\begin{align*}
	&\EE\{ U_t^\text{SR-prelim}(\beta, r_t, m_t)|H_t, A_t\} \\ 
	= &\omega(t) I_t \Big[ \EE(W_{t,\Delta}Y_{t,\Delta}|H_t, A_t) e^{-A_t f_t(S_t)^T \beta} - \big\{ 1 - \EE(W_{t,\Delta}Y_{t,\Delta}|H_t, A_t) \big\} e^{r_t(S_t)} \Big] 	\{ A_t - m_t(S_t) \} f_t(S_t). \numberthis \label{eq: first first part of EEUSR}
\end{align*}
We also have,
\begin{align*}
	&\EE\{ U_t^\text{SR-prelim}(\beta, r_t, m_t)|H_t\} \\ 
	= & \EE\{ U_t^\text{SR-prelim}(\beta, r_t, m_t)|H_t, I_t = 1\}I_t + \EE\{ U_t^\text{SR-prelim}(\beta, r_t, m_t)|H_t, I_t = 0\}(1 - I_t) \\ 
	=& \EE\{ U_t^\text{SR-prelim}(\beta, r_t, m_t)|H_t, I_t = 1, A_t = 1\}P(A_t = 1|H_t, I_t = 1)I_t \\ 
	&+ \EE\{ U_t^\text{SR-prelim}(\beta, r_t, m_t)|H_t, I_t = 1, A_t = 0\}P(A_t = 0|H_t, I_t = 1)I_t \\ 
	&+ \EE\{ U_t^\text{SR-prelim}(\beta, r_t, m_t)|H_t, I_t = 0\}(1 - I_t) \\ 
	=& \omega(t) I_t \big\{ \mu_{1t}^\star e^{-f_t(S_t)^T\beta} - (1 - \mu_{1t}^\star)e^{r_t(S_t)} \big\} \{ 1 - m_t(S_t) \} p_t(S_t) f_t(S_t)\\ 
	&- \omega(t) I_t \big\{ \mu_{0t}^\star - (1 - \mu_{0t}^\star)e^{r_t(S_t)} \big\} m_t(S_t) \{1 - p_t(S_t)\} f_t(S_t) \numberthis \label{eq: first second part of EEUSR}.
\end{align*}

Putting together \eqref{eq: first first part of EEUSR} and \eqref{eq: first second part of EEUSR}, we have
\begin{align*}
	&U_t^\text{SR-prelim}(\beta, r_t, m_t) - \EE\{ U_t^\text{SR-prelim}(\beta, r_t, m_t)|H_t, A_t\} + \EE\{ U_t^\text{SR-prelim}(\beta, r_t, m_t)|H_t\} \\
	=& \omega(t) I_t \bigg[ (W_{t, \Delta}Y_{t, \Delta} - \mu_t^\star )e^{-A_t f_t(S_t)^T \beta} - \big\{ (1-W_{t, \Delta}Y_{t, \Delta}) - (1-\mu_t^\star) \big\}  e^{r_t(S_t)} \{A_t - m_t(S_t)\} \nonumber \\
	& + \big\{ \mu_{1t}^\star e^{-f_t(S_t)^T\beta} - (1 - \mu_{1t}^\star )e^{r_t(S_t)} \big\} \{ 1 - m_t(S_t) \} p_t(S_t) \nonumber \\
	& - \big\{ \mu_{0t}^\star - (1 - \mu_{0t}^\star)e^{r_t(S_t)} \big\} m_t(S_t) \{1 - p_t(S_t)\} \bigg] f_t(S_t) \\ 
	=& \omega(t) I_t \bigg[ \{ W_{t, \Delta}Y_{t, \Delta} - \mu_t^\star \} \big\{ e^{-A_t f_t(S_t)^T \beta} + e^{r_t(S_t)} \big\} \{A_t - m_t(S_t)\} \nonumber \\
	& + \big\{ \mu_{1t}^\star e^{-f_t(S_t)^T\beta} - (1 - \mu_{1t}^\star)e^{r_t(S_t)} \big\} \{ 1 - m_t(S_t) \} p_t(S_t) \nonumber \\
	& - \big\{ \mu_{0t}^\star - (1 - \mu_{0t}^\star)e^{r_t(S_t)} \big\} m_t(S_t) \{1 - p_t(S_t)\} \bigg] f_t(S_t)
\end{align*} 

The terms in \eqref{eq: USR second part expanded} cannot be analytically derived without imposing additional models on the relationship between current and lagged variables in the longitudinal history. Therefore, we omit them when deriving the improved estimating function. This calculation was also used in \citet{bao2025estimating} and \citet{qian2025distal}. Therefore, this completes the derivation.

\subsection{\texorpdfstring{Derivation of the Estimation Equation $\UGR(\beta, \alpha, \mu)$}{Derivation of the Estimation Equation UGR(beta, alpha, mu)}}
\label{supp-sec:_texorpdfstring_derivation_of_the_estimation_equation_ugr_beta_alpha_mu_}
Define $U_t^\text{GR-prelim}(\beta, \alpha)$ to be a summand of $U^\text{GR-prelim}(\beta, \alpha)$ such that $U^\text{GR-prelim}(\beta, \alpha) = \sum_{t=1}^T U_t^\text{GR-prelim}(\beta, \alpha)$, and 
\begin{align*}
     U_t^\mathrm{GR-prelim}(\beta, \alpha) = \omega(t) \frac{A_t - p_t(H_t)}{p_t(H_t)\{1 - p_t(H_t)\}} I_t & \Big[ \expit\{g(S_t)^T\alpha - f_t(S_t)^T\beta\} A_t \nonumber \\
     & ~~ + W_{t, \Delta}  Y_{t, \Delta}(1-A_t) \Big] f_t(S_t). 
\end{align*}

To derive $\UGR(\beta, \alpha, \mu)$, we subtract from $U^\text{GR-prelim}(\beta, \alpha)$ its projection onto the score of the treatment assignment:
\begin{align}
	&\sum_{t=1}^T \Bigg(U_t^\text{GR-prelim}(\beta, \alpha) - \EE\{ U_t^\text{GR-prelim}(\beta, \alpha)|H_t, A_t\} + \EE\{ U_t^\text{GR-prelim}(\beta, \alpha)|H_t\} \Bigg) \label{eq: UGR first part expanded}\\
	&-\sum_{t=1}^T \Bigg(\sum_{1\leq u \leq T, u\neq t} \Big[ \EE\{ U_t^\text{GR-prelim}(\beta, \alpha)|H_u, A_u\} - \EE\{ U_t^\text{GR-prelim}(\beta, \alpha)|H_u\} \Big] \Bigg). \label{eq: UGR second part expanded}
\end{align}
Because the terms in \eqref{eq: UGR second part expanded} cannot be analytically derived without additional modeling, We will focus on calculating \eqref{eq: UGR first part expanded}. We have
\begin{align*}
	&\EE\{ U_t^\text{GR-prelim}(\beta, \alpha)|H_t, A_t\} \\
	=& \frac{A_t - p_t(H_t)}{p_t(H_t)\{1 - p_t(H_t)\}} \omega(t) I_t\Big[ \expit\{g(S_t)^T\alpha - f_t(S_t)^T\beta\} A_t + \EE(W_{t, \Delta}  Y_{t, \Delta}|H_t, A_t) (1-A_t) \Big] f_t(S_t) \\ 
	=& \frac{A_t - p_t(H_t)}{p_t(H_t)\{1 - p_t(H_t)\}} \omega(t) I_t\Big[ \expit\{g(S_t)^T\alpha - f_t(S_t)^T\beta\} A_t + \EE(W_{t, \Delta}  Y_{t, \Delta}|H_t, A_t = 0) (1-A_t) \Big] f_t(S_t). \numberthis \label{eq: first first part of EEUGR}
\end{align*}

We also have,
\begin{align*}
	&\EE\{ U_t^\text{GR-prelim}(\beta, \alpha)|H_t\} \\
	=& \EE\{ U_t^\text{GR-prelim}(\beta, \alpha)|H_t, I_t = 1\}I_t + \EE\{ U_t^\text{GR-prelim}(\beta, \alpha)|H_t, I_t = 0\}(1-I_t) \\
	= & \EE\{ U_t^\text{GR-prelim}(\beta, \alpha)|H_t, I_t = 1, A_t = 1\}P(A_t = 1|H_t, I_t = 1)I_t \\ 
	&+ \EE\{ U_t^\text{GR-prelim}(\beta, \alpha)|H_t, I_t = 1, A_t = 0\}P(A_t = 0|H_t, I_t = 1)I_t \\ 
	&+ \EE\{ U_t^\text{GR-prelim}(\beta, \alpha)|H_t, I_t = 0\}(1-I_t) \\
	= & \frac{1 - p_t(H_t)}{p_t(H_t)\{1 - p_t(H_t)\}} \expit\{g(S_t)^T\alpha - f_t(S_t)^T\beta\} f_t(S_t) p_t(H_t) I_t \omega(t) \\ 
	&+ \frac{ - p_t(H_t)}{p_t(H_t)\{1 - p_t(H_t)\}} \EE(W_{t, \Delta}  Y_{t, \Delta}|H_t, A_t = 0, I_t = 1) f_t(S_t) \{1 - p_t(H_t)\} I_t  \omega(t)\\ 
	= & \omega(t)I_t\Big[\expit\{g(S_t)^T\alpha - f_t(S_t)^T\beta\} - \EE(W_{t, \Delta}  Y_{t, \Delta}|H_t, A_t = 0, I_t = 1)\Big] f_t(S_t)I_t. \numberthis \label{eq: first second part of EEUGR}
\end{align*}
Combing \eqref{eq: first first part of EEUGR} and \eqref{eq: first second part of EEUGR}, we have 
\begin{align*}
	&U_t^\text{GR-prelim}(\beta, \alpha) - \EE\{ U_t^\text{GR-prelim}(\beta, \alpha)|H_t, A_t\} + \EE\{ U_t^\text{GR-prelim}(\beta, \alpha)|H_t\}  \\ 
	=& \frac{A_t - p_t(H_t)}{p_t(H_t)\{1 - p_t(H_t)\}} \omega(t) I_t \Big[ (W_{t, \Delta}  Y_{t, \Delta} - \mu_{0t}^\star) (1-A_t) \Big] f_t(S_t) \\ 
	&+  \omega(t) I_t \Big[\expit\{g(S_t)^T\alpha - f_t(S_t)^T\beta\} - \mu_{0t}^\star \Big] f_t(S_t). \\ 
		= &\omega(t)I_t \Bigg[\frac{\{A_t - p_t(H_t)\} (1-A_t)}{p_t(H_t)\{1 - p_t(H_t)\}}(W_{t, \Delta}  Y_{t, \Delta} - \mu_{0t}^\star) + \expit\{g(S_t)^T\alpha - f_t(S_t)^T\beta\} - \mu_{0t}^\star \Bigg] f_t(S_t). \\ 
	\numberthis \label{eq: UGR last equation}
\end{align*}
Since $A_t$ is binary, we can further simplify
\begin{align*}
 \frac{\{A_t - p_t(H_t)\} (1-A_t)}{p_t(H_t)\{1 - p_t(H_t)\}} = 
 \begin{cases}
0 \quad \quad \quad \quad \quad \ \   \text{when } A_t = 1 \\ -\frac{1}{1 - p_t(H_t)} \quad \quad \text{when } A_t = 0
 \end{cases}
\end{align*}
Thus, we have \eqref{eq: UGR last equation} equals to
\begin{align*}
	&U_t^\text{GR-prelim}(\beta, \alpha) - \EE\{ U_t^\text{GR-prelim}(\beta, \alpha)|H_t, A_t\} + \EE\{ U_t^\text{GR-prelim}(\beta, \alpha)|H_t\} \\ 
	= &\omega(t) I_t \Bigg[\expit\{g(S_t)^T\alpha - f_t(S_t)^T\beta\} - \mu_{0t}^\star -\frac{1 -A_t}{1 - p_t(H_t)}(W_{t, \Delta}  Y_{t, \Delta} - \mu_{0t}^\star) \Bigg] f_t(S_t).
\end{align*}
This completes the derivation.
\section{Proof of Theorem 1}
\label{supp-sec:proof_of_consistency_and_asymptotic_normality_of_method_1}

\subsection{Setup and Notations}
Consider the setting where the randomization probability $p_{t}(H_{t})$
depends at most on $S_{t}$, i.e., $P(A_{t}\mid H_{t},I_{t}=1)=p_{t}(S_{t})$. Let $\eta=(\eta_{1},...,\eta_{T})$ with $\eta_{t}=(r_{t},m_{t},\mu_{t})$, and $\eta'_{t}=(r'_{t},m'_{t},\mu'_{t})$ denote the limits of the fitted nuisance functions.
Define $U_t^\text{SR}(\beta,\eta_t)$ to be a summand of $U^\text{SR}(\beta,\eta)$ such that $U^\text{SR}(\beta,\eta) = \sum_{t=1}^T U_t^\text{SR}(\beta,\eta_t)$, and 
\begin{align}
U_t^\text{SR}(\beta,\eta_t)
= &  \omega(t) I_{t}\bigg[(W_{t,\Delta}Y_{t,\Delta}-\mu_{t})\Big\{ e^{-A_{t}f_{t}(S_{t})^{T}\beta}+e^{r_{t}(S_{t})}\Big\}\{A_{t}-m_{t}(S_{t})\}\nonumber \\
 & \ \ +\Big\{\mu_{1t}e^{-f_{t}(S_{t})^{T}\beta}-(1-\mu_{1t})e^{r_{t}(S_{t})}\Big\}\{1-m_{t}(S_{t})\}p_{t}(S_{t})\nonumber \\
 & -\Big\{\mu_{0t}-(1-\mu_{0t})e^{r_{t}(S_{t})}\Big\} m_{t}(S_{t})\{1-p_{t}(S_{t})\}\bigg]f_{t}(S_{t}).\label{eq: EE of method1 with improved efficiency}
\end{align}
 For simplicity of notation, let $\text{expit}(a)$ denote $1/(1+e^{-a})$.

\subsection{Assumptions and Regularity Conditions}
We state the necessary regularity conditions for establishing Theorem 1 in the main paper. 
\begin{asu}
\label{assu: unique zero}(Unique zero). There exists unique $\beta\in\Theta$
such that $\PP\{\USR(\beta,\eta')\}=0.$
\end{asu}
\begin{asu}
	\label{assu: Regularity conditions}(Regularity conditions)
	\begin{asulist}
    \item 
    \label{assu: Regularity conditions: support compact space} Suppose the parameter space $\Theta$ of $\beta$ is
compact.
    \item 
    \label{assu: (Regularity conditions: support of data is bounded)} Suppose the support of $O$ is bounded.
		\item 
		\label{assu: partialU is invertiable} Suppose
$\PP\{\partial_{\beta}\USR(\beta^{\star},\eta')\}$ is invertible.
		\item 
		\label{assu:(Regularity-conditions) mt rt are between 0 and 1} For each $t$, $\hat{r}_{t},r_{t}',\hat{m}_{t}$, $m'_{t}$
$\in[\epsilon,1-\epsilon]$ for $\epsilon>0$.
		\item 
		\label{assu:(Regularity-conditions) derivative is uniformly bounded} $\partial_{\beta}\USR(\beta,\eta)$ and $\USR(\beta,\eta)\USR(\beta,\eta)^{T}$
are uniformaly bounded by integrable function. 
\end{asulist}
\end{asu}

\begin{asu}
\label{assu:(Donsker-condition)}(Donsker condition) Suppose $\{\USR(\beta,\eta):\beta\in\Theta,\eta\in T\}$ and $\{\partial_{\beta}\USR(\beta,\eta):\beta\in\Theta,\eta\in T\}$
are P-Donsker classes.
\end{asu}

\subsection{Lemmas}
\begin{lem}
\label{lem: PU =00003D 0} Suppose either $r'_t = r^\star_t$ or $m'_t = m^\star_t$, then $\PP\{U^\text{SR-prelim}_{t}(\beta^{\star},r'_{t},m'_{t})\}=0$ for $t\in[T]$.
\end{lem}
\begin{proof}
The defintion of $\beta^\star$ in the simple randomization setting and the defintion of the nuisance function $r_t^\star(S_t)$ implies that
\begin{align*}
	\PP \Bigg(\sum_{t=1}^T \omega(t) \Big[\logit\{\PP(W_{t,\Delta}Y_{t,\Delta}|S_t, A_t = 1, I_t = 1\} - r_t^\star(S_t) - f_t(S_t)^T \beta^\star \Big] \Bigg) = 0. \numberthis \label{eq: the USR causal model with r and beta}
\end{align*}

We have
\begin{align}
&\PP\{U^\text{SR-prelim}(\beta,r',m')\} \nonumber \\
= & \sum_{t=1}^T \PP\Big[\omega(t)I_{t}\{W_{t, \Delta}Y_{t,\Delta}e^{-A_{t}f_{t}(S_{t})^{T}\beta}-(1-W_{t, \Delta}Y_{t,\Delta})e^{r'_{t}(S_{t})}\}\{A_{t}-m'_{t}(S_{t})\}f_{t}(S_{t})\Big]\nonumber \\
= & \sum_{t=1}^T \PP\Big[\omega(t)I_{t}e^{r_{t}'(S_{t})}\{W_{t, \Delta}Y_{t,\Delta}e^{-A_{t}f_{t}(S_{t})^{T}\beta-r_{t}'(S_{t})}-(1-W_{t, \Delta}Y_{t,\Delta})\}\{A_{t}-m'_{t}(S_{t})\}f_{t}(S_{t})\Big]\\
= & \sum_{t=1}^T \PP\Bigg\{  e^{r_{t}'(S_{t})} \PP\Big[\omega(t)I_{t} \{W_{t, \Delta}Y_{t,\Delta}e^{-A_{t}f_{t}(S_{t})^{T}\beta-r_{t}'(S_{t})}-(1-W_{t, \Delta}Y_{t,\Delta})\}|S_{t},A_{t}\Big]\nonumber \\
 & \ \ \ \ \ \ \ \ \ \ \ \ \times \{A_{t}-m'_{t}(S_{t})\}f_{t}(S_{t}) \Bigg\}\nonumber \\
\end{align}
where the last equality holds because of iterated expectation. Replacing $\beta$ with $\beta^\star$, we have
\begin{align}
& \PP\{U^\text{SR-prelim}(\beta^\star,r',m')\} \nonumber \\
 = & \sum_{t=1}^T \PP\Bigg\{  e^{r_{t}'(S_{t})} \PP\Big[\omega(t) I_{t} \{(1 - W_{t, \Delta}Y_{t, \Delta}) e^{r_t^\star(S_t) -r'_t(S_t)}-(1-W_{t, \Delta}Y_{t,\Delta})\}|S_{t},A_{t}\Big]\nonumber \\
 & \ \ \ \ \ \ \ \ \ \ \ \ \times \{A_{t}-m'_{t}(S_{t})\}f_{t}(S_{t}) \Bigg\} \label{eq: inside first lemma equation of expit} \\
 =&\sum_{t=1}^T \PP\Bigg\{  e^{r_{t}'(S_{t})} \PP\Big[\omega(t) I_{t} (1 - W_{t, \Delta}Y_{t, \Delta})  \{ e^{r_t^\star(S_t) -r'_t(S_t)}-1\}|S_{t},A_{t}\Big]\nonumber \\
 & \ \ \ \ \ \ \ \ \ \ \ \ \times \{A_{t}-m'_{t}(S_{t})\}f_{t}(S_{t}) \Bigg\} \label{eq: inside first lemma equation of dr} 
\end{align}
where \eqref{eq: inside first lemma equation of expit} holds because of \eqref{eq: the USR causal model with r and beta}.
Thus, \eqref{eq: inside first lemma equation of dr} equals 0 when either $r'_t(S_t) = r_t(S_t)^\star$ or $m'_t(S_t) = m_t(S_t)^\star$. The completes the proof.
\end{proof}

\begin{lem}
\label{lem: EE =00003D 0}Suppose either $r'_t = r^\star_t$ or $m'_t = m^\star_t$, then $\PP\{\USR(\beta^{\star},\eta')\}=0$ for $t\in[T]$.
\end{lem}
\begin{proof}
We can rewrite $\USR(\beta,\eta)$ as
\begin{align}
\USR(\beta,\eta) & =\sum_{t}\{U_{t}^\text{SR-prelim}(\beta,r_{t},m_{t})-H_{t}(\beta,\eta_{t})+K_{t}(\beta,\eta_{t})\},\label{eq: lemma on Ptilde_U =00003D 0 rewrite tilde_U}
\end{align}
where 
\begin{align*}
H_{t}(\beta,\eta_{t}) & =\omega(t)I_{t}\{\mu_{t}e^{-A_{t}f_{t}(S_{t})^{T}\beta}-(1-\mu_{t})e^{r_{t}(S_{t})}\}\{A_{t}-m_{t}(S_{t})\}f_{t}(S_{t})\\
K_{t}(\beta,\eta_{t}) & =\omega(t)I_{t}\{\mu_{1t}e^{-f_{t}(S_{t})^{T}\beta}-(1-\mu_{1t})e^{r_{t}(S_{t})}\}\{1-m_{t}(S_{t})\}f_{t}(S_{t})p_{t}(S_{t})\\
 & \ \ -\omega(t)I_{t}\{\mu_{0t}-(1-\mu_{0t})e^{r_{t}(S_{t})}\}m_{t}(S_{t})f_{t}(S_{t})\{1-p_{t}(S_{t})\}
\end{align*}
By Lemma \ref{lem: PU =00003D 0}, $\PP\{U_{t}^\text{SR-prelim}(\beta^{\star},r'_{t},m'_{t})\}=0$.
Thus, it suffices to show that $\PP\sum_{t}\{H_{t}(\beta,\eta_{t}')-K_{t}(\beta,\eta_{t}')\}=0$.
By iterated expectation, 
\begin{align*}
 & \PP\sum_{t}\{H_{t}(\beta,\eta_{t}')-K_{t}(\beta,\eta_{t}')\}\\
= & \sum_{t}\PP\Big[\PP\{H_{t}(\beta,\eta_{t}')-K_{t}(\beta,\eta_{t}')|H_{t},A_{t}\}\Big]\\
= & \sum_{t}\PP\Big[\PP\{H_{t}(\beta,\eta_{t}')|H_{t},A_{t}=1\}p_{t}(S_{t})+\PP\{H_{t}(\beta,\eta_{t}')|H_{t},A_{t}=1\}\{1-p(S_{t})\}\Big]\\
 & -\sum_{t}\PP\{K_{t}(\beta,\eta_{t}')\}\\
= & 0.
\end{align*}
The conclusion of the lemma follows.
\end{proof}

\begin{lem}
\label{lem: sup|PU - PU| =00003D o_p(1)}Suppose Assumptions
\ref{assu: Regularity conditions: support compact space}, \ref{assu: (Regularity conditions: support of data is bounded)},
and \ref{assu:(Regularity-conditions) mt rt are between 0 and 1};
hold. Then for $\USR(\beta,\eta)$ defined in (\ref{eq: EE of method1 with improved efficiency}),
$\sup_{\beta\in\Theta}|\PP\{\USR(\beta,\hat{\eta})\}-\PP\{\USR(\beta,\eta')\}|=o_{p}(1)$.
\end{lem}
\begin{proof}
Let 
\begin{align*}
g_{1t}(r_{t},m_{t}) & =\omega(t)I_{t}\Big[\PP(W_{t, \Delta}Y_{t, \Delta}|H_{t,}A_{t}=1)e^{r_{t}(S_{t})}\{1-m_{t}(S_{t})\}p_{t}(S_{t})\\
 & \ \ +\PP(W_{t, \Delta}Y_{t, \Delta}|H_{t,}A_{t}=0)\{1+e^{r_{t}(S_{t})}\}\{-m_{t}(S_{t})\}\{1+p_{t}(S_{t})\}\Big],\\
g_{2t}(m_{t}) & =\omega(t)I_{t}\PP(W_{t, \Delta}Y_{t, \Delta}|H_{t,}A_{t}=1)\{1-m_{t}(S_{t})\}p_{t}(S_{t}).
\end{align*}
Then, $\PP\{\USR(\beta,\eta)\}$ can be rewritten as 
\[
\PP\{\USR(\beta,\eta)\}=\PP\Big[\sum_{t}\Big\{ g_{1t}(r_{t},m_{t})+e^{f_{t}(S_{t})^{T}\beta}g_{2t}(m_{t})\Big\}\Big].
\]
Therefore, we have
\begin{align}
 & \sup_{\beta\in\Theta}\Big|\PP\{\USR(\beta,\hat{\eta})\}-\PP\{\USR(\beta,\eta')\}\Big|\nonumber \\
= & \sup_{\beta\in\Theta}\Bigg|\PP\Big[\sum_{t}\Big\{ g_{1t}(\hat{r}_{t},\hat{m}_{t})+e^{-f_{t}(S_{t})^{T}\beta}g_{2t}(\hat{m}_{t})\Big\}\Big]\nonumber -\PP\Big[\sum_{t}\Big\{ g_{1t}(r_{t}',m_{t}')+e^{-f_{t}(S_{t})^{T}\beta}g_{2t}(m_{t}')\Big\}\Big]\Bigg|\nonumber \\
= & \sum_{t}\sup_{\beta\in\Theta}\Bigg|\PP\Big[e^{-f_{t}(S_{t})^{T}\beta}\Big\{ g_{2t}(\hat{m}_{t})-g_{2t}(m_{t}')\Big\}  -\Big\{ g_{1t}(\hat{r}_{t},\hat{m}_{t})-g_{1t}(r_{t}',m_{t}')\Big\}\Big]\nonumber \\
\leq & \sup_{\beta\in\Theta}\Bigg\{\Bigg|\PP\Big[e^{-f_{t}(S_{t})^{T}\beta}\Big\{ g_{2t}(\hat{m}_{t})-g_{2t}(m_{t}')\Big\}\Big]\Bigg|  +\Bigg|\PP\Big\{ g_{1t}(\hat{r}_{t},\hat{m}_{t})-g_{1t}(r_{t}',m_{t}')\Big\}\Bigg|\Bigg\},\label{eq:in PU - PU =00003D o_p(1) lemma whole equations}
\end{align}
where the last line follows from triangle inequality. It suffices
to show that each of the term in (\ref{eq:in PU - PU =00003D o_p(1) lemma whole equations})
is $o_{p}(1)$.

To control the first term in (\ref{eq:in PU - PU =00003D o_p(1) lemma whole equations}),
Cauchy-Schwartz inequality yields
\begin{align*}
 & \sup_{\beta\in\Theta}\Bigg|\PP\Big[e^{-f_{t}(S_{t})^{T}\beta}\Big\{ g_{2t}(\hat{m}_{t})-g_{2t}(m_{t}')\Big\}\Big]\Bigg|\\
\leq & \sup_{\beta\in\Theta}\lVert e^{-f_{t}(S_{t})^{T}\beta}\rVert\lVert g_{2t}(\hat{m}_{t})-g_{2t}(m_{t}')\lVert,
\end{align*}
and $\sup_{\beta\in\Theta}\lVert e^{-f_{t}(S_{t})^{T}\beta}\rVert$
is bounded because of Assumption \ref{assu: Regularity conditions: support compact space}
and \ref{assu: (Regularity conditions: support of data is bounded)}.
To control $\lVert g_{2t}(\hat{m}_{t})-g_{2t}(m_{t}')\lVert$, we
have
\begin{align}
 & \lVert g_{2t}(\hat{m}_{t})-g_{2t}(m_{t}')\lVert^{2}\nonumber \\
= & \int\big\{ g_{2t}(\hat{m}_{t})-g_{2t}(m_{t}')\big\}^{2}dP\nonumber \\
\leq & C\int(\hat{m}_{t}-m'_{t})^{2}dP\label{eq: in lemma sup PU - PU =00003D o_p(1) form of g}\\
= & o_{p}(1),\label{eq:in lemma sup PU - PU =00003D o_p(1) convergence of delta}
\end{align}
where (\ref{eq: in lemma sup PU - PU =00003D o_p(1) form of g}) follows
from the form of $g_{2t}$, and Assumptions \ref{assu: (Regularity conditions: support of data is bounded)}
and \ref{assu:(Regularity-conditions) mt rt are between 0 and 1};
(\ref{eq:in lemma sup PU - PU =00003D o_p(1) convergence of delta})
follows from the limit of $\hat{m}_{t}$.

To control the second term in (\ref{eq:in PU - PU =00003D o_p(1) lemma whole equations}),
we have
\begin{align}
 & \Bigg|\PP\Big\{ g_{1t}(\hat{r}_{t},\hat{m}_{t})-g_{1t}(r_{t}',m_{t}')\Big\}\Bigg|\nonumber \\
\le & \int\Big|g_{1t}(\hat{r}_{t},\hat{m}_{t})-g_{1t}(r_{t}',m_{t}')\Big|dP\nonumber \\
\leq & C\int|\hat{\eta}-\eta'|dP\label{eq: in lemma sup PU - PU =00003D op(1) bound of g1t}\\
\leq & (C\int|\hat{\eta}-\eta'|^{2}dP)^{1/2}\label{eq:in lemma sup PU - PU =00003D op(1) cauchy schwartz to get square}\\
= & o_{p}(1),\label{eq:in lemma sup PU - PU =00003D op(1) convergence of delta again}
\end{align}
where (\ref{eq: in lemma sup PU - PU =00003D op(1) bound of g1t})
follows from the form of $g_{1t}$, and Assumption \ref{assu: (Regularity conditions: support of data is bounded)}
and \ref{assu:(Regularity-conditions) mt rt are between 0 and 1};
(\ref{eq:in lemma sup PU - PU =00003D op(1) cauchy schwartz to get square})
follows from Cauchy-Schwartz inequality; and (\ref{eq:in lemma sup PU - PU =00003D op(1) convergence of delta again})
holds because of limit of $\hat{\eta}$. The conclusion of
the lemma follows.
\end{proof}

\begin{lem}
\label{lem:Consistency of =00005Chat=00007B=00005Cbeta=00007D}(Consistency
of $\hat{\beta}$). Suppose Assumptions
\ref{assu: unique zero}, \ref{assu: Regularity conditions: support compact space},
and \ref{assu:(Donsker-condition)} hold. Suppose either $r'_t = r^\star_t$ or $m'_t = m^\star_t$, then $\hat{\beta}\pto\beta^{\star}$
as $n\rightarrow\infty$.
\end{lem}
\begin{proof}
For arbitray $\epsilon>0$, we would like to show $\lim_{n\rightarrow\infty}P(|\hat{\beta}-\beta^{\star}|>\epsilon)=0$. For any $\epsilon > 0$, consider $C:=\inf_{\beta\in\Theta:d(\beta,\beta^{\star})>\epsilon}|\PP\{\USR(\beta,\eta')\}|$.
Because the parameter space $\Theta$ is compact (Assumption \ref{assu: Regularity conditions: support compact space})
and the fact that $\PP\{\USR(\beta,\eta')\}$ is a continuous function
in $\beta$ , the infimum of $\PP\{\USR(\beta,\eta')\}$ is
attained, i.e. $\inf_{\beta\in\Theta:d(\beta,\beta^{\star})>\epsilon}|\PP\{\USR(\beta,\eta')\}|=\min_{\beta\in\Theta}|\PP\{\USR(\beta,\eta')\}|$.
Since $\beta^{\star}$ is the unique zero of $\PP\{\USR(\beta,\eta')\}$
(Assumption \ref{assu: unique zero} and Lemma \ref{lem: EE =00003D 0}), we have 
$\min_{\beta\in\Theta}|\PP\{\USR(\beta,\eta')\}|>0$. By onstructing
such C, we obtained that for any $\epsilon$, $|\beta-\beta^{\star}|>\epsilon$
implies $\PP\{\USR(\hat{\beta},\eta')\}>C$. Take $\beta=\hat{\beta}$,
we have 
$$
P(|\hat{\beta}-\beta^{\star}|>\epsilon)\leq P\big[\PP\{\USR(\hat{\beta},\eta')\} > C \big].
$$
Therefore, it suffices to show that $P\big[\PP\{\USR(\hat{\beta},\eta')\}\big]$ converges in probability to 0.

We first decompose 
\begin{align}
|\PP\{\USR(\hat{\beta},\eta')\}\}\leq|\PP\{\USR(\hat{\beta},\eta')\}-\PP\{\USR(\hat{\beta},\hat{\eta})\}|+|\PP\{\USR(\hat{\beta},\hat{\eta})\}-\PP_{n}\{\USR(\hat{\beta},\hat{\eta})\}|,\label{eq:in lemm consistency two terms}
\end{align}
where the inequality holds by triangle inequality. Next, we want to
show that both terms in (\ref{eq:in lemm consistency two terms})
are $o_{p}(1)$. For the first term in (\ref{eq:in lemm consistency two terms}),
by Lemma \ref{lem: sup|PU - PU| =00003D o_p(1)}, we have 
\[
|\PP\{\USR(\hat{\beta},\eta')\}-\PP\{\USR(\hat{\beta},\hat{\eta})\}|\leq\sup_{\beta\in\Theta}|\PP\{\USR(\beta,\eta')\}-\PP\{\USR(\beta,\hat{\eta})\}|=o_{p}(1).
\]
By Assumption \ref{assu:(Donsker-condition)}, $\{\USR(\beta,\eta):\beta\in\Theta,\eta\in\Tau\}$
is a Glivenko-Cantelli class. Thus, the second term in (\ref{eq:in lemm consistency two terms})
is 
\begin{align*}
|\PP\{\USR(\hat{\beta},\hat{\eta})\}-\PP_{n}\{\USR(\hat{\beta},\hat{\eta})\}| & \leq\sup_{\beta\in\Theta,\eta\in\Tau}|(\PP_{n}-\PP)\{\USR(\beta,\eta)\}|\\
 & =o_{p}(1).
\end{align*}
Therefore, $\PP\{\USR(\hat{\beta},\eta')\}$ converges in probability to 0 and conclusion of the lemma follows.
\end{proof}

\begin{lem}
\label{lem: Convergence of derivative}(Convergence of derivative).
Suppose Assumptions
\ref{assu: Regularity conditions: support compact space}, \ref{assu: (Regularity conditions: support of data is bounded)},
\ref{assu:(Donsker-condition)} , and \ref{assu:(Regularity-conditions) derivative is uniformly bounded}
hold. If $\hat{\beta}\pto\beta^{\star}$, then $\ensuremath{\PP_{n}\{\partial_{\beta}\USR(\hat{\beta},\hat{\eta})\}\pto\PP\{\partial_{\beta}\USR(\beta^{\star},\eta')\}}$
as $\ensuremath{n\rightarrow\infty}$.
\end{lem}
\begin{proof}
We have
\begin{align}
 & \ \PP_{n}\{\partial_{\beta}\USR(\hat{\beta},\hat{\eta})\}-\PP\{\partial_{\beta}\USR(\beta^{\star},\eta')\}\nonumber \\
 & \leq\sup_{\beta\in\Theta,\eta\in\Tau}(\PP_{n}-\PP)\{\partial_{\beta}\USR(\hat{\beta},\hat{\eta})\}\nonumber \\
 & \ \ \ \ +\Big[\PP\{\partial_{\beta}\USR(\hat{\beta},\hat{\eta})\}-\PP\{\partial_{\beta}\USR(\beta^{\star},\hat{\eta})\}\Big]+\Big[\PP\{\partial_{\beta}\USR(\beta^{\star},\hat{\eta})\}-\PP\{\partial_{\beta}\USR(\beta^{\star},\eta')\}\Big],\label{eq:in lemma convergence of derivative three terms}
\end{align}
To show that $|\PP_{n}\{\partial_{\beta}\USR(\hat{\beta},\hat{\eta})\}-\PP\{\partial_{\beta}\USR(\beta^{\star},\eta')\}|=o_{p}(1)$,
it suffices to show that each of the three terms in (\ref{eq:in lemma convergence of derivative three terms})
is $o_{p}(1)$.

For the first term in (\ref{eq:in lemma convergence of derivative three terms}),
Assumptino \ref{assu:(Donsker-condition)} implies that $\{\partial_{\beta}\USR(\beta,\eta):\beta\in\Theta,\eta\in\Tau\}$
is a Glivenko-Cantelli class. Thus, $\sup_{\beta\in\Theta,\eta\in\Tau}(\PP_{n}-\PP)\{\partial_{\beta}\USR(\hat{\beta},\hat{\eta})\}=o_{p}(1)$.

For the second term in (\ref{eq:in lemma convergence of derivative three terms}),
by the fact that $\hat{\beta}\pto\beta^{\star},$ dominatedness of
$\partial_{\beta}\USR(\beta,\eta)$ (Assumption\ref{assu:(Regularity-conditions) derivative is uniformly bounded}),
and dominated convergence theorem (\citet{chung2001course} Result (viii) of Chapter 3.2), we have 
\[
|\PP\{\partial_{\beta}\USR(\hat{\beta},\hat{\eta})\}-\PP\{\partial_{\beta}\USR(\beta^{\star},\hat{\eta})\}|=o_{p}(1).
\]

For the third term in (\ref{eq:in lemma convergence of derivative three terms}),
based on the form of $\USR(\beta,\eta)$ and Assumption \ref{assu: Regularity conditions: support compact space},
\ref{assu: (Regularity conditions: support of data is bounded)},
we have
\begin{align*}
\PP\{\partial_{\beta}\USR(\beta^{\star},\hat{\eta})\}-\PP\{\partial_{\beta}\USR(\beta^{\star},\eta')\} & \leq\int\Big[\PP\{\partial_{\beta}\USR(\beta^{\star},\hat{\eta})\}-\PP\{\partial_{\beta}\USR(\beta^{\star},\eta')\}\Big]dP\\
 & \leq C\int|\hat{\eta}-\eta'|dP\\
 & =o_{p}(1).
\end{align*}
This completes the proof.
\end{proof}

\begin{lem}
\label{lem:(Convergence-of-the meat term)}(Convergence of the \textquotedblleft meat\textquotedblright{}
term). Suppose Assumptions
\ref{assu:(Donsker-condition)} and \ref{assu:(Regularity-conditions) derivative is uniformly bounded}
hold. If $\hat{\beta}\pto\beta^{\star}$, then $\PP_{n}\{\USR(\hat{\beta},\hat{\eta})\USR(\hat{\beta},\hat{\eta})^{T}\}\pto\PP\{\USR(\beta^{\star},\eta')\USR(\beta^{\star},\eta')^{T}\}$. 
\end{lem}
\begin{proof}
We have
\begin{align}
 & \Big|\PP_{n}\{\USR(\hat{\beta},\hat{\eta})\USR(\hat{\beta},\hat{\eta})^{T}\}-\PP\{\USR(\beta^{\star},\eta')\USR(\beta^{\star},\eta')^{T}\}\Big|\nonumber \\
\leq & \sup_{\beta\in\Theta,\eta\in\Tau}|(\PP_{n}-\PP)\{\USR(\hat{\beta},\hat{\eta})\USR(\hat{\beta},\hat{\eta})^{T}\}|\nonumber \\
 & \ \ \ +\Big|\PP\{\USR(\hat{\beta},\hat{\eta})\USR(\hat{\beta},\hat{\eta})^{T}\}-\PP\{\USR(\beta^{\star},\eta')\USR(\beta^{\star},\eta')^{T}\}\Big|.\label{eq: in lemma convergence of meat term two terms}
\end{align}
To show $\PP_{n}\{\USR(\hat{\beta},\hat{\eta})\USR(\hat{\beta},\hat{\eta})^{T}\}\pto\PP\{\USR(\beta^{\star},\eta')\USR(\beta^{\star},\eta')^{T}\}$,
it suffices to show that each of the two terms in (\ref{eq: in lemma convergence of meat term two terms})
is $o_{p}(1)$.

For the first term in (\ref{eq: in lemma convergence of meat term two terms}),
by Assumption \ref{assu:(Donsker-condition)}, $\{\USR(\beta,\eta)\USR(\beta,\eta)^{T}:\beta\in\Theta,\eta\in\Tau\}$
is a Glivenko-Cantelli class. Thus, $\sup_{\beta\in\Theta,\eta\in\Tau}|(\PP_{n}-\PP)\{\USR(\hat{\beta},\hat{\eta})\USR(\hat{\beta},\hat{\eta})^{T}\}|=o_{p}(1)$.

For the second term in (\ref{eq: in lemma convergence of meat term two terms}),
since $(\hat{\beta},\hat{\eta})\pto(\beta^{\star},\eta'$) and dominatedness of $\USR(\beta,\eta)\USR(\beta,\eta)^{T}$
(Assumption \ref{assu:(Regularity-conditions) derivative is uniformly bounded}),
\[
\Big|\PP\{\USR(\hat{\beta},\hat{\eta})\USR(\hat{\beta},\hat{\eta})^{T}\}-\PP\{\USR(\beta^{\star},\eta')\USR(\beta^{\star},\eta')^{T}\}\Big|=o_{p}(1)
\]
by dominated convergence theorem (\citet{chung2001course} Result (viii) of Chapter 3.2). This completes the proof.
\end{proof}

\begin{lem}
\label{lem: bound of U(beta_star, delta_hat) and U(beta_star, delta')}Suppose
Assumptions
\ref{assu: Regularity conditions: support compact space}, \ref{assu: (Regularity conditions: support of data is bounded)}, and 
\ref{assu:(Regularity-conditions) mt rt are between 0 and 1} hold.
Then for $\USR(\beta,\eta)$ defined in $\eqref{eq: EE of method1 with improved efficiency}$,
we have $\lVert\USR(\beta^{\star},\hat{\eta})-\USR(\beta^{\star},\eta')\rVert^{2}=o_{p}(1)$.
\end{lem}
\begin{proof}
We have 
\begin{align*}
 & \lVert\USR(\beta^{\star},\hat{\eta})-\USR(\beta^{\star},\eta')\rVert^{2}\\
= & \int|\USR(\beta^{\star},\hat{\eta})-\USR(\beta^{\star},\eta')|^{2}dP\\
= & \int\Big|\sum_{t}\Big\{\USR_{t}(\beta^{\star},\hat{r}_{t},\hat{m}_{t},\hat{\mu}_{t,})-\USR_{t}(\beta^{\star},r'_{t},\hat{m}_{t},\hat{\mu}_{t,})+\USR_{t}(\beta^{\star},r'_{t},\hat{m}_{t},\hat{\mu}_{t,})-\USR_{t}(\beta^{\star},r'_{t},m'_{t},\hat{\mu}_{t,})\\
 & \ \ +\USR_{t}(\beta^{\star},r'_{t},m'_{t},\hat{\mu}_{t,})-\USR_{t}(\beta^{\star},r'_{t},m'_{t},\mu'_{t})\Big\}\Big|^{2}dP\\
\leq & 2T^{2}\Bigg\{\max_{1\leq t\leq T}\int|\USR_{t}(\beta^{\star},\hat{r}_{t},\hat{m}_{t},\hat{\mu}_{t,})-\USR_{t}(\beta^{\star},r'_{t},\hat{m}_{t},\hat{\mu}_{t,})|^{2}dP\\
 & \ \ \ +\max_{1\leq t\leq T}\int|\USR_{t}(\beta^{\star},r'_{t},\hat{m}_{t},\hat{\mu}_{t,})-\USR_{t}(\beta^{\star},r'_{t},m'_{t},\hat{\mu}_{t,})|^{2}dP\\
 & \ \ \ +\max_{1\leq t\leq T}\int|\USR_{t}(\beta^{\star},r'_{t},m'_{t},\hat{\mu}_{t,})-\USR_{t}(\beta^{\star},r'_{t},m'_{t},\mu'_{t})\Big\}\Big|^{2}dP\Bigg\}
\end{align*}
where the last inequality follows from (\citet{cheng2023efficient} 
Lemma B2). Thus, it suffices to show that for all $t$, 
\begin{align}
\int|\USR_{t}(\beta^{\star},\hat{r}_{t},\hat{m}_{t},\hat{\mu}_{t,})-\USR_{t}(\beta^{\star},r'_{t},\hat{m}_{t},\hat{\mu}_{t,})|^{2}dP & =o_{p}(1)\label{eq: in Lemma U - U is op(1) three term r_t}\\
\int|\USR_{t}(\beta^{\star},r'_{t},\hat{m}_{t},\hat{\mu}_{t,})-\USR_{t}(\beta^{\star},r'_{t},m'_{t},\hat{\mu}_{t,})|^{2}dP & =o_{p}(1)\label{eq: in Lemma U - U is op(1) three term m_t}\\
\int|\USR_{t}(\beta^{\star},r'_{t},m'_{t},\hat{\mu}_{t,})-\USR_{t}(\beta^{\star},r'_{t},m'_{t},\mu'_{t})\Big\}\Big|^{2}dP & =o_{p}(1).\label{eq:in Lemma U - U is op(1) three term mu_t}
\end{align}

To show (\ref{eq: in Lemma U - U is op(1) three term r_t}), we have
\begin{align}
 & \int|\USR_{t}(\beta^{\star},\hat{r}_{t},\hat{m}_{t},\hat{\mu}_{t,})-\USR_{t}(\beta^{\star},r'_{t},\hat{m}_{t},\hat{\mu}_{t,})|^{2}dP\nonumber \\
= & \int\Big|\big\{ e^{\hat{r}_{t}(S_{t})}-e^{r_{t}'(S_{t})}\big\} \omega(t) I_{t}f_{t}(S_{t})\Bigg[(W_{t,\Delta} Y_{t,\Delta}-\hat{\mu}_{t})\{A_{t}-\hat{m}_{t}(S_{t})\}\nonumber \\
 & \ \ +(1-\hat{\mu}_{1t})\{1-\hat{m}_{t}(S_{t})\}p_{t}(S_{t})+(1-\hat{\mu}_{0t})\hat{m}_{t}(S_{t})\{1-p_{t}(S_{t})\}\Bigg]\Big|^{2}dP\nonumber \\
\leq & C\int|e^{\hat{r}_{t}(S_{t})}-e^{r_{t}'(S_{t})}|^{2}dP\label{eq: in lemma U -U is op(1) r integral}\\
= & o_{p}(1),\label{eq:in lemma U -U is op(1) r is op(1)}
\end{align}
where (\ref{eq: in lemma U -U is op(1) r integral}) follows from
Assumptions \ref{assu: (Regularity conditions: support of data is bounded)}
and \ref{assu:(Regularity-conditions) mt rt are between 0 and 1}
, and (\ref{eq:in lemma U -U is op(1) r is op(1)}) follows from limit
of $\hat{r}_{t}$, 
exponential function is continuous, and $r_{t}$ is bounded away from
0 and 1 (Assumption \ref{assu:(Regularity-conditions) mt rt are between 0 and 1}).

To show (\ref{eq: in Lemma U - U is op(1) three term m_t}), we have
\begin{align}
 & \int|\USR_{t}(\beta^{\star},r'_{t},\hat{m}_{t},\hat{\mu}_{t,})-\USR_{t}(\beta^{\star},r'_{t},m'_{t},\hat{\mu}_{t,})|^{2}dP\nonumber \\
= & \int\Big|\{\hat{m}_{t}(S_{t})-m_{t}'(S_{t})\}\omega(t) I_{t}f_{t}(S_{t})\Big[(W_{t,\Delta}Y_{t,\Delta}-\mu_{t})\Big\{ e^{-A_{t}f_{t}(S_{t})^{T}\beta^{\star}}+e^{r_{t}'(S_{t})}\Big\}\nonumber \\
 & +\Big\{\mu_{1t}e^{-f_{t}(S_{t})^{T}\beta^{\star}}-(1-\mu_{1t})e^{r_{t}'(S_{t})}\Big\} p_{t}(S_{t})+\Big\{\mu_{0t}-(1-\mu_{0t})e^{r_{t}'(S_{t})}\Big\}\{1-p_{t}(S_{t})\}\Big]\Big|^{2}dP\nonumber \\
\leq & \int|\hat{m}_{t}(S_{t})-m_{t}'(S_{t})|^{2}dP\label{eq: in lemma U - U is op(1) difference of m}\\
= & o_{p}(1)
\end{align}
where (\ref{eq: in lemma U - U is op(1) difference of m}) follows
from Assumptions \ref{assu: Regularity conditions: support compact space},
\ref{assu: (Regularity conditions: support of data is bounded)} and
\ref{assu:(Regularity-conditions) mt rt are between 0 and 1}.

To show (\ref{eq:in Lemma U - U is op(1) three term mu_t}), we have
\begin{align*}
 & \int|\USR_{t}(\beta^{\star},r'_{t},m'_{t},\hat{\mu}_{t,})-\USR_{t}(\beta^{\star},r'_{t},m'_{t},\mu'_{t})\Big\}\Big|^{2}dP\\
= & \int\Big|(\mu_{t}'-\hat{\mu}_{t})\Big\{ e^{-A_{t}f_{t}(S_{t})^{T}\beta^{\star}}+e^{r_{t}'(S_{t})}\Big\}\Big|^{2}dP\\
\leq & \int|\mu_{t}'-\hat{\mu}_{t}|^{2}dP\\
= & o_{p}(1)
\end{align*}
Thus, (\ref{eq:in Lemma U - U is op(1) three term mu_t})
holds. The conclusion of the lemma follows.
\end{proof}

\begin{lem}
\label{lem:(Rate-Double-Robustness)}(Rate Double Robustness) Suppose
Assumptions \ref{assu: Regularity conditions: support compact space},
\ref{assu: (Regularity conditions: support of data is bounded)},
and \ref{assu:(Regularity-conditions) mt rt are between 0 and 1}
hold. Then for
$\USR(\beta,\eta)$ defined in (\ref{lem: EE =00003D 0}),
we have $|\PP\{\USR(\beta^{\star},\hat{\eta})\}|\lessrate\lVert\hat{r}_t-r^{\star}_t\rVert\lVert\hat{m}_t-m^{\star}_t\rVert$.
\end{lem}
\begin{proof}
We have 
\begin{align}
& \ \ \PP\{\USR(\beta^{\star},\hat{\eta})\} \\ 
& =\sum_{t=1}^{T}\PP\Big[ \omega(t) I_{t}\{W_{t,\Delta}Y_{t,\Delta}e^{-A_{t}f_{t}(S_{t})^{T}\beta^{\star}}-(1-Y_{t,\Delta})e^{\hat{r}_{t}(S_{t})}\}\{A_{t}-\hat{m}_{t}(S_{t})\}f_{t}(S_{t})\Big]\nonumber \\
 & =\sum_{t=1}^{T}\PP\Bigg\{\PP\Big[\omega(t) I_{t}\{W_{t,\Delta}Y_{t,\Delta}e^{-A_{t}f_{t}(S_{t})^{T}\beta^{\star}}-(1-W_{t,\Delta}Y_{t,\Delta})e^{\hat{r}_{t}(S_{t})}\}\{A_{t}-\hat{m}_{t}(S_{t})\}f_{t}(S_{t})\Big]\Big|S_{t},A_{t}\Bigg\}\nonumber \\
 & =\sum_{t=1}^{T}\PP\Bigg[\omega(t) I_{t}\Big\{\frac{e^{r_{t}^{\star}(S_{t})}}{1+e^{A_{t}f_{t}(S_{t})^{T}\beta^{\star}+r_{t}^{\star}(S_{t})}}-\frac{e^{\hat{r}_{t}(S_{t})}}{1+e^{A_{t}f_{t}(S_{t})^{T}\beta^{\star}+r_{t}^{\star}(S_{t})}}\Big\}\{m_{t}^{\star}(S_{t})-\hat{m}_{t}(S_{t})\}f_{t}(S_{t})\Bigg]\nonumber \\
 & =\sum_{t=1}^{T}\PP\Bigg[\omega(t) I_{t}\frac{e^{\hat{r}_{t}(S_{t})}}{1+e^{A_{t}f_{t}(S_{t})^{T}\beta^{\star}+r_{t}^{\star}(S_{t})}}\{e^{r_{t}^{\star}(S_{t})-\hat{r}_{t}(S_{t})}-1\}\{m_{t}^{\star}(S_{t})-\hat{m}_{t}(S_{t})\}f_{t}(S_{t})\Bigg].\label{eq: inside lemma rate DR simplified eq}
\end{align}
By Assumption \ref{assu: Regularity conditions: support compact space},
\ref{assu: (Regularity conditions: support of data is bounded)},
and \ref{assu:(Regularity-conditions) mt rt are between 0 and 1},
\[
\PP\{\USR(\beta^{\star},\hat{\eta})\}\lessrate\sum_{t=1}^{T}\PP\Big[\{e^{r_{t}^{\star}(S_{t})-\hat{r}_{t}(S_{t})}-1\}\{m_{t}^{\star}(S_{t})-\hat{m}_{t}(S_{t})\}\Big].
\]
By Taylor expansion on exponential function, $e^{r_{t}^{\star}(S_{t})-\hat{r}_{t}(S_{t})}-1\lessrate r_{t}^{\star}(S_{t})-\hat{r}_{t}(S_{t})$
when $\lVert r^{\star}_t-\hat{r}_t\rVert^{2}=o_{p}(1)$. Thus, $|\PP\{\USR(\beta^{\star},\hat{\eta})\}|\lessrate\lVert\hat{r}_t-r^{\star}_t\rVert\lVert\hat{m}_t-m^{\star}_t\rVert$. 
\end{proof}

\subsection{\texorpdfstring{Asymptotic Normality of $\hat{\beta}^\text{SR}$}{Asymptotic Normality of beta-hat SR}}
\begin{thm}
Suppose Assumption 1 (Consistency, Positivity, and Sequentially
ignorability) in the main paper, Assumptions \ref{assu: unique zero}, 
\ref{assu: Regularity conditions} hold.
For each $t \in [T]$, suppose that either $r'_t = r^\star_t$ or $m'_t = m^\star_t$, then $\hat{\beta}^\text{SR}$ is consistent. Furthermore, if the fitted nuisance functions satisfies 
\begin{align}
     \lVert\hat{r}_t-r_t^{\star}\rVert\lVert\hat{m}_t-m^{\star}_t\rVert=o_{p}(n^{-1/2}),
\end{align}
for each $t \in [T]$, then $\hat{\beta}^\text{SR}$ is asymptotically normal: $\sqrt{n}(\hat{\beta}^\text{SR}-\beta^{\star})\xrightarrow{d}N(0,V^\text{SR})$
as $n\rightarrow\infty$ where 
\[
 V^\text{SR} =\EE\{\partial_{\beta}\USR(\beta^{{\star}}, \eta')\}^{-1}\PP\{\USR(\beta^{{\star}},\eta')\USR(\beta^{{\star}}, \eta')^{T}\}\PP\{\partial_{\beta}\USR(\beta^{{\star}},\eta')\}^{-1,T},
\]
 and $ V^\text{SR} $ can be consistently estimated by 
\[
\PP_{n}\{\partial_{\beta}\USR(\hat{\beta}^\text{SR},\hat{\eta})\}^{-1}\PP_{n}\{\USR(\hat{\beta}^\text{SR},\hat{\eta})\USR(\hat{\beta}^\text{SR},\hat{\eta})^{T}\}\PP_{n}\{\partial_{\beta}\USR(\hat{\beta}^\text{SR},\hat{\eta})\}^{-1,T}.
\]
\end{thm}
\begin{proof}
Combing the fact that $\PP_{n}\USR(\hat{\beta},\hat{\eta})=0$ and
the Lagrange mean value theorem, we have
\begin{align*}
\PP_{n}\USR(\beta^{\star},\hat{\eta})+\{\frac{{\partial}}{\partial\beta^{T}}\PP_{n}\USR(\beta',\hat{\eta})\}(\hat{\beta}-\beta^{\star}) & =0\\
\PP_{n}\USR(\beta^{\star},\hat{\eta})+\{\PP_{n}\partial_{\beta}\USR(\beta',\hat{\eta})\}(\hat{\beta}-\beta^{\star}) & =0,
\end{align*}
where $\beta'$ is between $\hat{\beta}$ and $\beta^{\star}$. By
Lemma \ref{lem:Consistency of =00005Chat=00007B=00005Cbeta=00007D},
$\hat{\beta}\pto\beta^{\star}$. Thus, $\beta'\pto\beta^{\star}$.
By Lemma \ref{lem: Convergence of derivative}, $\PP_{n}\{\partial_{\beta}\USR(\hat{\beta},\hat{\eta})\}\pto\PP\{\partial_{\beta}\USR(\beta^{\star},\eta')\}$.
By Assumption \ref{assu: partialU is invertiable}, we have 
\begin{equation}
\sqrt{n}(\hat{\beta}-\beta^{\star})=-\{\PP_{n}\partial_{\beta}\USR(\beta^{\star},\eta')\}^{-1}\sqrt{n}\{\PP_{n}\USR(\beta^{\star},\hat{\eta})\}.\label{eq: inside theorem decompose equatioin}
\end{equation}
For $\sqrt{n}\{\PP_{n}\USR(\beta^{\star},\hat{\eta})\}$ in
(\ref{eq: inside theorem decompose equatioin}), we can expand it
into 
\begin{align}
\sqrt{n}\PP_{n}\USR(\beta^{\star},\hat{\eta}) & =\sqrt{n}\Big\{\PP_{n}\USR(\beta^{\star},\hat{\eta})-\PP\USR(\beta^{\star},\hat{\eta})+\PP\USR(\beta^{\star},\hat{\eta})+\PP\USR(\beta^{\star},\eta')\Big\}\label{eq: inside theorem using lemma unique zero}\\
 & =\sqrt{n}(\PP_{n}-\PP)\{\USR(\beta^{\star},\hat{\eta})\}+\sqrt{n}\PP\{\USR(\beta^{\star},\hat{\eta})-\USR(\beta^{\star},\eta')\},\label{eq: inside lemma showing the second part}
\end{align}
where (\ref{eq: inside theorem using lemma unique zero}) holds because
of Lemma \ref{lem: EE =00003D 0}. Next, we want to show that the
first term in (\ref{eq: inside lemma showing the second part}), $(\PP_{n}-\PP)\{\USR(\beta^{\star},\hat{\eta})\}$
is asymptotically normal. By Assumption \ref{assu:(Donsker-condition)},
Lemma \ref{lem: bound of U(beta_star, delta_hat) and U(beta_star, delta')},
and (Lemma 19.24 of \citet{van2000asymptotic}), we have 
\begin{align}
\sqrt{n}(\PP_{n}-\PP)\{\USR(\beta^{\star},\hat{\eta})-\USR(\beta^{\star},\eta')\} & =o_{p}(1)\nonumber \\
\Rightarrow\ \ \ \ \ \ \ \ \ \ \ \ \sqrt{n}(\PP_{n}-\PP)\{\USR(\beta^{\star},\hat{\eta})\} & =\sqrt{n}(\PP_{n}-\PP)\{\USR(\beta^{\star},\eta')\}+o_{p}(1).\label{eq: inside theorem swap between delta_hat and delta'}
\end{align}
Since $\PP\{\USR(\beta^{\star},\eta')\}=0$ (Lemma \ref{lem: PU =00003D 0}),
$$\sqrt{n}(\PP_{n}-\PP)\{\USR(\beta^{\star},\eta')\}\dto N(0,\PP\{\USR(\beta^{\star},\eta')\USR(\beta^{\star},\eta')^{T}\})$$
by Linderberg-Feller Central Limit Theorem. By (\ref{eq: inside theorem swap between delta_hat and delta'}),
we have 
\[
\sqrt{n}(\PP_{n}-\PP)\{\USR(\beta^{\star},\hat{\eta})\}\dto N(0,\PP\{\USR(\beta^{\star},\eta')\USR(\beta^{\star},\eta')^{T}\}).
\]

For the second term in (\ref{eq: inside lemma showing the second part}),
combing with Lemma \ref{lem: EE =00003D 0} and \ref{lem:(Rate-Double-Robustness)}
yields that
\[
|\PP\{\USR(\beta^{\star},\hat{\eta})\}-\PP\{\USR(\beta^{\star},\eta')\}|=|\PP\{\USR(\beta^{\star},\hat{\eta})\}|\lessrate\lVert\hat{r}_t-r^{\star}_t\rVert\lVert\hat{m}_t-m^{\star}_t\rVert.
\]
By Assumption of the theorem, $|\PP\{\USR(\beta^{\star},\hat{\eta})\}-\PP\{\USR(\beta^{\star},\eta')\}|=o_{p}(1/\sqrt{n}).$
Therefore, it follows from Slutsky\textquoteright s theorem and the
continuous mapping theorem that $\sqrt{n}(\hat{\beta}-\beta^{\star})\dto N(0,V^\text{SR})$
with $V^\text{SR}$ defined in the theorem statement. By Lemma \ref{lem: Convergence of derivative},
Lemma \ref{lem:(Convergence-of-the meat term)}, and the continuous
mapping theorem, consistency of the variance estimator follows. This
completes the proof.
\end{proof}

\section{Proof of Theorem 2}
\label{supp-sec:proof_of_consistency_and_asymptotic_normality_of_method_2}
\subsection{Setup and Notations}
Define $\UGR_t(\beta, \alpha, \mu_t)$ to be a summand of $\UGR(\beta, \alpha, \mu)$ such that $\UGR(\beta, \alpha, \mu) = \sum_{t=1}^T \UGR_t(\beta, \alpha, \mu_t)$, and
\begin{align}
&\UGR_t(\beta, \alpha, \mu_t) \\
  = &\sum_{t=1}^{T} \omega(t) I_t\bigg[ \expit\{g_t(S_t)^T \alpha - f_t(S_t)^T\beta\} - \mu_{0t} - \frac{1 - A_t}{1 - p_t(H_t)} \{W_{t, \Delta}Y_{t, \Delta} - \mu_{0t} \} \bigg] f_t(S_t).\label{eq: m2 - estimating equation}
\end{align}
Let $\theta = (\beta, \alpha)$ and $\mu_t'$ denote the limit of the fitted nuisance function $\hat{\mu}_t$.

\subsection{Assumptions}
We state the following necessary regularity conditions for establishing Theorem 2 in the main paper
\begin{asu}
\label{assu: convergence of the nuisance limit} (Convergence of nuisance parameter)
There exists some $\mu_t'$ such that $\lVert \hat{\mu}_t - \mu_t' \rVert = o_p(1)$.
\end{asu}

\begin{asu}
\label{assu:m2 - Unique zero}(Unique zero) Given the limit $\mu'$, $\PP\{\UGR(\theta, \mu')\}$ as a function of $\theta$ has a unique zero.
\end{asu}
\begin{asu} (Regularity conditions)
	\label{assu: m2 - regularity condition}
	\begin{asulist}
	\item \label{assu: m2-Regularity conditions: support compact space} Suppose the parameter sapce $\Theta$ of $\beta$ is
compact.
	\item \label{assu: m2- (Regularity conditions: support of data is bounded)} Suppose the support of $O$ is bounded.
	\item \label{assu: m2 - partialU is invertiable}
Suppose $\PP\{\partial_{\theta}\UGR(\theta^{\star}, \mu')\}$ is invertible.
	\item \label{assu: m2 - (Regularity-conditions) derivative is uniformly bounded} $\partial_{\theta}\UGR(\theta,\mu)$ and $\UGR(\theta,\mu)\UGR(\theta,\mu)^{T}$
are uniformly bounded by integrable function. 
	\end{asulist}
\end{asu}
\begin{asu}
\label{assu:m2 - donsker class}(Donsker condition). Suppose for each t, the estimator $\hat{\mu}_t$ take values in Donsker class.
\end{asu}

\subsection{Lemmas}
\begin{lem}
\label{lem: m2 - EE =00003D 0} $\PP\{\UGR(\theta^{\star}, \mu)\}=0$ for any $\mu$.
\end{lem}
\begin{proof}
The definition of $\beta^\star$ and $\alpha^\star$ implies that,
\begin{align*}
	\PP \Bigg\{ \sum_{t = 1}^T \omega(t) \Bigg( g_t(S_t)^T \alpha^\star - f_t(S_t)^T\beta^\star - \logit[\PP\{\PP(W_{t,\Delta}Y_{t,\Delta}|H_t, A_t = 0)\}|S_t, I_t = 1] \Bigg) \Bigg\}= 0 \numberthis \label{eq: the model alpha and beta together}
\end{align*}

For any $\mu$, we have
\begin{align}
 & \PP\{\UGR(\theta, \mu)\}\nonumber \\
= & \sum_{t=1}^{T}\PP\Big[\PP\big\{\UGR_{t}(\beta,\alpha,\mu_{t})|H_{t,}A_{t}=1\big\} p_{t}(H_t)+\PP\big\{\UGR_{t}(\beta, \alpha,\mu_{t})|H_{t,}A_{t}=0\big\} \{1-p_{t}(H_t)\} \Big]\nonumber \\
= & \sum_{t=1}^{T}\PP\Bigg\{ \omega(t) I_{t}\Big[\expit\{ g_t(S_{t})^T \alpha -f_{t}(S_{t})^{T}\beta\}-\mu_{0t}\Big]f_{t}(S_{t})p_{t}(H_t)\nonumber \\
 & \ \ \ \ \ + \omega(t)I_{t}\Big[\expit\{g_t(S_{t})^T \alpha-f_{t}(S_{t})^{T}\beta\}-\mu_{0t} \\ 
 &\quad \quad -\frac{1}{1-p_t(H_t)}\{\PP(W_{t,\Delta} Y_{t,\Delta}|H_{t},A_{t}=0)-\mu_{0t}\}\Big]f_{t}(S_{t})\{1-p_{t}(H_t)\}\Bigg\}\nonumber \\
= & \sum_{t=1}^{T}\PP\Bigg\{ \omega(t) I_{t}\Big[\expit\{g_t(S_{t})^T \alpha -f_{t}(S_{t})^{T}\beta\}-\PP(W_{t,\Delta}Y_{t,\Delta}|H_{t},A_{t}=0)\Big]f_{t}(S_{t})\Bigg\}\nonumber \\
= & \sum_{t=1}^{T}\PP\Bigg[\PP\Bigg\{ \omega(t) I_{t}\Big[\expit\{g_t(S_{t})^T \alpha-f_{t}(S_{t})^{T}\beta\}-\PP(W_{t,\Delta}Y_{t,\Delta}|H_{t},A_{t}=0)\Big]f_{t}(S_{t})\Big|S_{t},I_{t}=1\Bigg\}\Bigg],\label{eq: m2 in lemma EE =00003D 0 the last step}
\end{align}
Replacing $(\beta, \alpha)$ with $(\beta^{\star}, \alpha^{\star})$ in \eqref{eq: m2 in lemma EE =00003D 0 the last step},
we have $\PP\{\UGR(\beta^{\star}, \alpha^\star, \mu)\}=0$ by \eqref{eq: the model alpha and beta together}.
\end{proof}

\begin{lem}
\label{lem: m2- convergence of sup between hat mu and mu prime}  Suppose Assumptions \ref{assu: convergence of the nuisance limit}, \ref{assu: m2-Regularity conditions: support compact space}, and \ref{assu: m2- (Regularity conditions: support of data is bounded)} hold, then  
$\sup_{\theta \in \Theta}|\PP \UGR(\theta, \hat{\mu}) - \PP\UGR(\theta, \mu')| = o_p(1)$.
\end{lem}
\begin{proof}
Let $h_t(\mu_t) = \omega(t)I_t\big[\mu_{0t} + \frac{1 - A_t}{1 - p_t(H_t)} \{W_{t, \Delta}Y_{t, \Delta} - \mu_{0t} \} \big] $. Then we have
\begin{align*}
	&\sup_{\theta \in \Theta}|\PP \UGR(\theta, \hat{\mu}) - \PP\UGR(\theta, \mu')| \\
	= & \PP \{ \sum_t h_t(\hat{\mu}_t) - h_t(\mu'_t) \} f_t(S_t) 
\end{align*}
where the last line is bounded by $o_p(1)$ because $\hat{\mu}_t$ converges to $\mu_t'$ (Assumption \ref{assu: convergence of the nuisance limit}) and Assumptions \ref{assu: m2-Regularity conditions: support compact space}, \ref{assu: m2- (Regularity conditions: support of data is bounded)}. This concludes the proof.
\end{proof}

\begin{lem}
\label{lem: m2- convergence of square}  Suppose Assumptions \ref{assu: m2- (Regularity conditions: support of data is bounded)} and \ref{assu: convergence of the nuisance limit}  hold, then $\lVert \UGR(\theta^\star, \hat{\mu}) -  \UGR(\theta^\star, \mu')\rVert^2 = o_p(1)$.
\end{lem}
\begin{proof}
We have
\begin{align*}
 &\lVert \UGR(\theta^\star, \hat{\mu}) -  \UGR(\theta^\star, \mu')\rVert^2 \\
&= \int | \UGR(\theta^\star, \hat{\mu}) -  \UGR(\theta^\star, \mu')|^2 dP \\
&= \int | \sum_t \UGR_t(\theta^\star, \hat{\mu}_t)  -  \sum_t \UGR_t(\theta^\star, \mu'_t)|^2 dP \\
&= 2T^2 \Big\{ \max_{1\leq t \leq T} \int | \UGR_t(\theta^\star, \hat{\mu}_t) -  \UGR_t(\theta^\star, \mu'_t)|^2 dP\Big\}.
\end{align*}
Thus, it suffices to show that $\int | \UGR_t(\theta^\star, \hat{\mu}_t) -  \UGR_t(\theta^\star, \mu'_t)|^2 dP = o_p(1)$. By expanding the equation, we have
\begin{align*}
	&\int | \UGR_t(\theta^\star, \hat{\mu}_t) -  \UGR_t(\theta^\star, \mu'_t)|^2 dP \\
	= &\int \Big| \big\{ \mu'_{0t} -\hat{\mu}_{0t} - \frac{1 - A_t}{1 - p_t(H_t)} (\mu'_{0t} - \hat{\mu}_{0t}) \big\} f_t(S_t) \Big|^2dP \\
	\leq  &C \int |\mu'_{0t} -\hat{\mu}_{0t}|^2 dP \numberthis \label{eq:m2 in the convergence of square mu1}\\
= &o_p(1) \numberthis \label{eq:m2 in the convergence of square mu2}
\end{align*}
where \eqref{eq:m2 in the convergence of square mu1} follows from Assumption \ref{assu: m2- (Regularity conditions: support of data is bounded)} and \eqref{eq:m2 in the convergence of square mu2} follows from the limit of $\mu_t'$.
\end{proof}

\subsection{\texorpdfstring{Asymptotic Normality of $\hat{\beta}^\text{GR}$}{Asymptotic Normality of beta-hat GR}}
\label{subsec:_texorpdfstring_asymptotic_normality_of_hat_beta}
\begin{thm}
Suppose Suppose Assumption 1 (Consistency, Positivity, and Sequentially
ignorability) in the main paper and Assumptions \ref{assu: convergence of the nuisance limit}, \ref{assu:m2 - Unique zero}, \ref{assu: m2 - regularity condition}, \ref{assu:m2 - donsker class}  hold. Suppose there exists $\alpha^\star$ such that $\psi_t^\star = g(S_t)^T\alpha^\star$ for $\psi_t^\star$ defined in the main paper. Let $\Phi(\beta, \alpha, \mu):=(\UGR(\beta, \alpha, \mu), Q(\alpha))$. 
Then $\hat{\beta}^\text{GR}$ is consistent and asymptotically normal: $\sqrt{n}(\hat{\beta}^\text{GR}-\beta^{\star})\xrightarrow{d}N(0,V^\text{GR})$
as $n\rightarrow\infty$. Furthermore, $V^\text{GR}$ can be consistently
estimated by the upper diagonal $p$ by $p$ block matrix of 
\[
\PP_{n} \big\{ \frac{\partial \Phi(\hat{\beta}, \hat{\alpha}, \hat{\mu}) }{\partial_{(\beta^T, \alpha^T)}} \big\}^{-1}\PP_{n}\{\Phi(\hat{\beta}, \hat{\alpha}, \hat{\mu}) \Phi(\hat{\beta}, \hat{\alpha}, \hat{\mu}) ^{T}\}\PP_{n}\big\{ \frac{\partial \Phi(\hat{\beta}, \hat{\alpha}, \hat{\mu}) }{\partial_{(\beta^T, \alpha^T)}} \big\}^{-1,T},
\]
\end{thm}
\begin{proof}
Let $\theta = (\beta, \alpha)$. By applying Theorem 5.1 of \citet{cheng2023efficient}, we have the asymptotic normality of $\hat{\beta}^\text{GR}$. Thus, it suffices to check the assumptions requ{}ired for Theorem 5.1 hold for $\UGR(\theta, \mu)$. 

Assumption 5.1 of \citet{cheng2023efficient} holds because of Lemma \ref{lem: m2 - EE =00003D 0} and Assumption \ref{assu:m2 - Unique zero}.

Assumption 5.2 (i) of \citet{cheng2023efficient} is verified by Lemma \ref{lem: m2- convergence of sup between hat mu and mu prime}; Assumption 5.2 (ii) of \citet{cheng2023efficient} is validated by  Lemma \ref{lem: m2- convergence of square}; Assumption 5.2 (iii) and 5.2(iv) of \citet{cheng2023efficient} hold by the Assumption that $\partial_{\theta}\UGR(\theta, \mu)$, $\UGR(\theta, \mu) \UGR(\theta, \mu)^T$
are each bounded by integrable function (Assumptions \ref{assu: m2 - (Regularity-conditions) derivative is uniformly bounded}) and the dominated convergence theorem (\citet{chung2001course} Result (viii)
of Chapter 3.2).

Assumption 5.3 and Assumption 5.4 of \citet{cheng2023efficient} is verified by Assumptions \ref{assu: m2 - regularity condition} and \ref{assu:m2 - donsker class}.

Therefore, the conclusion of the theorem follows.
\end{proof}

\section{Proof of Theorem 3}
\label{sec:proof_of_theorem_3}

\subsection{Assumptions}
\label{subsec:assumptions of theorem 3}

\begin{asu}
\label{assu: thm3 convergence of the nuisance limit} (Convergence of nuisance parameter)
There exists some $\alpha'$ such that $\lVert \hat{\alpha} - \alpha' \rVert = o_p(1)$. There exists some $\mu_t'$ such that $\lVert \hat{\mu}_t - \mu_t' \rVert = o_p(1)$.
\end{asu}

\begin{asu} 
\label{assu: thm3 unique zero}(Unique zero)
		Given the limit $\mu'$, there exists unique $\alpha$ such that $\PP \{Q(\alpha)^T, \UGR(\beta^\star, \alpha, \mu')^T\}$ when $\beta^\star = 0$.
\end{asu}

\begin{asu} 
\label{assu: thm3 regularity conditions} (Regularity conditions)
\begin{asulist}
	\item \label{assu:thm3 regularity condition compact space} Suppose the parameter space of $\alpha$ and $\beta$ are compact.
	\item \label{assu:thm3 regularity condition bounded support} Suppose the support of $O$ is bounded.
	\item \label{assu:thm3 regularity condition invertibel}Suppose $\PP\{ \partial_{(\beta, \alpha)} \UGR(\beta^\star, \alpha', \mu')  \}$ is invertible.
	\item \label{assu:thm3 regularity condition integrable function}$\partial_{(\beta, \alpha)} \UGR(\beta, \alpha, \mu)$ and $\UGR(\beta, \alpha, \mu)\UGR(\beta, \alpha, \mu)^T$ are uniformly bounded by integrable function.
\end{asulist}
\end{asu}

\begin{asu} 
\label{assu: thm3 donsker condition}(Donsker condition)
Suppose for each t, the estimator $\hat{\mu}_t$ take values in Donsker class.
\end{asu}

\subsection{Lemmas}
\begin{lem}
\label{lem: thm3 PPU = 0}
	When $\beta^\star = 0$, $\PP\{\UGR(\beta^\star, \alpha', \mu)\} = 0$ for any $\mu$.
\end{lem}

\begin{proof}
Based on the estimating function for $\alpha$, Assumptions \ref{assu: thm3 convergence of the nuisance limit} (convergence of the nuisance function) and \ref{assu: thm3 unique zero} (unique zero), the limit of $\hat{\alpha}$, $\alpha'$ satisfies 
\begin{align*}
	 \PP \Bigg( \sum_t  \frac{A_t}{p_t(H_t)} I_t \Big[W_{t,\Delta} Y_{t,\Delta} -  \expit\{ g_t(S_t)^T \alpha' \} \Big] g_t(S_t) \Bigg) = 0.
\end{align*}
By the law of total expectation, we have
\begin{align}
\PP \Big(\sum_t I_t \Big[\PP(W_{t,\Delta} Y_{t,\Delta}|H_t, A_t = 1) - \expit\{ g_t(S_t)^T \alpha' \} \Big]  g_t(S_t) \Big) = 0. \numberthis \label{eq: GR robustness PPU = 0 equality from EE alpha}
\end{align}

When $\beta = \beta^\star$, the estimating function $\UGR(\beta^\star, \alpha, \mu_t)$ becomes
\begin{align}
	& \sum_t \omega(t) I_t \Big[ \expit\{ g_t(S_t)^T \alpha \} - \mu_{0t} - \frac{1 - A_t}{1 - p_t(H_t)} \{W_{t,\Delta} Y_{t, \Delta} - \mu_{0t}\} \Big] f_t(S_t) \nonumber \\ 
	= & \sum_t \omega(t) I_t \Big[ \expit\{ g_t(S_t)^T \alpha \} - \frac{1 - A_t}{1 - p_t(H_t)} W_{t,\Delta} Y_{t, \Delta} \Big] f_t(S_t) \label{eq: GR robusness PPU=0 alpha part} \\ 
	 & + \sum_t \omega(t) I_t \frac{p_t(H_t) - A_t}{1 - p_t(H_t)} \mu_{0t}  f_t(S_t) \label{eq: GR robusness PPU=0 mu part}
\end{align}

To show that $\PP\{\UGR(\beta^\star, \alpha', \mu)\} = 0$ for any $\mu$, we will first show the expectation of \eqref{eq: GR robusness PPU=0 mu part} equals 0 for any $\mu$. Using iterated expectation, we have for any $\mu$,
\begin{align*}
	&\PP  \Big\{  \sum_t \omega(t) I_t \frac{p_t(H_t) - A_t}{1 - p_t(H_t)} \mu_{0t}  f_t(S_t) \Big\} \\ 
	=& \PP  \Big\{  \sum_t \omega(t) I_t \frac{p_t(H_t) - \PP(A_t|H_t)}{1 - p_t(H_t)} \mu_{0t}  f_t(S_t) \Big\} \\ 
	= & 0 
\end{align*}
Thus, for any $\mu$, we have
\begin{align*}
	&\PP\{\UGR(\beta^\star, \alpha', \mu)\} \\ 
	=& \sum_t \PP \Big( \omega(t) I_t \Big[ \expit\{ g_t(S_t)^T \alpha' \} - \frac{1 - A_t}{1 - p_t(H_t)} W_{t,\Delta} Y_{t, \Delta} \Big] f_t(S_t) \Big) \\ 
	= & \sum_t \PP \Big(\omega(t) I_t \expit\{ g_t(S_t)^T \alpha' \} f_t(S_t) p_t(H_t)  \\ 
	&+ \omega(t) I_t \Big[ \expit\{ g_t(S_t)^T \alpha' \} - \frac{1}{1 - p_t(H_t)} \PP(W_{t,\Delta} Y_{t, \Delta} |H_t, A_t = 0) \Big] f_t(S_t) \{1 - p_t(H_t)\}\Big) \\ 
	= & \sum_t \PP \Big(\omega(t) I_t \Big[ \expit\{ g_t(S_t)^T \alpha' \} - \PP(W_{t,\Delta} Y_{t, \Delta}|H_t, A_t = 0) \Big] f_t(S_t) \Big)\\ 
	= & 0
\end{align*}
where the last equality follows from \eqref{eq: GR robustness PPU = 0 equality from EE alpha}. This concludes the proof
\end{proof}

{}

\begin{lem}
\label{lem: thm3 convergence of sup between hat mu and mu prime}
Suppose Assumptions \ref{assu:thm3 regularity condition compact space}, \ref{assu:thm3 regularity condition bounded support} hold, then $\sup|\PP \{\UGR(\beta, \alpha, \hat{\mu} \} - \PP \{\UGR(\beta, \alpha, \mu' \}| = o_p(1)$.
\end{lem}
\begin{proof}
The proof is similar to the proof in Lemma \ref{lem: m2- convergence of sup between hat mu and mu prime} by setting $\theta = (\beta, \alpha)$.
\end{proof}

\begin{lem}
\label{lem: thm3 convergence of square}
Suppose Assumptions \ref{assu: thm3 convergence of the nuisance limit}, \ref{assu:thm3 regularity condition compact space}, \ref{assu:thm3 regularity condition bounded support} hold, then $\lVert \UGR(\beta^\star, \alpha', \hat{\mu}) -  \UGR(\beta^\star, \alpha', \mu')\rVert^2 = o_p(1)$
\end{lem}

\begin{proof}
The proof is similar to the proof in Lemma \ref{lem: m2- convergence of square} by setting $\theta^\star = (\beta^\star, \alpha')$.
\end{proof}

\begin{thm}
 Suppose Assumption 1 (consistency, positivity, and sequentially ignorability) in the main paper and Assumptions \ref{assu: thm3 convergence of the nuisance limit}, \ref{assu: thm3 unique zero}, \ref{assu: thm3 regularity conditions}, and \ref{assu: thm3 donsker condition} hold. Suppose the estimator $\hat{\alpha}$ solves the estimating equation $\PP_n \{Q(\alpha)\} = 0$. When $\beta^\star = 0$, $\hat{\beta}^\mathrm{GR}$ is consistent and asymptotically normal: $\sqrt{n}(\hat{\beta}^\mathrm{GR}-\beta^{\star})\xrightarrow{d}N(0,V^\mathrm{GR})$ as $n\rightarrow\infty$. Furthermore, $V^\text{GR}$ can be consistently
estimated by the upper diagonal $p$ by $p$ block matrix of 
\[
\PP_{n} \big\{ \frac{\partial \Phi(\hat{\beta}, \hat{\alpha}, \hat{\mu}) }{\partial_{(\beta^T, \alpha^T)}} \big\}^{-1}\PP_{n}\{\Phi(\hat{\beta}, \hat{\alpha}, \hat{\mu}) \Phi(\hat{\beta}, \hat{\alpha}, \hat{\mu}) ^{T}\}\PP_{n}\big\{ \frac{\partial \Phi(\hat{\beta}, \hat{\alpha}, \hat{\mu}) }{\partial_{(\beta^T, \alpha^T)}} \big\}^{-1,T},
\]
where $\Phi(\hat{\beta}, \hat{\alpha}, \hat{\mu}) = (\UGR(\hat{\beta}, \hat{\alpha}, \hat{\mu})^T, Q(\hat{\alpha})^T)^T$.
\end{thm}

\begin{proof}
By applying Theorem 5.1 of \citet{cheng2023efficient}, we have the asymptotic normality of $\hat{\beta}^\mathrm{GR}$. Thus, we will focus on verifying the assumptions of Theorem 5.1.

Assumption 5.1 holds because of Lemma \ref{lem: thm3 PPU = 0}.

Assumption 5.2 (i) and Assumption 5.2 (ii) of \citet{cheng2023efficient} are verified by Lemma \ref{lem: thm3 convergence of square} and Lemma \ref{lem: thm3 convergence of sup between hat mu and mu prime} respectively; Assumption 5.2 (iii) and 5.2 (iv) hold because of Assumption \ref{assu:thm3 regularity condition integrable function} and the dominated convergence theorem (\citet{chung2001course} Result (viii) of Chapter 3.2).

Assumption 5.3 holds because of Assumption \ref{assu: thm3 regularity conditions}.

Assumprion 5.4 holds because $\mu_t(\cdot)$ is from Donsker class, $\UGR(\beta, \alpha, \mu)$ is a polynomial of $\mu_t$ and thus a Donsker class (by Lipschitz transformation).
\end{proof}

\section{Extension to Other Link Functions}
\label{sec:extension_to_other_link_functions}
We will show that Theorem 2 also holds under generalized link functions. We will first state the necessary assumptions and show lemmas.

\begin{asu}
\label{assu: generalized-convergence of the nuisance limit} (Convergence of nuisance parameter)
There exists some $\mu_t'$ such that $\lVert \hat{\mu}_t - \mu_t' \rVert = o_p(1)$.
\end{asu}

\begin{asu}
\label{assu: generalized-m2 - Unique zero}(Unique zero) Given the limit $\mu'$, $\PP\{U^\text{GR-Generalized}(\beta, \alpha, \mu_t')\}$ has a unique zero.
\end{asu}
\begin{asu} (Regularity conditions)
	\label{assu: generalized-m2 - regularity condition}
	\begin{asulist}
	\item \label{assu: generalized-m2-Regularity conditions: support compact space} Suppose the parameter sapce of $\beta$ and $\alpha$ are
compact.
	\item \label{assu: generalized-m2- (Regularity conditions: support of data is bounded)} Suppose the support of $O$ is bounded.
	\item \label{assu: generalized-m2 - partialU is invertiable}
Suppose $\PP\{\partial_{(\beta, \alpha)} U^\text{GR-Generalized}(\beta^{\star}, \alpha^{\star}, \mu')\}$ is invertible.
	\item \label{assu: generalized-m2 - (Regularity-conditions) derivative is uniformly bounded} $\partial_{(\beta, \alpha)} U^\text{GR-Generalized}(\beta,\alpha, \mu)$ and $U^\text{GR-Generalized}(\beta,\alpha, \mu) U^\text{GR-Generalized}(\beta,\alpha, \mu)^{T}$
are uniformly bounded by integrable function. 
	\end{asulist}
\end{asu}
\begin{asu}
\label{assu: generalized-m2 - donsker class}(Donsker condition). Suppose for each t, the estimator $\hat{\mu}_t$ take values in Donsker class.
\end{asu}

\begin{lem}
	\label{lem: generalized-PPU=0}
	$\PP\{U^\text{GR-Generalized}(\beta^\star, \alpha^\star, \mu)\} = 0$ for any $\mu$.
\end{lem}

\begin{proof}
	For any $\mu$, we have
	\begin{align*}
	&\PP\{U^\text{GR-Generalized}(\beta^\star, \alpha^\star, \mu)\} \\ 
= & \sum_{t=1}^{T}\PP\Big[\PP\big\{U^\text{GR-Generalized}_{t}(\beta^\star,\alpha^\star,\mu_{t})|H_{t,}A_{t}=1\big\} p_{t}(H_t) \\ 
 &\ \ \ \ \ +\PP\big\{U^\text{GR-Generalized}_{t}(\beta^\star, \alpha^\star,\mu_{t})|H_{t,}A_{t}=0\big\} \{1-p_{t}(H_t)\} \Big] \\ 
=&  \sum_{t=1}^{T}\PP \Bigg[ \omega(t) I_{t} \Bigg\{ h^{-1} \bigg( h\big[l^{-1}\{g(S_t)^T\alpha^\star \} \big]  - f_t(S_t)^T\beta^\star \bigg) - \mu_{0t} \Bigg\}f_t(S_t) p_t(H_t) \\ 
&  \ \ \ \ \ + \omega(t)I_{t} \Bigg\{ h^{-1} \bigg( h\big[l^{-1}\{g(S_t)^T\alpha^\star \} \big]  - f_t(S_t)^T\beta^\star \bigg) - \mu_{0t} \\ 
& \quad \quad - \frac{1}{1-p_t(H_t)} \{\PP(W_{t,\Delta}Y_{t,\Delta}|H_{t},A_{t}=0) - \mu_{0t}\} \Bigg\}f_t(S_t) \{1 - p_t(H_t)\} \Bigg] \\ 
= & \sum_{t=1}^{T}\PP \Bigg[ \omega(t) I_{t} \Bigg\{ h^{-1} \bigg( h\big[l^{-1}\{g(S_t)^T\alpha^\star \} \big]  - f_t(S_t)^T\beta^\star \bigg) - \PP(W_{t,\Delta}Y_{t,\Delta}|H_{t},A_{t}=0) \Bigg\} f_t(S_t) \Bigg] \\
= &0
	\end{align*}
where last line of quality follows from 
\begin{align*}
	\PP \Bigg\{ \sum_{t = 1}^T \omega(t) \Bigg( h\big[l^{-1}\{g(S_t)^T\alpha^\star \} \big] - f_t(S_t)^T\beta^\star - h[\PP\{\PP(W_{t,\Delta}Y_{t,\Delta}|H_t, A_t = 0)\}|S_t, I_t = 1] \Bigg) \Bigg\}= 0 
\end{align*}
\end{proof}

\begin{lem}
\label{lem: generalized-m2- convergence of sup between hat mu and mu prime}  Suppose Assumptions \ref{assu: generalized-convergence of the nuisance limit}, \ref{assu: generalized-m2-Regularity conditions: support compact space}, and \ref{assu: generalized-m2- (Regularity conditions: support of data is bounded)} hold, then  
$\sup_{(\beta, \alpha)}|\PP$ $\{U^\text{GR-Generalized}(\beta, \alpha, \hat{\mu})\}$  $- \PP \{U^\text{GR-Generalized}(\beta, \alpha, \mu')\}| = o_p(1)$.
\end{lem}
\begin{proof}
Let $h_t(\mu_t) = \omega(t)I_t\big[\mu_{0t} + \frac{1 - A_t}{1 - p_t(H_t)} \{W_{t, \Delta}Y_{t, \Delta} - \mu_{0t} \} \big] $. Then we have
\begin{align*}
	&\sup_{(\beta, \alpha)}|\PP U^\text{GR-Generalized}(\beta, \alpha, \hat{\mu}) - \PP U^\text{GR-Generalized}(\beta, \alpha, \mu')| \\
	= & \PP \{ \sum_t h_t(\hat{\mu}_t) - h_t(\mu'_t) \} f_t(S_t) 
\end{align*}
where the last line is bounded by $o_p(1)$ because $\hat{\mu}_t$ converges to $\mu_t'$ (Assumption \ref{assu: generalized-convergence of the nuisance limit}) and Assumptions \ref{assu: generalized-m2-Regularity conditions: support compact space} and \ref{assu: generalized-m2- (Regularity conditions: support of data is bounded)}. This concludes the proof.
\end{proof}

\begin{lem}
\label{lem: generalized-m2- convergence of square}  Suppose Assumptions \ref{assu: generalized-convergence of the nuisance limit}, \ref{assu: generalized-m2-Regularity conditions: support compact space}, and \ref{assu: generalized-m2- (Regularity conditions: support of data is bounded)} hold, then $\lVert U^\text{GR-Generalized}$ $(\beta^\star, \alpha^\star, \hat{\mu})$ -  $U^\text{GR-Generalized}(\beta^\star, \alpha^\star, \mu')\rVert^2 = o_p(1)$.
\end{lem}
\begin{proof}
We have
\begin{align*}
 &\lVert  U^\text{GR-Generalized}(\beta^\star, \alpha^\star, \hat{\mu}) -   U^\text{GR-Generalized}(\beta^\star, \alpha^\star, \mu')\rVert^2 \\
&= 2T^2 \Big\{ \max_{1\leq t \leq T} \int |  U^\text{GR-Generalized}_t(\beta^\star, \alpha^\star, \hat{\mu}_t) -   U^\text{GR-Generalized}_t(\beta^\star, \alpha^\star, \mu_t)|^2 dP\Big\}.
\end{align*}
Thus, it suffices to show that $\int | U^\text{GR-Generalized}_t(\beta^\star, \alpha^\star, \hat{\mu}_t) -   U^\text{GR-Generalized}_t(\beta^\star, \alpha^\star, \mu_t)|^2 dP = o_p(1)$. By expanding the equation, we have
\begin{align*}
	&\int | U^\text{GR-Generalized}_t(\beta^\star, \alpha^\star, \hat{\mu}_t) -   U^\text{GR-Generalized}_t(\beta^\star, \alpha^\star, \mu_t)|^2 dP \\
	= &\int \Big| \big\{ \mu'_{0t} -\hat{\mu}_{0t} - \frac{1 - A_t}{1 - p_t(H_t)} (\mu'_{0t} - \hat{\mu}_{0t}) \big\} f_t(S_t) \Big|^2dP \\
	\leq  &C \int |\mu'_{0t} -\hat{\mu}_{0t}|^2 dP \numberthis \label{eq: generalized-m2 in the convergence of square mu1}\\
= &o_p(1) \numberthis \label{eq: generalized-m2 in the convergence of square mu2}
\end{align*}
where \eqref{eq: generalized-m2 in the convergence of square mu1} follows from Assumption \ref{assu: generalized-m2-Regularity conditions: support compact space} and \eqref{eq: generalized-m2 in the convergence of square mu2} follows from the limit of $\mu_t'$ (Assumption \ref{assu: generalized-convergence of the nuisance limit}).
\end{proof}

\begin{thm}
Suppose Suppose Assumption 1 (Consistency, Positivity, and Sequentially
ignorability) in the main paper and Assumptions \ref{assu: generalized-convergence of the nuisance limit}, \ref{assu: generalized-m2 - Unique zero}, \ref{assu: generalized-m2 - regularity condition}, \ref{assu: generalized-m2 - donsker class}  hold. Suppose there exists $\alpha^\star$ such that $\psi_t^\star = g(S_t)^T\alpha^\star$ for $\psi_t^\star$ defined in the main paper. 
Then $\hat{\beta}^\text{GR-Generalized}$ is consistent and asymptotically normal: $\sqrt{n}(\hat{\beta}^\text{GR-Generalized}-\beta^{\star})\xrightarrow{d}N(0,V^\text{GR})$
as $n\rightarrow\infty$. Furthermore, $V^\text{GR}$ can be consistently
estimated by the upper diagonal $p$ by $p$ block matrix of 
\[
\PP_{n} \big\{ \frac{\partial \Phi(\hat{\beta}, \hat{\alpha}, \hat{\mu}) }{\partial_{(\beta^T, \alpha^T)}} \big\}^{-1}\PP_{n}\{\Phi(\hat{\beta}, \hat{\alpha}, \hat{\mu}) \Phi(\hat{\beta}, \hat{\alpha}, \hat{\mu}) ^{T}\}\PP_{n}\big\{ \frac{\partial \Phi(\hat{\beta}, \hat{\alpha}, \hat{\mu}) }{\partial_{(\beta^T, \alpha^T)}} \big\}^{-1,T},
\]
where $\Phi(\beta, \alpha, \mu):=(U^\text{GR-Generalized} (\beta, \alpha, \mu), Q^\text{Generalized}(\alpha))$. 
\end{thm}

\begin{proof}
	The proof is similar to the proof in Section \ref{subsec:_texorpdfstring_asymptotic_normality_of_hat_beta} by applying Theorem 5.1 of \citep{cheng2023efficient}, thus is omitted here.
\end{proof}

\section{Additional Application Results: OR, RR, and RD Estimates}
\label{supp-sec:additional_application_results_or_rr_and_rd_estimates}
For comparison, we present the results of risk difference (RD) and risk ratio (RR) for the Drink Less Data. The proximal outcome is still near-term engagement, defined as an indicator of whether the participant opens the app in the hour following the notification (8 p.m. to 9 p.m.). Similar to the data analysis presented in the main paper, we conducted one marginal analysis ($S_t = \emptyset$) and three moderation analyses, with the moderator being the decision time index, an indicator for app use before 8 p.m. on that day, and an indicator for receiving a notification on the previous day as a proxy for treatment burden, respectively. For each analysis, we used the \texttt{wcls} function for estimating RD and \texttt{emee} functions estimating RR from the \texttt{MRTAnalysis} R package \citep{MRTAnalysis}. 

Table \ref{tab: RR and RD of application results} reports the OR, RR, and RD estimates. For the three moderation analysis, the moderated OR, RR, and RD are qualitatively in the same direction.

\begin{table}
\centering
\begin{tabular}{c|cccc}
\toprule
	Moderator & OR (SR) & OR (GR) & RR & RD \\ 
	\hline
	None & 3.89	 & 3.89	 & 3.53 & 0.09 \\ 
	\hline
	\multirow{2}{*}{Decision point} & 4.57& 4.56 & 3.90 & 0.15\\
	& 0.99 & 0.99 & 0.99 & 0.00 \\
		\hline 
	\multirow{2}{*}{Before 8 p.m.} & 4.04	& 4.02	&3.73	&0.07 \\ 
	& 1.02	&1.04	& 0.89	& 0.11 \\ 
		\hline
	\multirow{2}{*}{Habituation} & 3.91	& 3.93	& 3.59	 & 0.09 \\ 
	& 0.99	&0.98	&0.97	&0.00 \\ 
	\bottomrule
\end{tabular}
\caption{The estimates from OR, RR, and RD measures using Drink Less data. SR and GR associated with OR representst the Simple Randomization and General Randomization estimators.}
\label{tab: RR and RD of application results}
\end{table}

